\documentclass[12pt]{article}
\usepackage[linesnumbered,ruled,vlined]{algorithm2e}
\usepackage[english]{babel}
\usepackage[justification=centering]{caption}
\usepackage{wrapfig, color}
\usepackage[margin=1in]{geometry}
\usepackage{diagbox,threeparttable} 
\usepackage{natbib}
\usepackage{subcaption}
\usepackage{amsmath,amsthm, amssymb,enumitem,esvect,bbm, cases}
\usepackage{graphicx, wrapfig}
\usepackage{booktabs}
\usepackage{cancel}
\usepackage{comment, setspace, empheq, url}
\usepackage[T1]{fontenc}
\usepackage[colorlinks=true, citecolor=green, linkcolor=blue, urlcolor=blue, hypertexnames=false]{hyperref}
\usepackage{textcomp, multirow}
\usepackage{soul}
\usepackage[linesnumbered,ruled,vlined]{algorithm2e}
\usepackage{algorithmic}
\usepackage{mathtools, nccmath}
\usepackage{tocloft}
\usepackage[compact]{titlesec}  
\usepackage{setspace}
\allowdisplaybreaks
\def\thm@space@setup{%
  \thm@preskip=0pt plus 1pt minus 1pt
  \thm@postskip=\thm@preskip 
}

\newtheorem{defn}{Definition}
\newtheorem{theorem}{Theorem}
\newtheorem{proposition}[theorem]{Proposition}
\newtheorem{lemma}[theorem]{Lemma}

\newtheorem{remark}{Remark}
\newtheorem{con}[theorem]{Assumption}

\newcommand{\bs}{\boldsymbol}
\newcommand{\M}{\mathcal{M}}
\newcommand{\y}{\mathbf{y}}
\newcommand{\x}{\mathbf{x}}
\newcommand{\X}{\mathbf{X}}
\def\BibTeX{{\mathcal{M}m B\kerP-.05em{\sc i\kerP-.025em b}\kerP-.08em
    T\kerP-.1667em\lower.7ex\hbox{E}\kerP-.125emX}}

\begin{document}

\title{\LARGE{PRECISE: PRivacy-loss-Efficient and Consistent Inference based on poSterior quantilEs}}
\author{Ruyu Zhou, Fang Liu\footnote{$\;$Correspondence author: fliu2@nd.edu}\\
 \footnotesize Department of Applied and Computational Mathematics and Statistics,\\
 \footnotesize University of Notre Dame, IN 46556, USA} 
\date{}
\maketitle\vspace{-30pt}

\begin{abstract}
Differential Privacy (DP) is a mathematical framework for releasing information with formal privacy guarantees. 
While numerous DP procedures have been developed for statistical analysis and machine learning, valid statistical inference methods offering high utility under DP constraints remain limited. 
We formalize this gap by introducing the notion of valid Privacy-Preserving Interval Estimation (PPIE) and propose a new PPIE approach -- PRECISE -- to constructing privacy-preserving posterior intervals with the goal of offering a better privacy-utility tradeoff than existing DP inferential methods. PRECISE is a general-purpose and model-agnostic method that generates intervals using quantile estimates obtained from a sanitized posterior histogram with DP guarantees. We explicitly characterize the global sensitivity of the histogram formed from posterior samples for the parameter of interest, enabling its sanitization with formal DP guarantees. We also analyze the sources of error in the mean squared error (MSE) of the histogram-based private quantile estimator and prove its consistency for the true posterior quantiles as the sample size or privacy loss increases with along with its rate of convergence. 
We conduct extensive experiments to compare the utilities of PRECISE with common existing privacy-preserving inferential approaches across a wide range of inferential tasks, data types and sizes, DP types, and privacy loss levels. The results demonstrated a significant advantage of PRECISE with its nominal coverage and substantially narrower intervals than the existing methods, which are prone to either under-coverage or impractically wide intervals. 

\vspace{6pt}
\noindent \textbf{keywords}:  Bayesian, differential privacy, MSE consistency, privacy-preserving interval estimation (PPIE),  privacy loss, quantile.
\end{abstract} 

\maketitle

\setlength{\abovedisplayskip}{6pt}
\setlength{\belowdisplayskip}{6pt}

\setstretch{1.02}
\section{Introduction}\label{sec:intro}
\subsection{Background}
The unprecedented availability of data containing sensitive information has heightened concerns about the potential privacy risks associated with the direct release of such data and the outputs of statistical analyses and machine learning tasks. Providing a rigorous framework for privacy guarantees, Differential Privacy (DP) has been widely adopted for performing privacy-preserving analysis since its debut in 2006 \citep{dwork2006calibrating, dwork2006our} and gained enormous popularity among privacy researchers and in practice (e.g., Apple \citep{apple}, Google \citep{erlingsson2014rappor}, the U.S. Census \citep{abowd2018us}). 
Many DP procedures have been developed for various statistical problems, including sample statistics (e.g., mean \citep{smith2011privacy}, median \citep{dwork2009differential}, variance or covariance  \citep{amin2019differentially, biswas2020coinpress}), linear regression \citep{alabi2020differentially, wang2018revisiting}, empirical risk minimization (ERM) \citep{chaudhuri2011differentially}, and so on. 

Most existing DP methods to date focus on releasing privatized or sanitized statistics without uncertainty quantification, limiting their usefulness for robust decision-making. Though there exists work on DP statistical inference, including both hypothesis testing and interval estimation, this line of research is still in its early stages and is largely focused on relatively simple inference tasks, such as DP $\chi^2$-test \citep{gaboardi2016differentially}, uniformly most powerful tests for Bernoulli data \citep{awan2018differentially}, $F$-test in linear regression \citep{alabi2022hypothesis}, and privacy-preserving interval estimation for Gaussian means or regression coefficients. Our work contributes to this field by introducing a new procedure for valid \emph{privacy-preserving interval estimation (PPIE)}. 


\subsection{Related work}\label{sec:related_work}
Most of the existing works on PPIE can be loosely grouped into two broad categories. The first group obtains PPIE through the derivation of the asymptotic distribution of privacy-preserving (PP) estimator, where either asymptotic Gaussian distributions (e.g., inferring univariate Gaussian mean  \citep{du2020differentially, evans2023statistically, d2015differential, karwa2017finite}, multivariate sub-Gaussian mean \citep{biswas2020coinpress}, proportion \citep{lin2024differentially}, and complicated problems like M-estimators \citep{avella2023differentially} and ERM \citep{Wang_Kifer_Lee_2019}), or 
asymptotic $t$-distributions  (e.g., inferring linear regression coefficient
\citep{sheffet2017differentially} and the general-purpose multiple sanitization (MS) procedure \citep{liu2016model})
are assumed. 
The second group employs a quantile-based approach. The frequentist methods in this category primarily rely on the bootstrap technique to build PPIE, such as the simulation approach \citep{du2020differentially} and the parametric bootstrap method  \citep{ferrando2022parametric}. The BLBquant method \citep{chadha2024resampling} and the GVDP (General Valid DP) method \citep{covington2021unbiased} employ the Bag of Little Bootstraps (BLB) technique \citep{kleiner2014scalable} to obtain private quantiles. BLBquant provides quantitative error bounds and outperforms GVDP empirically.  \citep{wang2022differentially} leverages deconvolution (a technique that deals with contaminated data) to analyze DP bootstrap estimates and obtain PPIE. In the Bayesian framework, the existing methods focus on incorporating DP noise in PP posterior inference and computation, such as PP regression coefficient estimation through sufficient statistics perturbation \citep{bernstein2019differentially, kulkarni2021differentially} and data augmentation MCMC sampler \citep{ju2022data}.
Outside these two categories, other PPIE approaches include non-parametric methods for population medians \citep{drechsler2022nonparametric}, synthetic data-based methods \citep{bojkovic2024differentially, raisa2023noise,liu2016model}, and simulation-based methods \citep{awan2023simulation}, among others.

While research on PPIE has been growing, limitations remain in current techniques from methodological, computational, and application perspectives. \emph{First}, most methods are designed for specific basic inferential tasks (e.g., Gaussian means), creating a need for more general PPIE procedures that can accommodate a wide range of statistical inference problems. \emph{Second}, some existing methods are compatible only with certain types of DP guarantees (e.g., $\varepsilon$-DP), limiting their applicability. \emph{Third}, even for basic PPIE tasks, there is considerable room to improve existing methods in achieving an optimal trade-off between privacy and utility, particularly in practically meaningful privacy loss settings with better computational efficiency, as some existing sampling-based methods tend to be computationally intensive. \emph{Fourth}, there is a lack of comprehensive comparisons on the practical feasibility and utility of existing PP inferential methods across common statistical problems, which is essential for guiding their practical applications.
\subsection{Our work and contributions}
In this work, we address several research and application gaps in PPIE identified in Section \ref{sec:related_work}. 
Our main contributions are summarized below. 
\begin{itemize}[leftmargin=24pt, itemsep=0pt] 
    \item We introduce a formal definition of valid PPIE and propose PRECISE, a new approach for PPIE via consistent estimation of PP posterior quantiles. PRECISE is problem- and model-agnostic and broadly applicable whenever Bayesian posterior samples can be obtained. It is robust to user-specified global bounds on data or parameter space -- a persistent challenge in preserving utility for PP inference. This results in a superior privacy-utility tradeoff compared to existing PPIE methods, which are often highly sensitive to global bounds specifications.
    \item We define global sensitivity (GS) of the posterior density for the parameter of interest to be used with a proper randomized mechanism to achieve formal DP guarantees. We further introduce a convenient analytically approximate GS variant when the sample data size is large, as well as an upper bound for the GS to facilitate practical implementation. 
    \item We theoretically analyze the Mean-Squared-Error (MSE) consistency of the private quantiles obtained by PRECISE toward their non-private posterior quantiles, and derive the convergence rate in sample size and privacy loss parameter.  
    \item Our empirical results in various experimental settings show that PRECISE achieves nominal coverage with significantly narrower intervals, whereas other PPIE methods may either under-cover or produce unacceptably wide intervals at low privacy loss, providing strong evidence on the superior performance by PRECISE in the privacy-utility trade-offs.
    \item We also propose an exponential-mechanism-based PP posterior quantile estimator and examine its theoretical properties and practical limitations, offering further context for the advantages of PRECISE.
\end{itemize}

\section{Preliminaries}\label{sec:prelim}
In this section, we provide a brief overview of the basic concepts in differential privacy (DP). For the definitions below, we refer two datasets $\mathbf{x}$ and $\mathbf{x}'$ as neighbors (denoted by $d(\mathbf{x}, \mathbf{x}')=1$) if $\mathbf{x}$ differs from $\mathbf{x}'$ by exactly one record by either removal or substitution.
\begin{defn}[$(\varepsilon,\delta)$-DP \citep{dwork2006our, dwork2006calibrating}]\label{def:epsdelta}
    A randomized algorithm $\mathcal{M}$ is of $(\varepsilon, \delta)$-DP if for all pairs of neighboring datasets $(\mathbf{x}, \mathbf{x}')$ and for any subset $\mathcal{S} \subset \mbox{Image}(\mathcal{M})$,
    \begin{equation}
        \Pr(\mathcal{M}(\mathbf{x})\in \mathcal{S}) \leq e^{\varepsilon}\cdot\Pr(\mathcal{M}(\mathbf{x}')\in \mathcal{S}) + \delta.
    \end{equation}
\end{defn}
$\varepsilon \!>\! 0$ and $\delta \!\geq\! 0$ are privacy loss parameters. When $\delta \!=\! 0$, $(\varepsilon, \delta)$-DP reduces to pure $\varepsilon$-DP. Smaller values of $\varepsilon$ and $\delta$ imply stronger privacy guarantees for the individuals in a dataset as the outputs based on $\mathbf{x}$ and $\mathbf{x}'$ are more similar. 

\textbf{Laplace mechanism} \citep{dwork2006calibrating} is a widely-used mechanism to achieve $\varepsilon$-DP. Let $\boldsymbol{s} = (s_1, \ldots, s_r)$ denote the statistics calculated from a dataset; and its sanitized version via the Laplace mechanism is $\mathbf{s}^*= \mathbf{s} + \mathbf{e}$, where $\mathbf{e}=\{e_j\}_{j=1}^r$ with $e_j\sim \mbox{Laplace}(0, \Delta_1/\varepsilon)$, and $\Delta_1 \!=\!\max_{\mathbf{x}, \mathbf{x}', d(\mathbf{x}, \mathbf{x}')=1}\!||\boldsymbol{s}(\mathbf{x}) - \boldsymbol{s}(\mathbf{x}')||_1$ is the \textit{$\ell_1$ global sensitivity} of $\boldsymbol{s}$, representing the maximum change in $\boldsymbol{s}$ between two neighboring datasets in $\ell_1$ norm. Higher sensitivity requires more noise to achieve the pre-set privacy guarantee.

\textbf{Exponential mechanism} \citep{mcsherry2007mechanism} is a general mechanism of $\varepsilon$-DP and releases sanitized $s^*$ with probability $\propto\exp(\varepsilon\cdot u(s^*|\mathbf{x})/2\Delta_u)$, where $u$ is a utility function that assigns a score to every possible output $s^*$ and $\Delta_u$ is the $\ell_1$ global sensitivity of $u$.

\begin{defn}[$\mu$-GDP \citep{dong2022gaussian}]\label{def:mu}     Let $\mathcal{M}$ be a randomized algorithm and $\mathcal{S}$ be any subset of $\mbox{Image}(\mathcal{M})$. Consider the hypothesis test $H_0\!: S\!\sim\! \mathcal{M}(\mathbf{x})$ versus $H_1\!: S\!\sim\! \mathcal{M}(\mathbf{x}')$, where $d(\mathbf{x}, \mathbf{x}')=1$. $\mathcal{M}$ is of $\mu$-Gaussian DP if it satisfies
    \begin{equation}
    T(\mathcal{M}(\mathbf{x}), \mathcal{M}(\mathbf{x}'))(\alpha)\geq \Phi(\Phi^{-1}(1-\alpha)-\mu), 
    \end{equation}
    where $T(\cdot, \cdot)(\alpha)$ is the minimum type II error among all such tests at significance level $\alpha$ and $\Phi(\cdot)$ is the CDF of the standard normal distribution. 
\end{defn}

In less technical terms, Definition \ref{def:mu} states that $\mathcal{M}$ is of $\mu$-GDP if distinguishing any two neighboring datasets given the information sanitized via $\mathcal{M}$  is at least as difficult as distinguishing $\mathcal{N}(0,1)$ and $\mathcal{N}(\mu, 1)$. $(\varepsilon,\delta)$-GDP relates to $\mu$-GDP, with one $\mu$ corresponding to infinite pairs of  $(\varepsilon,\delta)$. 

\begin{lemma}[Conversion between $(\varepsilon, \delta)$-DP and $\mu$-GDP \citep{dong2022gaussian}]
\label{lemma:mu_epsdelta} 
A mechanism is of $\mu$-GDP if and only if it is of $(\varepsilon, \delta(\varepsilon))$-DP for all $\varepsilon \geq 0$, where $\delta(\varepsilon) = \Phi(-\varepsilon/\mu + \mu/2) - e^{\varepsilon}\Phi(-\varepsilon/\mu - \mu/2)$.
\end{lemma}

\textbf{Gaussian mechanism} can be used achieve both $(\varepsilon,\delta)$-GDP and  $\mu$-GDP. In this work, we use the Gaussian mechanism of  $\mu$-GDP. Specifically, sanitized $\boldsymbol{s}^*\!=\! \boldsymbol{s}+\boldsymbol{e}$, where  $\mathbf{e}=\{e_j\}_{j=1}^r$ and $e_j \!\sim\! \mathcal{N}(0, \Delta_2^2/\mu^2)$, and $\Delta_2\!=\!\max_{\mathbf{x}, \mathbf{x}', d(\mathbf{x}, \mathbf{x}')=1}\!||\boldsymbol{s}(\mathbf{x}) \!-\! \boldsymbol{s}(\mathbf{x}')||_2$ is the $\ell_2$ global sensitivity of $\boldsymbol{s}$, the maximum change in $\boldsymbol{s}$ between two neighboring datasets in $\ell_2$ norm.

DP and many of its variants, $\mu$-GDP included, have appealing properties for both research and practical applications. For example, they are \textit{immune to post-processing}; that is, any further processing on the differentially private output without accessing the original data maintains the same privacy guarantees. Furthermore, the \textit{privacy loss composition} property of DP tracks overall privacy loss from repeatedly accessing and releasing information from a dataset. The basic composition principle states that if $\mathcal{M}_1$ is of $(\varepsilon_1,\delta_1)$-DP (or $\mu_1$-GDP) and $\mathcal{M}_2$ is of  $(\varepsilon_2,\delta_2)$-DP (or $\mu_2$-GDP), then $\mathcal{M}_1 \!\circ\! \mathcal{M}_2$ is of $(\varepsilon_1 \!+\! \varepsilon_1, \delta_1\!+\!\delta_2)$-DP (or $\sqrt{\mu_1^2 \!+\! \mu_2^2}$-GDP) if  $\mathcal{M}_1$ and  $\mathcal{M}_2$ operate on the same dataset. The privacy loss composition bound in $\mu$-GDP is tighter than that of the $(\varepsilon,\delta)$-DP.

\section{Privacy-Preserving Interval Estimation (PPIE)}\label{sec:PPIE}
Before introducing PRECISE as a PPIE procedure, we provide a formal definition of PPIE in Definition \ref{defn:CI}. The definition applies to PP intervals constructed in both the Bayesian and frequentist frameworks. 

\begin{defn}[privacy-preserving interval estimate (PPIE)]\label{defn:CI}
Let $\mathbf{x}$ be a sensitive dataset of $n$ records that is a random sample from the probability distribution $f(\x|\boldsymbol{\theta}_0)$ with unknown parameters $\boldsymbol{\theta}_0$.
Denote the non-private interval estimator for $\bs\theta_0$ at confidence level $1-\alpha$ by $(l(\x), u(\x))$ and the DP mechanism by $\M$. The PPIE at privacy loss $\bs\eta$ for $\boldsymbol{\theta}_0$,  denoted by interval ($\M(l(\mathbf{x})), \M(u(\mathbf{x}))$, satisfies 
\begin{equation}\label{eqn:CI}
    \textstyle{\Pr_{\M,\mathbf{x}}}(\M(l(\mathbf{x}))<\boldsymbol{\theta}_0<\M(u(\mathbf{x})))\ge 1-\alpha \; \mbox{ for every } \boldsymbol{\theta}_0,\boldsymbol{\eta}.
\end{equation}
\end{defn}\vspace{-6pt}
Definition \ref{defn:CI} is not exact but conservative --  that is, instead of requiring $ \Pr(\M(l(\mathbf{x}))\!<\!\boldsymbol{\theta}_0\!<\!\M(u(\mathbf{x})))\!=\!1\!-\!\alpha$, it requires the probability $\ge1-\alpha$ as intervals with exact coverage may be difficult to construct, especially dealing with discrete distributions. On the other and, the construction of an interval should aim to keep the width  $|\M(u(\mathbf{x}))-\M(l(\mathbf{x}))|$ as small as possible while satisfying Definition \ref{defn:CI}; otherwise, the PPIE would be conservative and meaningless.\footnote{For completeness, Definition \ref{defn:CI} can be extended to scenarios where there exist other parameters that are not of immediate inferential interest. That is, $\mathbf{x}\sim f(X|\boldsymbol{\theta}_0,\boldsymbol{\beta}_0)$ and Eq.~\eqref{eqn:CI} holds for every $\boldsymbol{\theta}_0,\boldsymbol{\beta}_0$ and $\bs\eta$.}
The DP mechanism $\M$ in Definition \ref{defn:CI} can be of  $(\varepsilon,\delta)$-DP ($\bs\eta\!=\!(\varepsilon,\delta)$) or any of its variants, such as $\mu$-GDP ($\bs\eta\!=\!\mu$). 

\subsection{Overview of the PRECISE procedure}\label{sec:RAP}\vspace{-6pt}
Let $ \{\x_i\}_{i=1}^n\overset{\text{i.i.d}}{\sim} f(\x|\bs\theta_0)$ denote a dataset containing data points from $n$ individuals whose privacy are to be protected; $\x_i \in \mathbb{R}^q$ and $\bs\theta_0 = (\theta^{(1)}_0, \theta^{(2)}_0, \ldots, \theta^{(p)}_0)^{\top} \in \bs\Theta$ represents the $p$-dimensional true parameter vector. The PRECISE  procedure constructs a pointwise interval estimation for $\bs\theta_0$ in a Bayesian framework, with the steps outlined as follows. 

First, it draws $m$ posterior samples $\{\theta_j^{(k)}\}_{j=1}^m$ 
for $k=1,\ldots,p$ and constructs a histogram $H^{(k)}$ based on these $m$ samples.  
Second, it perturbs the bin counts in $H^{(k)}$ using a DP mechanism $\mathcal{M}$ at a pre-specified privacy loss $\bs\eta$, resulting in a \emph{Privacy-Preserving Posterior} (P$^3$) histogram $H^{(k)*}$. \emph{Third}, at a specified confidence coefficient $1-\alpha\! \in\!(0,1)$ for PPIE, identify two bins in $H^{(k)*}$, the cumulative probability up to which is the closest to $\frac{\alpha}{2}$ and $1-\frac{\alpha}{2}$ respectively. Finally, release a random sample from each of these two identified intervals as the PP estimate of the population posterior quantiles $F^{-1}_{\theta^{(k)}|\{\x_i\}_{i=1}^n}\left(\frac{\alpha}{2}\right)$ and $F^{-1}_{\theta^{(k)}|\{\x_i\}_{i=1}^n}\left(1-\frac{\alpha}{2}\right)$. 

A key step in the PRECISE procedure is the construction of the  P$^3$  posterior histogram for the parameter of interest. To achieve this, we will first determine the sensitivity of the histogram given a certain number of posterior samples and then design a proper randomized mechanism to ensure its DP guarantees. 

\subsection{Global sensitivity of posterior histogram}\label{sec:GS}
It is important to note that \emph{sanitizing a histogram constructed from a set of posterior samples of $\theta^{(k)}$ given sensitive data $\x$ is fundamentally different and more complex than sanitizing a histogram $H(\x)$ of the sensitive data $\x$ itself}. Specifically, the DP definition pertains to changing one record in the sensitive dataset $\x$. Removing a record from $\x$ only affects one bin in $H(\x)$ and thus the global sensitivity of $H(\x)$, represented in the count, is 1 if the neighboring relation is removal and 2 if the neighboring relation is substitution. 
In contrast, our goal is to sanitize the histogram of a parameter $H(\theta|\x)$ given a set of posterior samples from $f(\theta|\x)$. Changing one record in $\x$ will alter the whole posterior distribution from $f(\theta|\x)$ to  $f(\theta|\x')$ and the influence is indirect and more complex compared to how it affects $H(\x)$, eventually complicating the calculation of the sensitivity of $H(\theta|\x)$. Figure  \ref{fig:toy} illustrates how the sensitivities of $H(\x)$ and $H(\theta|\x)$ differ using a toy example.
\begin{figure}[!h]
    \centering \vspace{-2.5in}
    \begin{picture}(400,400) 
        \put(-30,0){\includegraphics[width=0.6\textwidth]{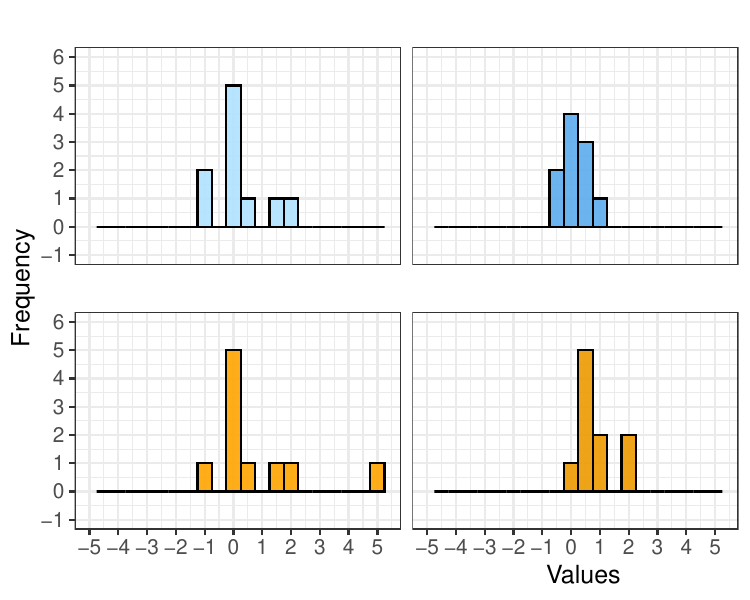}}
        \put(258,115){\includegraphics[width=0.288\textwidth, trim={0.25in 0.09in 0in, 0in},clip]{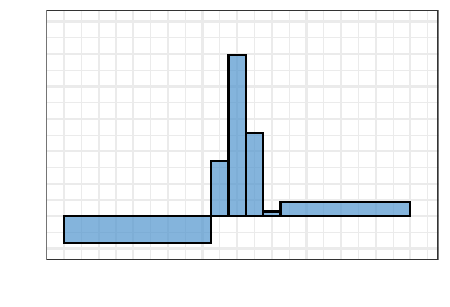}}
        \put(258,6.5){\includegraphics[width=0.286\textwidth, trim={0.25in 0.1in 0in, 0.31in},clip]{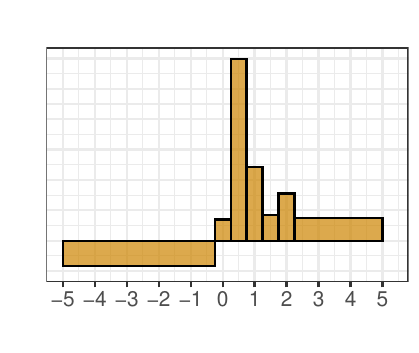}}
        \put(50, 212){\footnotesize{$H(\x)$}}
        \put(130, 212){\footnotesize{$H(\theta|x): \{\theta_j\}_{j=1}^{m=10} \sim f(\theta|\x)$}}
        \put(310, 211){\footnotesize{$H^*(\theta|\x)$}}
        \put(50, 113){\footnotesize{$H(\x')$}}
        \put(130, 113){\footnotesize{$H(\theta|\x'): \{\theta'_j\}_{j=1}^{m=10} \sim f(\theta|\x')$}}
        \put(310, 112){\footnotesize{$H^*(\theta|\x')$}}
    \end{picture}
    \captionsetup{justification=justified, singlelinecheck=false}\vspace{-9pt}
    \caption{A toy example to illustrate the difference in how alternating one individual in dataset $\x$ affects the histogram of data $\x$ (first column) and the histogram of posterior samples drawn from $f(\theta|\x)$ (second column). $\x=\{x_i\}_{i=1}^{10}\sim\mathcal{N}(0,1)$;  $L_{\x} = L = -5$, $U_{\x}=U=5$; the neighboring dataset $\x'$ is constructed by substituting the min($\x$) with $U_{\x}$.  P$^3$ histogram (third column) is obtained via the Laplace mechanism at $\varepsilon=1$.}
    \label{fig:toy}\vspace{-12pt}
\end{figure}
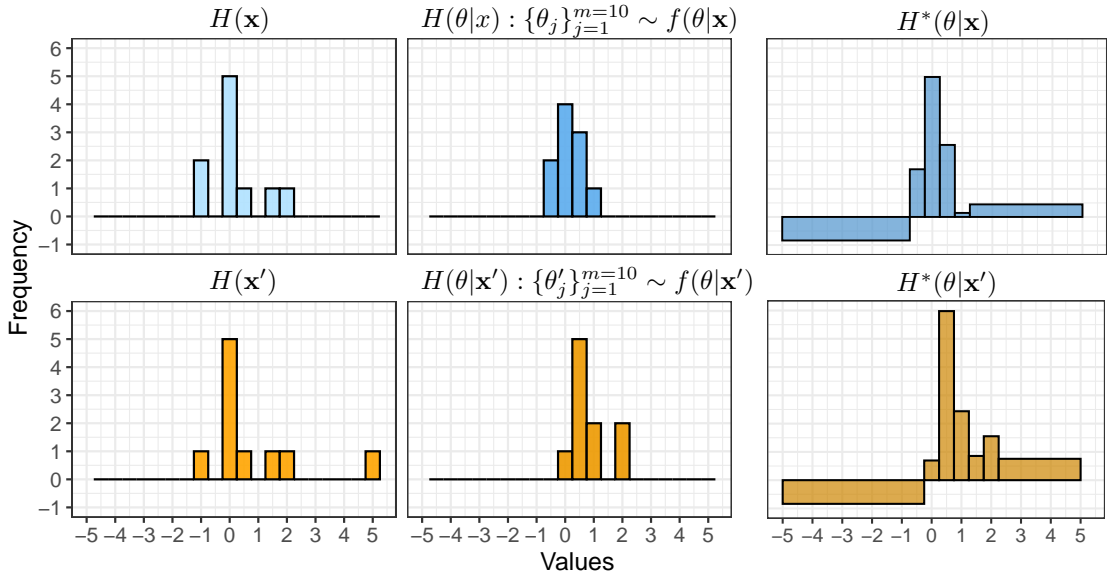

\begin{defn}[Global sensitivity of posterior distribution  $G(n)$]\label{def:Gn} 
For a scalar parameter $\theta$, the global sensitivity (GS) of its posterior distribution given data $\x$ of size $n$ is
\begin{align}\label{Eqn:Gn}
G(n) \triangleq\sup_{\theta\in \Theta, d(\x,\x')=1}|f(\theta|\x) \!-\! f(\theta|\x')| = \sup_{d(\x,\x')=1} \mbox{TVD}(f(\theta|\x),f(\theta|\x')),
\end{align}
where TVD stands for total variation distance.
\end{defn}

Given that the ``statistic'' in this case is a probability distribution/mass function, other divergence or distance measures between distributions can also be used to define GS, such as KL divergence, Hellinger distance, Wasserstein distance, etc, in addition to TVD. For example, the KL relates to TVD by Pinsker’s inequality and $2G(n)$ in Eq.~\eqref{Eqn:Gn} serves as a lower bound to the GS defined with the KL divergence; for Hellinger distance, given its relationship with TVD, $\sqrt{G(n)}$ upper bounds the Hellinger-distance-based GS. 
Compared to these potential alternative GS definitions, the TVD-based GS in  Definition \ref{def:Gn} offers several advantages. First, it does not involve integrals like the other metrics (KL, Hellinger,  Wasserstein), the computation of which poses significant analytical difficulties and often lacks closed-form solutions. Second, in cases where the posterior distribution does not have a closed-form expression and is instead approximated based on Monte Carlo posterior samples of $\theta$, the TVD-based definition naturally aligns with the GS based on the histogram, though every bin count can be affected, as illustrated in Figure \ref{fig:toy}. 

Definition \ref{def:Gn} suggests $G(n)$ is a function of sample size $n$, which is assumed to be fixed and known, especially for the sake of statistical inference.  WLOS, we examine the substitution neighboring relation, where $\x$ and $\x'$ are both of sample size $n$. Direct evaluation of $G(n)$ requires comparing the posterior densities constructed from all possible neighboring datasets of size $n$, an impossible task unless $f(\theta|\x)$ has a discrete, finite domain. A more practical alternative is to derive an analytical approximation or an upper bound for $G(n)$ that can be conveniently calculated in practical implementations. 
Theorem \ref{thm:G} provides such an approximation for when $n$ is large.

\begin{theorem}[Analytical approximation $G_0$ to $G(n)$]\label{thm:G} 
Assume that prior $f(\theta)$ is non-informative relative to the amount of data, let $\widehat{\theta}_n$ and $\widehat{\theta}_n'$ be the maximum a posteriori (MAP) estimates evaluated on two neighboring datasets $\x$ and $\x'$ with substitution relation, respectively.  If $\widehat{\theta}'_n - \widehat{\theta}_n \approx \mathcal{O}(n^{-1})$, 
then
\begin{align}
    G(n)  \approx \left(\frac{CI_{\theta_0}}{\sqrt{2\pi}}\right)e^{-\frac{1}{2}+\mathcal{O}(n^{-\frac{1}{2}})}+ \mathcal{O}(n^{-\frac{1}{2}})  \Longrightarrow G_0  = \frac{C\cdot I_{\theta_0}}{\sqrt{2e\pi}}\mbox{ as } n\rightarrow \infty,\label{eqn:G_uni}
\end{align}
where $C$ is a constant and $I_{\theta_0}$ is the Fisher information at the true parameter $\theta_0$ given a single data point $x$. 
\end{theorem}

The detailed proof of Theorem \ref{thm:G} is provided in Appendix \ref{ape:G_proof_uni}\footnote{Theorem \ref{thm:G} can be extended to the multidimensional case $\bs\theta \!=\! (\theta^{(1)}, \ldots, \theta^{(p)})^{\top} \!\in \bs\Theta$ for $p>1$. We show that $G(n) \asymp n^{\frac{p-1}{2}}$ (see Appendix \ref{ape:G_proof_multi} for the proof), implying that $G(n)$ does not converge to a constant as $n\rightarrow \infty$ when $p > 1$. We focus on one-dimensional parameter $\theta$ in this work.}. In brief, based on the Bernstein-von Mises theorem, we approximate the two posteriors given the two neighboring datasets as Gaussian for a large $n$ and then apply the third-order Taylor expansion to approximate their difference and leverage the symmetry of the neighboring datasets to identify the maximizer of the absolute difference. 

Eq.~\eqref{eqn:G_uni} suggests that $G(n)$ converges to $G_0$ at a rate of $\mathcal{O}(n^{-1/2})$, implying that $G(n)$ can be approximated in practice by its limiting constant $G_0=CI_{\theta_0}/\sqrt{2e\pi}$ for a sufficiently large $n$. 
The constant $C$ and Fisher information $I_{\theta_0}$ are parameter- or model-dependent.  For example, consider the mean $\mu$ and variance $\sigma^2$ of a Gaussian likelihood. Assume a non-informative prior $p(\mu,\sigma^2)\propto (\sigma^2)^{-3/2}$, their MAP estimates are the sample mean and variance of data $\x$, respectively. Suppose  $\x=\{x_i\}_{i=1}^n$ and $\x'$ differ from $\x$ in the last observation WLOG, that is, $x_n$ is replaced by $x'_n$, then
\begin{align*}
C & \triangleq |x'_n-x_n|  \mbox { for }\hat{\mu} = \bar{x}; \\
C  & \triangleq |(x_n - \bar{x}_{n-1})^2 - (x_n' - \bar{x}_{n-1})^2|  \mbox { for } \hat{\sigma}^{2} = \textstyle n^{-1}\sum_{i=1}^n(x_i-\bar{x})^2. 
\end{align*}
To calculate the constant $C$ above, global bounds $(L_\x, U_\x)$ on $\x$ will need to be specified. In terms of $I_{\theta_0}$, it is $\sigma_0^{-2}$ for $\mu_0$ and  $\sigma_0^{-4}/2$ for $\sigma_0^2$, both of which involve the unknown $\sigma_0^{-2}$. To calculate $I_{\theta_0}$,  a lower global bound for $\sigma^2_0$ can be assumed, or an estimate of $\sigma_0^2$ can be plugged. In the latter case, the estimate needs to be sanitized before being plugged, incurring additional privacy costs. In many cases, $I_0$ may not have an analytically closed form like in this simple example, especially for uncommon likelihoods, high-dimensional data, and data with complex dependency structures. In such cases, numerical methods can be used.  

For small $n$, we recommend using an upper bound $\overline{G_0}$ to ensure DP guarantees. Despite its potential conservativeness, $\overline{G_0}$ may be a more preferable and practical choice than $G_0$ even when $n$ is large, as the analytical approximation of $G_0$ can be tedious if not challenging, as demonstrated above even for the simple Gaussian case.
Section \ref{sec:G0Gn} compares numerical approximation of $G(n)$,  analytical approximation $G_0$ from
Theorem~\ref{thm:G}, and upper bound $\overline{G_0}$ in some specific examples. 
 
After $G(n)$ is calculated numerically or approximated analytically or upper-bounded, we may proceed with computing the GS of the posterior distribution of a parameter. Because Bayesian interval estimation is typically obtained using posterior samples since many practical problems lack closed-form posteriors, even when the $f(\theta|\x)$ is available in closed form, we compute the GS of the histogram constructed from posterior samples of $\theta$ rather than for $f(\theta|\x)$ per se. The result is stated in Theorem \ref{thm:H}. 
\begin{theorem}[GS of posterior histogram]\label{thm:H}
Let $n$ be the sample size of data $\x$, $m$ be the number of posterior samples on parameter $\theta$ from $f(\theta|\x)$, and $H$ be the histogram based on the $m$ posterior samples with bin width $h$. The GS of $H$ is 
\begin{equation}\label{eqn:DeltaH}
   \Delta_{H}=2mhG(n). 
\end{equation}
\end{theorem}
The detailed proof of Theorem \ref{thm:H} is provided in Appendix \ref{ape:DPproof}. Briefly, the main proof idea is to 
upper bound the TVD between two discretized distributions (the posterior histograms given two neighboring datasets) using the mean value theorem for integrals. 

Users can pre-specify $m$. Eq.~\eqref{eqn:DeltaH} suggests that $\Delta_{H}$ increases linearly with $m$, which makes sense as releasing more posterior samples implies more information in the data $\x$ is also leaked, requiring a larger scale parameter of the randomized mechanism to ensure DP at a preset privacy loss.  A more user-friendly usage of Eq.~\eqref{eqn:DeltaH} is to fix $\Delta_{H}$ at a constant -- a convenient choice would be $1$ -- and back-calculate $m$. That is, 
\begin{equation}\label{eqn:DeltaH1}
   \Delta_{H}=1=2mhG(n) \quad\Longrightarrow\quad  m=\frac{1}{2hG(n)}. 
\end{equation}
Given the GS of $H$, one can sanitize the bin counts in $H$ to obtain a Privacy-preserving posterior (P$^3$) histogram $H^*$ via a proper randomized mechanism. Denote the vector of $B$ bin counts by $\mathbf{c}=\{c_1,\ldots,c_B\}$, where $\sum^B_{b=1} c_b=m$; then
\begin{align}
    H^* &= \M(H) =\mathbf{c}+\mathbf{e}, \mbox{ where } \mathbf{e}=\{e_1,\ldots,e_B\}\mbox{ and } e_j\mbox{ for } j=1,\ldots, B\sim\notag\\
    &\begin{cases}
        \mbox{Laplace$(0,\Delta_H/\epsilon)$ if the Laplace mechanism of $\epsilon$-DP is used}\\
    \mbox{$\mathcal{N}(0,\Delta^2_{H}/\mu^2)$ if the Gaussian mechanism of $\mu$-GDP is used.}
    \end{cases}\label{eqn:M}
\end{align}


\subsection{Construction of P$^3$ Histogram with DP guarantees}
Algorithm \ref{alg:P3_hist} lists the steps for constructing a univariate P$^3$ histogram. Based on some mild regularity conditions listed in Assumption \ref{con:DP}, which are readily satisfied as long as $h$ is not too small,  Alg. \ref{alg:P3_hist} adheres to DP guarantees (Theorem \ref{thm:DP}).
\begin{algorithm}[h]
\caption{Construction of P$^3$ Histogram}\label{alg:P3_hist}
\SetAlgoLined
\SetKwInOut{Input}{input}
\SetKwInOut{Output}{output}
\Input{posterior distribution $f(\theta|\x)$  (up to a constant), global bounds $(L,U)$ for $\theta$, bin width $h$, DP mechanism $\M$,  privacy loss $\bs\eta$,  P$^3$ histogram version, $G(n)$,
collapsing thresholds $(\tau_L,\tau_U)$}
\Output{P$^3$ histogram $H^*$.}
Calculate the number of posterior sample $m$ (Eq.~\eqref{eqn:DeltaH1}) \label{step:m}\;
Draw posterior samples $\{\theta_j\}_{j=1}^m \overset{\text{iid}}{\sim} f(\theta|\x)$\label{step:sample}\;
Form a histogram $H$ with bin width $h$ (number of bins $B\!=\!(U-L)/h$) based on the $m$ samples. Denote the bins by 
$\Lambda_b \!=\! [L + (b-1)h, \;L + b\cdot h)$ and the bin counts by  $c_b = \textstyle{\sum_{j=1}^m }\mathbbm{1}(\theta_j \!\in\! \Lambda_b) \text{ for }b=1,\ldots,B$ \label{step:hist}\;
Set $b_L \leftarrow \arg \min_{b\in \{1, \ldots, B\}}\{c_b>\tau_L\}$ and $b_U \leftarrow \arg \max_{b\in \{1, \ldots, B\}}\{c_b>\tau_U\}$, where $(\tau_L,\tau_U)$ are non-negative integers; or $b_L \leftarrow \arg \min_{b\in \{1, \ldots, B\}}\{b/B\ge\tau_L\}$ and $b_U \leftarrow \arg \max_{b\in \{1, \ldots, B\}}\{b/B\ge\tau_U\}$, where $(\tau_L,\tau_U)\in[0,1]$ are small constants\; \label{step:identify} 
Set $\Lambda_0 \leftarrow \cup_{b< b_L}\{\Lambda_b\}$ and $\Lambda_{B'+1}\leftarrow\cup_{b> b_U}\{\Lambda_b\}$, where $B'= b_U-b_L+1$\; 
Re-index uncollapsed bins using indices $1$ to $B'$\;\label{step:reindex}
\For{$b = 0, \ldots, B'\!+\!1$}{\label{step:sanitization}
    \textbf{if} P$^3$ histogram version $==``+"$, \textbf{then} $c^*_b \leftarrow \max\{0, \M(c_b,\bs\eta)\}$ (Eq.~\eqref{eqn:M}) \label{step:DP}\;
    \textbf{if}  P$^3$ histogram version $==``-"$, \textbf{then} $c^*_b \leftarrow \M(c_b,\bs\eta)$  (Eq.~\eqref{eqn:M}); \label{step:DP*} 
} \label{step:endsanitization}
Return $H^*$ with sanitized bin counts $\{c^*_b\}_{b=0}^{B'+1}$ for bins $\{\Lambda_b\}_{b=0}^{B'+1}$.
\end{algorithm}

\vspace{-6pt}\begin{con}\label{con:DP}
Let $F^{-1}_{\theta|\x}(q)=\inf \{\theta: F(\theta|\x)\geq q\}$ for $0<q<1$. Assume
\begin{itemize}
    \item [(a)] $f(\theta|\x)\ge0$ in its support $\Theta$ and the corresponding cumulative distribution function (CDF) $F(\theta|\x)$ is continuous on any closed interval $\Lambda_b$ for $b\in\{1,\ldots,B\}$.  
    \item [(b)] $\sum_{b=1}^{b_L}f(\xi_b|\x) \le G(n)$, where $\xi_b\in \Lambda_b$ for $b\in\{1,\ldots,b_L\}$, and $\sum_{b=b_U}^B f(\xi_b|\x)\le G(n)$, where $\xi_b\in \Lambda_b$ for $b\in\{b_U,\ldots,B\}$.
\end{itemize}
\end{con}
\begin{theorem}[DP Guarantee of P$^3$ Histogram]\label{thm:DP}
Under Assumption~\ref{con:DP}, the P$^3$ histogram output by Algorithm \ref{alg:P3_hist} satisfies $\bs\eta$-DP when $m =(2hG(n))^{-1}$, where $G(n)$ is as defined in Eq.~\eqref{Eqn:Gn}.
\end{theorem}

Per the discussion in Section \ref{sec:GS} regarding the calculation of $G(n)$, one may approximate $G(n)$ numerically, replace it with the analytical approximate $G_0$ in Eq.~\eqref{eqn:G_uni} when $n$ is large or with its upper bound $\overline{G_0}$ for a conservative $G(n)$ regardless of $n$. 

\subsubsection{Two versions of P$^3$ histogram}We provide two versions of P$^3$ histogram in Algorithm \ref{alg:P3_hist}, depending on whether non-negativity correction is applied to the bin counts of the sanitized posterior histogram ($+$ representing Yes vs.~$-$ for No; lines \ref{step:sanitization} to \ref{step:endsanitization}). Since counts are inherently non-negative, the correction ($+$ versions) is more intuitive but overestimates the original bin counts.

\subsubsection{Choice of bin width $h$} 
A key hyperparameter that users need to specify in Algorithm \ref{alg:P3_hist} is the histogram bin width $h$. A large $h$ would result in a coarse histogram estimate of $f(\theta|\x)$, leading to biased quantile estimates of $F^{-1}_{\theta|\x}(q)$ from the subsequent histogram-based PRECISE procedure (even in the absence of  DP). Conversely, a small $h$, while reducing the global sensitivity $\Delta_H$ or leading to a large $m$ value for a fixed $\Delta_H$, would result in a large number of bins and thus a sparse histogram with numerous empty or low-count bins, which would also compromise the utility of the histogram\footnote{A similar narrative exists for $m$ when it is not back-calculated by fixing $\Delta_H$; a smaller $m$ implies lower $\Delta_H$ thus less DP noise but also leads to worse quantile estimation due to the data sparsity issue; and a higher $m$ implies richer information about the posterior distribution and more accurate quantile estimation but higher $\Delta_H$ and thus more DP noise, which counteracts the accuracy gains from the larger $m$.}.  

Furthermore, the choice of $h$ is critical in establishing the MSE consistency of the PP quantiles estimates based on the P$^3$ histogram (see Section \ref{sec:MSE_RAP} and Theorem \ref{thm:MSE_RAP}). In particular, if $h$ is too large, the discretization error can dominate the overall MSE -- regardless of data size and privacy level -- thereby undermining estimation accuracy. Conversely, if $h$ is too small, the sanitized bin index identified by the subsequent PRECISE procedure in Algorithm~\ref{alg:PRECISE} may fail to converge to the correct bin. More theoretical justification and requirements on $h$ and trade-offs are discussed in detail in Section \ref{sec:MSE_RAP}.

\subsubsection{Effects of bounds $(L, U)$ on P$^3$ utility and bin collapsing} \label{sec:collap}
Unknown parameters in a statistical model may be naturally bounded (e.g., proportions $\in[0,1]$) or unbounded, such as Gaussian mean $\in(-\infty,\infty)$, or bounded on one end (e.g., variance $\in(0,\infty)$). In the DP framework, bounds on numeric quantities, whether statistics or parameters, are necessary in many cases to design or apply a mechanism to achieve DP guarantees. Though this may be regarded as a strong assumption from the statistical theory perspective, real-life data and scenarios often support bounding on data or parameters, justifying bounding for practical applications. 

The global bounds $(L,U)$ for $\theta$ impact PRECISE's performance on PP quantile estimation based on the P$^3$ histogram, as they affect the amount of noise required to reach the preset DP guarantee level -- wider bounds often imply more noise. 
On the other hand, it is important not to impose unreasonably tight bounds to the extent that they cause significant bias or information loss. As a result, in practice, $(L,U)$ are often wide regardless of $n$.  

As $n$ increases, $f(\theta|\x)$ becomes increasingly concentrated around the underlying ``population'' parameter $\theta_0$. If $(L,U)$  are static and remain wide regardless of $n$, they may become unnecessarily conservative and degrade the privacy-utility trade-off. In Algorithm \ref{alg:P3_hist},  wide $(L,U)$ can lead to many empty or low-count bins in the histogram when $n$ is large, particularly in the tails. This, in turn, causes the P$^3$ histogram $H^*$ to be heavily perturbed with excessive noise added to the bin counts. Tighter bounds should be considered instead to leverage the increasing concentration of  $f(\theta|\x)$ as $n$ grows, reducing information loss and improving the accuracy of $H^*$. However,  even though it is theoretically possible to characterize the rate at which the interval width $U(n)-L(n)$ decreases with increasing $n$, this alone is insufficient for implementing a DP mechanism, which requires explicit values for $L(n)$ and $U(n)$ individually.  Determining these values would depend on the unknown true parameter $\theta_0$, posing a practical challenge to proposing analytical bounds $(L(n), U(n))$. 

To address this, we incorporate a subroutine in Algorithm \ref{alg:P3_hist} (lines~\ref{step:identify} to \ref{step:reindex}) to allow the procedure to adjust overly conservative global bounds $(L,U)$ by collapsing empty or small bins at the two tails of the posterior histogram before adding DP noise. This collapsing step effectively reduces excessive noise injection. There are two types of collapsing thresholds: $(\tau_L,\tau_U)$ can be thresholds on bin counts or proportions of bins to be collapsed, to be pre-set at the discretion of the data curator. In the former, they can be 0 or small positive integers such as 1, meaning the empty or bins with counts $\le1$ are kept collapsing until a bin with count $\ge1$ or $\ge2$ is encountered; in the latter, they are values close to 0, such as 2\% of bins on the left and 3\% on the right tail; suppose $B = 100$, this would correspond to collapsing 2 bins on the left and 3 bins on the right, regardless of the bin counts.  Note that the collapsing subroutine does not incur additional privacy loss for several reasons. First, the collapsing does not affect global sensitivity $\Delta_H$ under Assumption~\ref{con:DP}. Second, the bin breakpoints are data-independent and determined by pre-specified $(U, L, h)$. Third, the counts for all bins are sanitized, including the newly formed bins from collapsing. Fourth, $H^*$, the output from  Algorithm \ref{alg:P3_hist},  is an intermediate product; it will not be released but rather serves as input to the PRECISE procedure in Algorithm \ref{alg:PRECISE} (Section \ref{sec:PRECISE}) to generate PP quantiles by uniformly sampling from sanitized bins. Consequently, it is not possible to infer how many bins were collapsed at either end based solely on the released PP quantiles. Ultimately, the data curator may opt not to share the values of $(\tau_L,\tau_U)$ with the public for absolute assurance, as they are of no use for users who are interested in PP quantiles or PPIE.

\subsection{PRECISE based on P$^3$ Histogram}\label{sec:PRECISE}
After obtaining the P$^3$ histogram $H^*$ for parameter $\theta$, we can derive the PP quantile estimates for $\theta$ from $H^*$ via the PRECISE procedure in Algorithm \ref{alg:PRECISE} with the same DP guarantees per immunity to post-processing of DP.

\subsubsection{Four versions of PRECISE}
PRECISE has four versions $\{+m, +m^*, -m, -m^*\}$. $+$ and $-$ are inherited from the P$^3$ histogram procedure in Algorithm \ref{alg:P3_hist}. The choice between $m$ and $m^*$ depends on whether the cumulative density function estimate based on $H^*$ is normalized by the pre-specified total $m$ of posterior samples of $\theta$, or by the sum of sanitized bin counts  $m^*=\sum_{i=0}^{B'\!+1}c^*_i$.  While using $m$ leverages the fact that it is a known constant and enhances the stability of the output from  Algorithm \ref{alg:PRECISE}, it at the same time introduces intra-inconsistency for normalized $H^*$ as the individual bin counts in $H^*$ are sanitized and their sum is highly unlikely to be equal to $m$ in actual implementations. In fact, the sum equals to $m$ only by expectation for the $-$ version of P$^*$ and is biased upward for the $+$ version. The simulation studies in Section \ref{sec:simu} compare the performance of the four versions of PRECISE.

\begin{algorithm}[!h]
\caption{PRECISE}\label{alg:PRECISE}
\SetAlgoLined
\SetKwInOut{Input}{input}
\SetKwInOut{Output}{output}
\Input{P$^3$ histogram $H^*$ from Algorithm \ref{alg:P3_hist} with bins $\{\Lambda_b\}_{b=0}^{B'+1}$ and sanitized bin counts $\{c^*_b\}_{b=0}^{B'+1}$, confidence level $1\!-\!\alpha\!\in\!(0,1)$, PRECISE version} 
\Output{PPIE at level $1-\alpha$ for $\theta_0$: $\big(\theta^*_{\alpha/2}, \theta^*_{1-\alpha/2}\big)$.}
\textbf{if} PRECISE version $==+m^*$ or $-m^*$, \textbf{then} obtain the indices per\label{step:normalization*}
\[b^*_{\alpha/2} = \min\Bigg\{\underset{b \in \{0, \ldots, B' + 1\}}{\arg\min}\Big|\!\sum_{i=0}^bc^*_i- \frac{\alpha}{2}\!\sum_{i=0}^{B'\!+1}c^*_i\Big|\Bigg\},\; b^*_{1-\alpha/2} = \min\Bigg\{\underset{b \in \{0, \ldots, B' + 1\}}{\arg\min}\Big|\!\sum_{i=b}^{B'+1}c^*_i- \frac{\alpha}{2}\!\sum_{i=0}^{B'\!+1}c^*_i\Big|\Bigg\}.\]\\
\textbf{if} PRECISE version $==+m$ or $-m$, \textbf{then} obtain  the indices per\label{step:normalization}
\[b^*_{\alpha/2} = \min\Bigg\{\underset{b \in \{0, \ldots, B' + 1\}}{\arg\min}\Big|\!\sum_{i=0}^bc^*_i\!-\! \frac{\alpha}{2}m\Big|\Bigg\}, \;b^*_{1-\alpha/2} = \min\Bigg\{\underset{b \in \{0, \ldots, B' + 1\}}{\arg\min}\Big|\!\sum_{i=b}^{B'+1}c^*_i\!-\! \frac{\alpha}{2}m\Big|\Bigg\}.\]\\
Draw $\theta^*_{\alpha/2}$ uniformly from $I_{b^*_{\alpha/2}}$, and  $\theta^*_{1-\alpha/2}$ uniformly from $I_{b^*_{1-\alpha/2}}$.\label{step:uniform}
\end{algorithm}

\subsubsection{Rationale for PRECISE}
The key to any valid PP inference -- PPIE included -- is to acknowledge and account for the additional source of variability introduced by DP sanitation, on top of the sampling variability of the data. Ignoring the former would lead to invalid inference, and in the context of PPIE, potential under-coverage and failure to satisfy Definition \ref{defn:CI}.
PRECISE accounts for both sources of variability. Rather than taking the route of explicitly quantifying the uncertainty for a PP estimate of $\theta$ and then calculating the half-width for its PP interval estimate based on the quantified uncertainty, PRECISE instead reframes the interval estimation for $\theta$ as a point estimation problem, leveraging the Bayesian principles.

Specifically, the central posterior interval with level of $(1-\alpha)\times 100\%$ is formulated as $\big(F^{-1}_{\theta|\x}(\alpha/2),F^{-1}_{\theta|\x}(1-\alpha/2)\big)$ by definition.  PRECISE first identifies an index $b^*$ (lines~\ref{step:normalization*} and \ref{step:normalization} in Algorithm \ref{alg:PRECISE}) by minimizing the absolute difference between the  ``empirical'' cumulative counts of samples  at $\alpha/2$ vs.~the expected cumulative counts out of a total of $m^*$ (or $m$) from both ends of the distribution; that is, $\sum_{i\leq b}c^*_i$ and  $\alpha m^*/2$ (or $\alpha m/2$), and  $\sum_{i\geq b}c^*_i$ and $\alpha m^*/2$  (or $\alpha m/2$). A random sample of $\theta$ is then drawn from the identified bin $I_{b^*}$ as the $\alpha/2\times100\% $ PP quantile estimate; similarly for the  $(1-\alpha/2)\times100\%$ PP quantile estimate. The two PP quantile estimates can then be plugged in directly to form PPIE $\big(\theta^*_{\alpha/2}, \theta^*_{1-\alpha/2}\big)$ for $\theta$. Section \ref{sec:MSE_RAP} shows that the PP quantile outputs from PRECISE are consistent estimators for the true posterior quantiles as the sample size $n$ or privacy loss goes to $\infty$.


\subsubsection{MSE consistency of PRECISE}\label{sec:MSE_RAP}
We establish the MSE consistency for the pointwise PPIE from PRECISE ($+m^*$) in Theorem \ref{thm:MSE_RAP} with $\varepsilon$-DP. Results for the other three PRECISE variants ($+m,-m^*,-m$) and other DP notation variants (e.g., $\mu$-GDP) can be similarly proved. 

\begin{theorem}[MSE consistency of PRECISE ($+m^*$)]\label{thm:MSE_RAP}
Given i.i.d. data $\mathbf{x} \!=\! \{x_i\}_{i=1}^n\!\sim \!f(X|\theta)$ and prior $f(\theta)$ that is non-informative relative to the amount of data; let $\{\theta_j\}_{j=1}^m$ denote the set of samples drawn from the posterior distribution $f(\theta|\mathbf{x})$. Under Assumption~\ref{con:DP} and the constraint $m = o\left(e^{\varepsilon\sqrt{n}/2} n^{1/4} \varepsilon^{1/2} \right)$, the $q^{th}$ posterior quantile $\theta^*_{(q)}$ released by $\mathcal{M}$: PRECISE $(+m^*)$ with $\varepsilon$-DP as a PP estimate for $F^{-1}_{\theta|\mathbf{x}}(q)$ (the true $q^{th}$ from the posterior $f(\theta|\x)$ satisfies
\begin{align}
&\mathbb{E}_{\boldsymbol{\theta}}\mathbb{E}_{\mathcal{M}|\boldsymbol{\theta}}\big(\theta^*_{(q)} - F^{-1}_{\theta|\mathbf{x}}(q)\big)^2 \notag\\
\leq \; &\underbrace{\mathcal{O}\left(m^{-2}\right)}_{T_0}+\underbrace{\mathcal{O}\left(\frac{1}{\sqrt{n}}e^{-\varepsilon\sqrt{n}/2}\right)}_{T_1} + 
\underbrace{\mathcal{O}\left(\frac{1}{mn}\right)}_{T_2}\label{eqn:thm2}\\
=\: & \begin{cases}
    T_0 = \mathcal{O}(m^{-2})      &\text{if } m = o(n)\\
    T_2 = \mathcal{O}(m^{-1}n^{-1}) &\text{if } m \in \left(\Omega(n),  \mathcal{O}(e^{\varepsilon\sqrt{n}/2} n^{-1/2})\right)\\
    T_1 = \mathcal{O}\left(\frac{1}{\sqrt{n}}e^{-\varepsilon\sqrt{n}/2}\right) &\text{if } m \in \left(\Omega(e^{\varepsilon\sqrt{n}/2}n^{-1/2}), o(e^{\varepsilon\sqrt{n}/2}n^{1/4}\varepsilon^{1/2})\right).
\end{cases}
\end{align}
\end{theorem}
The detailed proof is provided in Appendix \ref{proof:PRECISE}. The core idea is to use the Cauchy-Schwarz inequality to decompose the MSE between the privatized posterior quantile $\theta^*_{(q)}$ and the population posterior quantile $F^{-1}_{\theta|\mathbf{x}}(q)$ into two components: 1) the MSE between $\theta^*_{(q)}$ and posterior quantile $\theta_{(q)}$ based on $m$ posterior samples, which accounts for both the histogram discretization error and DP sanitization error (terms $T_0$ and $T_1$); 2) the MSE between $\theta_{(q)}$ and $F^{-1}_{\theta|\mathbf{x}}(q)$ to account for the posterior sampling error as compared to the analytical quantile $F^{-1}_{\theta|\mathbf{x}}(q)$ (term $T_2$). The first MSE  component involving $\theta^*_{(q)}$ and $\theta_{(q)}$ is further analyzed in two steps. First, we show that the sanitized bin index identified in line~\ref{step:normalization*} of Algorithm \ref{alg:PRECISE} converges to the bin containing $\theta_{(q)}$, which requires a necessary upper bound on the number of posterior samples $m\!=\!o\left(e^{\varepsilon\sqrt{n}/2} n^{1/4} \varepsilon^{1/2} \right)$. Then, conditional on the bin being correctly identified,  the error due to discretization  and the subsequent uniform sampling from the identified bin (line~\ref{step:uniform}) is quantified as
 $T_0$, and the DP-induced error is captured by $T_1$.

Theorem~\ref{thm:MSE_RAP} reveals a critical trade-off governed by the number of posterior samples $m$, which relates to the bin width $h=1/(2mG(n))$ of the posterior histogram with DP guarantees. The MSE contains three components -- discretization error $T_0$, DP-induced error $T_1$, and posterior sampling error $T_2$, and its practical interpretation depends on which term dominates for different ranges of $m$. If $m$ is too small (and $h$ is consequently too large), the discretization error $T_0 = \mathcal{O}(m^{-2})$ dominates, resulting in an lower bound on the overall MSE regardless of how large the data size $n$ or privacy loss parameter $\varepsilon$ is. Therefore, to leverage large $n$ or more permissive privacy settings so to achieve more accurate PP quantile estimation, Theorem~\ref{thm:MSE_RAP} highlights the necessity for $m$ to grow with $n$ and  $\varepsilon$ as along as it does go beyond the upper bound $m\!=\!o\left(e^{\varepsilon\sqrt{n}/2} n^{1/4} \varepsilon^{1/2} \right)$ required for the correct bin identification.

Based on Theorem \ref{thm:MSE_RAP}, we show the PPIE from PRECISE asymptotically satisfies Definition \ref{defn:CI} and achieves the nominal coverage as $n$ or $\varepsilon$ increases in the proposition below. 
\begin{proposition}[Asymptotical nominal coverage of PRECISE]\label{prop:coverage}
    Under the conditions of Theorem \ref{thm:MSE_RAP}, as $n\rightarrow \infty$ or $\varepsilon \rightarrow \infty$,
    the PPIE via PRECISE for the true parameter $\theta_0$ at level $(1-\alpha)$ satisfies $\Pr(\theta^*_{\alpha/2}\leq \theta_0 \leq \theta^*_{1-\alpha/2}| \mathbf{x})\rightarrow 1-\alpha$.
\end{proposition}

\subsubsection{Other usage and extensions of the PRECISE Procedure}

The PRECISE procedure can also be used to release a PP quantile estimate for $F_{\theta|\x}^{-1}(q)$ from the posterior distribution of $\theta$ given sensitive data $\x$ and $q\in(0,1)$ (e.g., median, Q1, Q3, etc)\footnote{PRECISE should not be used to sanitize sample quantiles of the sensitive data $\x$ itself, which is a well-studied problem; for that, users may use existing procedures like PrivateQunatile \citep{smith2011privacy} and JointExp \citep{gillenwater2021differentially}.}, in addition to constructing PP interval estimates as demonstrated above.  If users opt to collapse bins on the tails, PRECISE may not be accurate when $q$ is very close to 0 or 1, such as the minimum and maximum. 

In the multivariate case of $\bs\theta = (\theta^{(1)}, \ldots, \theta^{(p)})^{\top}$, both  P$^3$ in Algorithm \ref{alg:P3_hist} and PRECISE in Algorithm \ref{alg:PRECISE} can be utilized to generate \emph{pointwise} PPIE for each dimension $\theta^{(k)}$ where $1\leq k \leq p$. This can be achieved by obtaining posterior samples from the marginal posterior distribution of $\theta^{(k)}$, while allocating the privacy budget $\bs\eta$ across all $p$ dimensions according to the privacy loss composition principle of the specific DP notion being employed. Note that obtaining marginal posterior samples on $\theta^{(k)}$ from the posterior distribution does not mean that the sampling has to come directly from $f(\theta^{(k)}|\x)$; one can sample from the joint posterior from $f(\bs\theta|\x)$ if it is easier, but only retain the samples on $\theta^{(k)}$ after sampling for the P$^3$ histogram construction and PP quantile release on $\theta^{(k)}$.  The vast majority of interval estimation problems in practice rely on pointwise interval estimations even in the non-private setting. In cases where there is an interest in obtaining joint or simultaneous PPIE for $\bs\theta$ when $p\!\ge\!2$, Algorithms \ref{alg:P3_hist} and \ref{alg:PRECISE} can still be used in principle, but the GS in Theorem \ref{thm:G} no longer applies and $G(n)$ for the \emph{joint} posterior $f(\bs\theta|\x)$ must be derived for Algorithm \ref{alg:P3_hist}, which can be analytically or numerically hard.  Though the entire privacy budget $\bs\eta$ can be devoted to sanitizing a single multi-dimensional $H^*$ without being split among the $p$ dimensions, given the potentially large $G(n)$ when $p\ge2$ and the well-known curse of dimensionality associated with histograms, the simultaneous PPIE can be of low utility. This is compounded by the fact that joint intervals encode additional dependency information across the $p$ dimensions, necessitating greater noise injection to maintain DP guarantees at $\bs\eta$.

\subsubsection{An alternative to PRECISE}\label{sec:PP}
We also develop an alternative approach to PRECISE that is based on the exponential mechanism for privately estimating the posterior quantile. The approach, termed \emph{Private Posterior quantile estimator (PPquantile)}, is detailed in Algorithm \ref{alg:DPP}.  PPquantile is inspired by the PrivateQuantile procedure\footnote{We establish the MSE consistency of PrivateQuantile towards population quantiles in Theorem \ref{thm:PQ} for interested readers, the first result on that to our knowledge.} \citep{smith2011privacy} (Algorithm \ref{alg:PQ} in the appendix) for releasing private sample quantiles of data $\x$, but is also fundamentally different from the latter. PrivateQuantile outputs a sanitized sample quantile directly from the sensitive data $\x$, whereas PPquantile, like Algorithm \ref{alg:P3_hist}, outputs sanitized posterior quantiles for parameter $\theta$, which, as discussed in Section \ref{sec:GS}, is a significantly more complex problem to design a DP randomized mechanism for. Lines~\ref{step:identify} to \ref{step:reindex} in Algorithm \ref{alg:P3_hist} are unique and necessary for PPquantile, highlighting the fundamental differences from PrivateQuantile.  

\begin{algorithm}[!htbp]
\caption{PPquantile}\label{alg:DPP}
\SetAlgoLined
\SetKwInOut{Input}{input}
\SetKwInOut{Output}{output}
\Input{posterior distribution $f(\theta|\x)$,  quantile $q \!\in\! (0,\!1)$,  privacy loss $\varepsilon$, global bounds $(L,U)$ for $\theta$, number of posterior samples $m$.}
\Output{PP estimate $\theta^*_{(q)}$ of the population posterior quantile $F^{-1}_{\theta|\mathbf{x}}(q)$.}
Generate posterior samples $\{\theta_j\}_{j=1}^m \overset{\text{iid}}{\sim} f(\theta|\x)$\;
\vspace{1pt}
Replace $\theta_j \!<\! L$ with $L$ and $\theta_j \!>\! U$ with $U$\;
\vspace{1pt}
Sort $\theta_i$ in ascending order as $L=\theta_{(0)} \leq \theta_{(1)}\leq \ldots \leq \theta_{(m)}\leq \theta_{(m+1)} = U$\; 
\vspace{1pt}
Set $k = \arg\min_{j\in \{0,1, \ldots, m+1\}}{|\theta_{(j)}-F^{-1}_{\theta|\mathbf{x}}(q)|}$ \;
\vspace{-3pt}
For $j = 0, 1, \ldots, m$, set $y_j \!=\! (\theta_{(j+1)}\!-\!\theta_{(j)}) \cdot \mbox{exp}\left(-\frac{\varepsilon|j-k|}{2(m+1)}\right)$\;
\vspace{-2pt}
Sample an integer $j^* \!\in\! \{0, 1, 2, \dots, m\}$ with probability $y_{j^*}/\sum_{j=0}^{m} y_{j}$\;\label{step:exp_norm}
\vspace{-1pt}
Draw $\theta^*_{(q)} \sim \mbox{Uniform}(\theta_{(j^*)}, \theta_{(j^*+1)})$.
\end{algorithm}


\vspace{-6pt}\begin{theorem}[Utility guarantees for PPquantile in Algorithm \ref{alg:DPP}]\label{thm:utility_PP}
Assume $f(\theta|\mathbf{x})$ is continuous at $F^{-1}_{\theta|\mathbf{x}}(q)$ for $q\in(0,1)$. Let $m$ be the number of posterior samples on $\theta$ from its posterior distribution $f(\theta|\mathbf{x})$, $(U,L)$ be the global bounds on $\theta$, $\xi = e^{-\frac{\varepsilon}{2(m+1)}}$, $s \!=\! \min_{j\in \{0, 1, \ldots, m\}}(\theta_{(j+1)} \!-\! \theta_{(j)})$, and $p_{\min} \!=\! \inf_{|\tau-F^{-1}_{\theta|\mathbf{x}}(q)|\leq 2\eta} f_{\theta|\mathbf{x}}(\tau)$ for $\eta\!>\!0$. The PP  $q^{th}$ quantile $\theta^*_{(q)}$ from PPquantile of $\varepsilon$-DP in Algorithm \ref{alg:DPP}  satisfies
\begin{align}\label{eqn:PPquantile}
    \Pr\left(\Big|\theta^*_{(q)}-\theta_{(k)}\Big|>2u\right)
    \leq &\;\frac{U-L - 4u}{s}\cdot\frac{1-\xi}{1+\xi-\xi^{k+1}-\xi^{m-k+1}}\cdot\exp\left(-\frac{\varepsilon u p_{\min}}{4}\right)\\
    &\;+ \frac{2\eta}{u}\exp\left(-\frac{(m+1)u p_{\min}}{8}\right) + 2\exp\left(-\frac{m\eta^2p^2_{\min}}{12(1-q)}\right)\mbox{ for } 0\!\leq\! u\!\leq \!\eta.\notag
\end{align}
\end{theorem}


The detailed proof is provided in Appendix \ref{proof:PPQ}. In brief, per the Bernstein-von Mises theorem, $\theta_{(m)} -\theta_{(1)} \asymp n^{-1/2}$ as $n \rightarrow \infty$, implying $s=\mathcal{O}(n^{-1/2}m^{-1})$. Under  regularity condition that $p_{\min}$ is constant for $\eta \asymp  n^{-1/2}$, we let $\{\varepsilon\asymp m, m \asymp n^{2}, u \!\asymp\! n^{-1}\}$ and set $U-L \asymp n^{-5/2}$, then the right hand side of Eq.~\eqref{eqn:PPquantile} $\to1$ and $\theta^*_{(q)} \overset{p}{\rightarrow} \theta_{([qm])}$ asymptotically.

Theorem \ref{thm:utility_PP} suggests that PPquantile can return an asymptotically accurate posterior quantile estimate with a high probability. The probability depends on multiple hyperparameters $(U,L,u,\eta, p_{\text{min}})$ and assumption regarding their relationship (e.g. how $L,U$ individually shrink with $n$) that can be challenging to verify, making it difficult for practical implementation and also leading to potential under-performance in finite-sample scenarios (e.g., if overly conservative bounds $L,U$ are used). We will continue to explore ways to enhance the practical application of the PPquantile procedure, given that it is theoretically sound.

Algorithm \ref{alg:DPP} and Theorem \ref{thm:utility_PP}  are presented for achieving $\varepsilon$-DP guarantees. Although they can be extended to achieve other DP guarantees, such as zCDP or GDP, optimally quantifying the privacy loss of the exponential mechanism under relaxed DP frameworks remains challenging. For example, the classical conversion from
\citet{bun2016concentrated} suggests that $\varepsilon$-DP implies $\varepsilon^2/2$-zCDP, but this bound is often too loose for practical applications. More refined analyses rely on the concept of bounded range (BR) \citep{durfee2019practical}, under which an exponential mechanism of $\varepsilon$-BR is shown to satisfy $(\varepsilon^2/8 + O(\varepsilon^4), \varepsilon^2/8)$-zCDP \citep{pmlr-v132-cesar21a,dong2020optimal} -- a notable improvement, though still suboptimal. \citet{gopi2022private} derives optimal GDP bounds of $G/\sqrt{\mu}$ for the exponential mechanism when the utility functions are assumed to be $\mu$-strongly convex and the perturbations are $G$-Lipschitz continuous. However, these assumptions are often difficult to verify or calibrate in practice, which limits their applicability.

\section{Experiments}\label{sec:exp}
We evaluate our methods (4 versions of PRECISE) for generating PPIE through extensive simulation studies (Section \ref{sec:simu}), where we also compare the methods to some existing PPIE procedures, and two real-world data applications (Section \ref{sec:case}).

\subsection{Simulation studies}\label{sec:simu}
The goals of the simulation studies are 1) to validate that PRECISE achieves nominal coverage across various inferential tasks in data of different sample sizes at varying privacy loss; 2) to showcase the improved performance of PRECISE over the existing PPIE methods, i.e., narrower intervals while maintaining correct coverage.  \emph{The experiment results presented below suggest that both goals have been attained.}

We simulate and examine several common data and inferential scenarios -- Gaussian mean and variance, Bernoulli proportion, Poisson mean, and linear regression. We choose these inferential tasks because they are commonly studied by the existing PPIE methods  (see Section \ref{sec:related_work}); focusing on these tasks enables a more comprehensive and fair comparison between PRECISE and these methods.
The performance of PRECISE in comparison with some existing PPIE methods is summarized in Table \ref{tab:methods}.

\begin{table}[!htb]
\caption{PPIE methods examined in the experiments and performance summary}\label{tab:methods}\vspace{-3pt}
\centering
\setlength{\tabcolsep}{1pt}
\resizebox{1\textwidth}{!}{
\begin{tabular}{lccl}
\toprule
Method & Experiments & Applicable & \multicolumn{1}{c}{(nominal CP achieved ?) Performance summary}\\
&& DP & \\
\midrule
PRECISE ($+$) & all & all & $(\checkmark)$ $+m^*$ is the best, $+m$ is worse than $+m^*\!,-m^*\!,-m$\\
PRECISE ($-$) & all & all & $(\times$) $-m^*,-m$ similar, under-coverage when $n\varepsilon\leq100$\\
MS \citep{liu2016model} & all & all & ($\checkmark$) fast and flexible, wide intervals\\
PB \citep{ferrando2022parametric}  & all & all & ($\checkmark$) fast and flexible, wide intervals \\
deconv \citep{wang2022differentially}  & Gaus. & $\mu$-GDP & ($\checkmark$) the widest intervals\\
repro \citep{awan2023simulation}  & Gaus. & $\mu$-GDP & ($\checkmark$) slow computation for large $n$, wide intervals\\
BLBquant \citep{chadha2024resampling} & Gaus., Bern., Pois. & $\varepsilon$-DP & ($\times$) slow for large $n$, under-coverage due to narrow width \\
Aug.MCMC \citep{ju2022data}  & linear regression &  $\varepsilon$-DP & ($\checkmark$) slow computation \& convergence, sensitive to prior\\
\midrule
\multicolumn{4}{p{1.3\linewidth}}{$^\dagger$ There are PPIE methods \citep{karwa2017finite, covington2021unbiased, evans2023statistically, d2015differential} designed specifically for Gaussian means, they are not evaluated in this work as previous studies \citep{du2020differentially,ferrando2022parametric} have demonstrated that they are inferior to those listed in the table.}\\
\multicolumn{4}{p{1.3\linewidth}}{We opt to exclude the approaches in \citet{avella2023differentially} and \citet{Wang_Kifer_Lee_2019}. 
Both are procedurally complicated and would be excessive for the inferential tasks in our experiments with closed-form estimators.}\\ 
\bottomrule
\end{tabular}} \vspace{-3pt}
\end{table}

\subsubsection{Simulation settings}\label{sec:settings} \vspace{-3pt}
We examine a wide range of sample size $n\in(100, 500, 1000, 5000, 10000, 50000)$ and privacy loss $\varepsilon\in(0.1, 0.5, 1, 2, 5, 10, 50)$ for the Laplace mechanism of $\varepsilon$-DP and $\mu \in (0.1, 0.5, 1, 5)$ for the Gaussian mechanism of $\mu$-GDP. The very large values for $\varepsilon$ or $\mu$ are used to demonstrate whether the  PPIE converges to the original inferences as the privacy loss increases.

We simulate Gaussian data  $\x\!\sim\!\mathcal{N}(\mu\!=\!0, \sigma^2\!=\!1)$,  Poisson data  $\x\!\sim\!\mbox{Pois}(\lambda\!=\!10)$, and Bernoulli data $\x\!\sim\!\mbox{Bern}(p\!=\!0.3)$. When implementing PRECISE, we use prior $f(\mu, \sigma^2)\!\propto\!(\sigma^2)^{-1}$ for the Gaussian data, the corresponding marginal posteriors are $f(\mu|\x) = t_{n-1}(\bar{x}, s^2/n)$ and $f(\sigma^2|\x) = \text{IG}((n\!-\!1)/2, (n\!-\!1)s^2/2)$, respectively, 
where $s^2$ is the sample variance and IG($\alpha,\beta$) is the inverse-gamma distribution with shape parameter $\alpha$ and scale parameter $\beta$. For the Poisson data,  $f(\lambda)\!=\!\text{Gamma}(\alpha \!=\! 0.1, \beta\!=\!0.1)$ and the corresponding posterior is $f(\lambda|\x) \!=\!\text{Gamma}(\alpha \!+\! \sum_{i=1}^n x_i, \beta \!+\! n)$. For the Bernoulli data, $f(p)\!=\!$ beta$(\alpha \!=\! 1,\beta \!=\! 1)$ and the corresponding posterior is  $f(p|\x)\!=\!$ beta($\alpha+\sum_{i=1}^n \!x_i, \beta +n \!-\!\sum_{i=1}^n\!x_i$). For the linear model $\x\! =\! \beta_0 + \beta_1\mathbf{z} + \mathcal{N}(0,\sigma^2\!=\!0.25^2)$, where $\beta_0 \!=\! 1, \beta_1 \!=\! 0.5$ and $\mathbf{z} \!\sim \!\mathcal{N}(0,1)$, we use prior $f(\beta_0, \beta_1,\sigma^2)\!\propto\! \sigma^{-2}$; and the marginal posterior is $\beta_1|\mathbf{z}, \mathbf{x} \sim t_{n-2}\left(\hat{\beta}_1, \frac{\hat{\sigma}^2}{\sum_{i=1}^n(z_i-\bar{z})^2}\right)$, where $\hat{\sigma}^2 = \frac{\sum_{i=1}^n(x_i-\mathbf{z}_i\widehat{\boldsymbol{\beta}})^2}{n-2}$ and $\widehat{\boldsymbol{\beta}} \!=\! (\mathbf{z}^{\top}\mathbf{z})^{-1}\mathbf{z}^{\top}\x$\footnote{Though conjugate priors are employed in all the experiments that have closed-form posterior distributions are easy to sample from, this is not a requirement for PRECISE, which can be coupled with all posterior sampling methods such as MCMC to construct PPIE.}.  
Other implementation details, including the choice of the hyperparameters for the P$^3$ and PRECISE algorithms and for the comparison methods in Table \ref{tab:methods}, are provided in Appendix \ref{ape: hyper}.

\subsubsection{Results on PPIE validity}\label{sec:results}
The inferential results are summarized by coverage probability (CP) and widths of $95\%$ interval estimates over 1,000 repeats in each simulation setting. Due to space limitations, we present the results with $\varepsilon$-DP guarantees in Figures \ref{fig:Normal} and \ref{fig:Bern+Pois} (the results with $\mu$-DP  are available in Appendix \ref{ape:exp_results}, the findings on the relative performance across the methods are consistent with those with $\varepsilon$-DP). 

\begin{figure}[!htbp]
 \centering
  \footnotesize{\textbf{Gaussian mean}}\\[2pt]
  \includegraphics[width = 0.9\textwidth, trim={0.05in 0.1in 0.05in, 0in},clip]{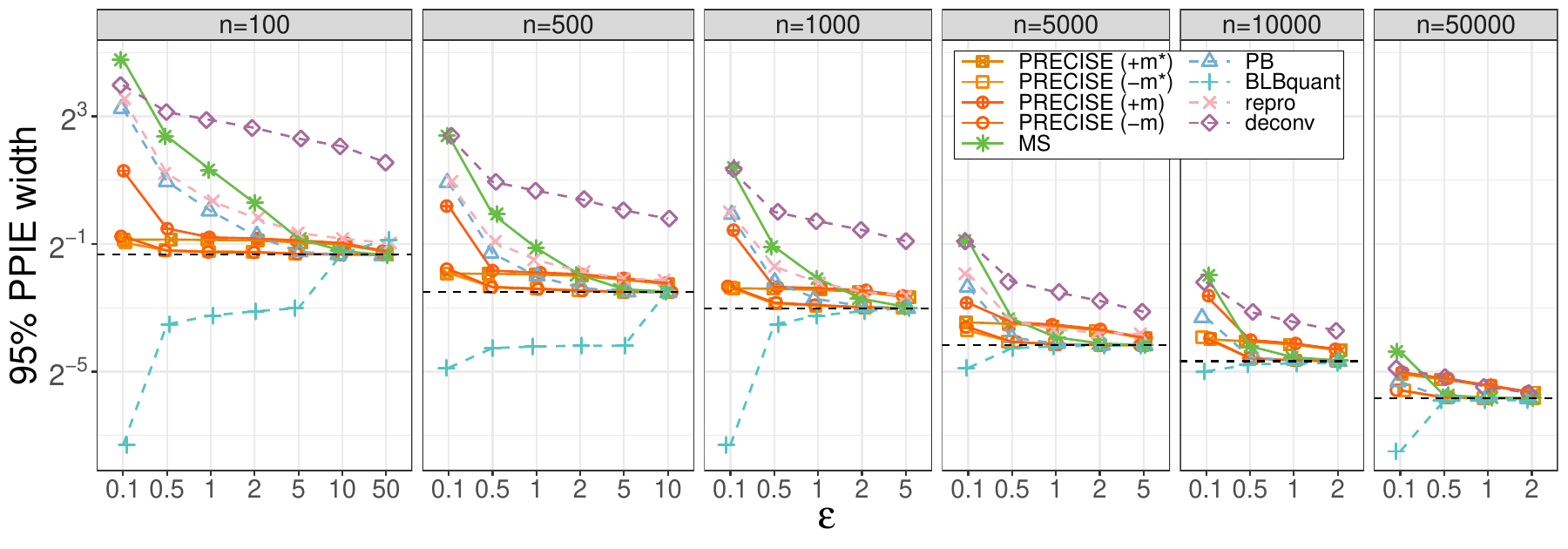}
  \includegraphics[width = 0.9\textwidth, trim={0.1in 0.1in 0.05in, 0.05in},clip]{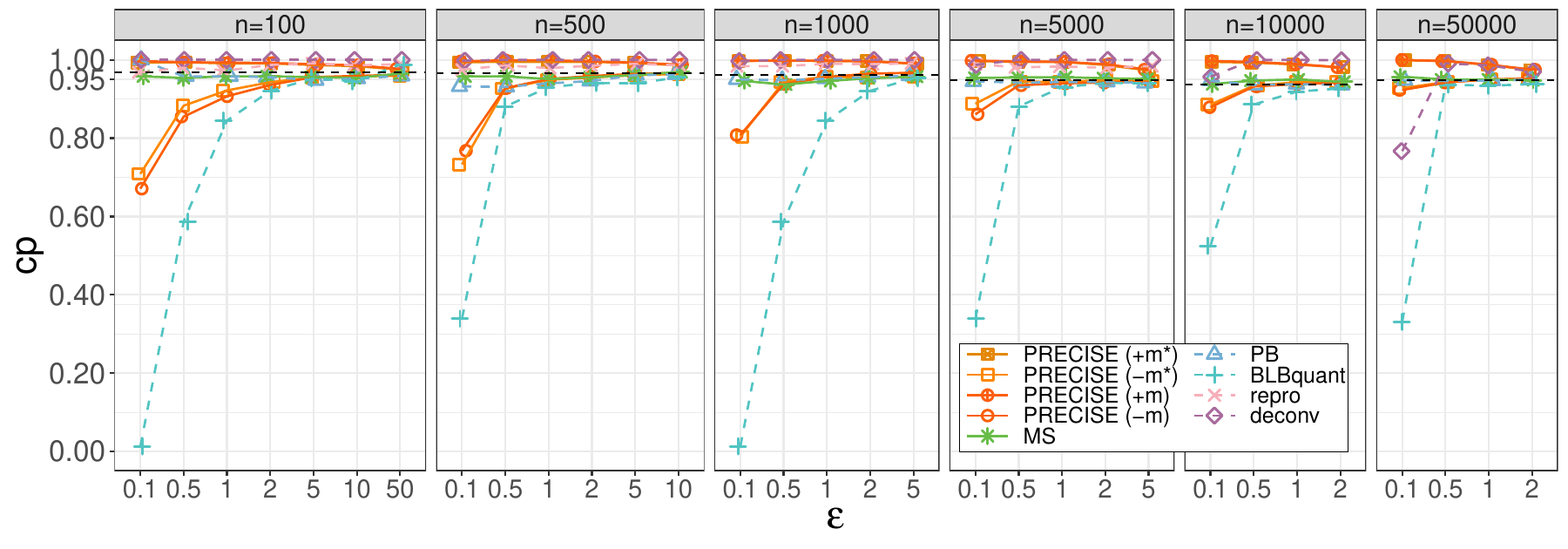}\vspace{-3pt}
  \footnotesize{\textbf{Gaussian variance}}\\[2pt]
  \includegraphics[width = 0.9\textwidth, trim={0.05in 0.1in 0.05in, 0in},clip]{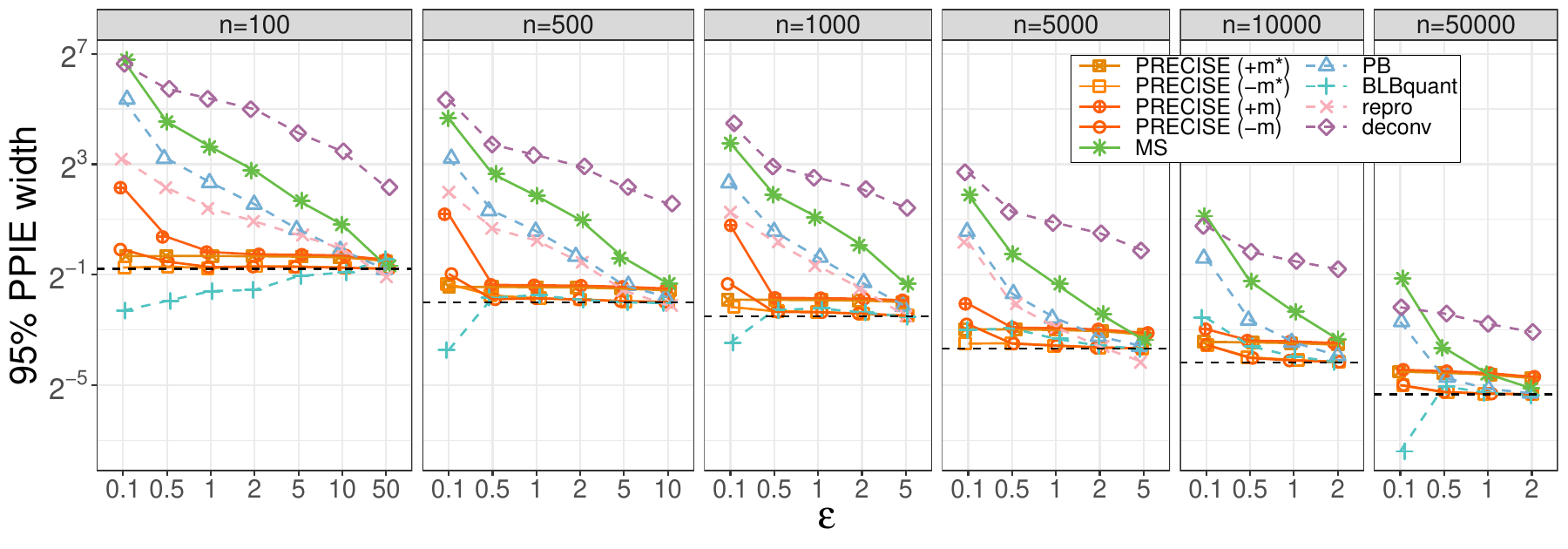}
  \includegraphics[width = 0.9\textwidth, trim={0.1in 0.1in 0.05in, 0.05in},clip]{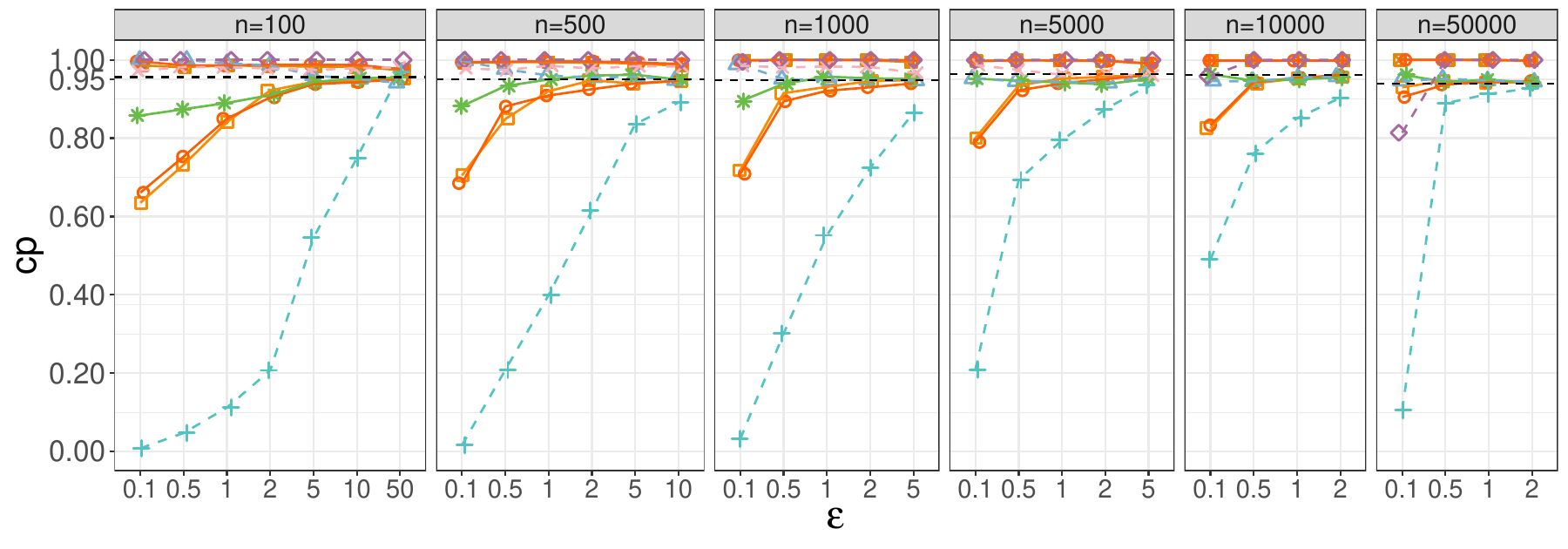}
  \vspace{-6pt}
\captionsetup{justification=raggedright, singlelinecheck=false}
\caption{Comparisons of PPIE width and CP for Gaussian mean and variance. All methods use the Laplace mechanism of $\varepsilon$-DP except for deconv and repro that are designed for $\mu$-GDP; $\mu$ is calculated given $\varepsilon$ and $\delta\!=\!1/n$ per Lemma \ref{lemma:mu_epsdelta}.   Black dashed lines represent the original non-private results. The results on $\mu$-GDP are in the appendix.}\label{fig:Normal}\vspace{-12pt}
\end{figure}

\begin{figure}[!htbp]
  \centering
  \footnotesize{\textbf{Poisson mean}}\\[4pt]
  \includegraphics[width = 0.9\textwidth, trim={0.05in 0.1in 0.05in, 0in},clip]{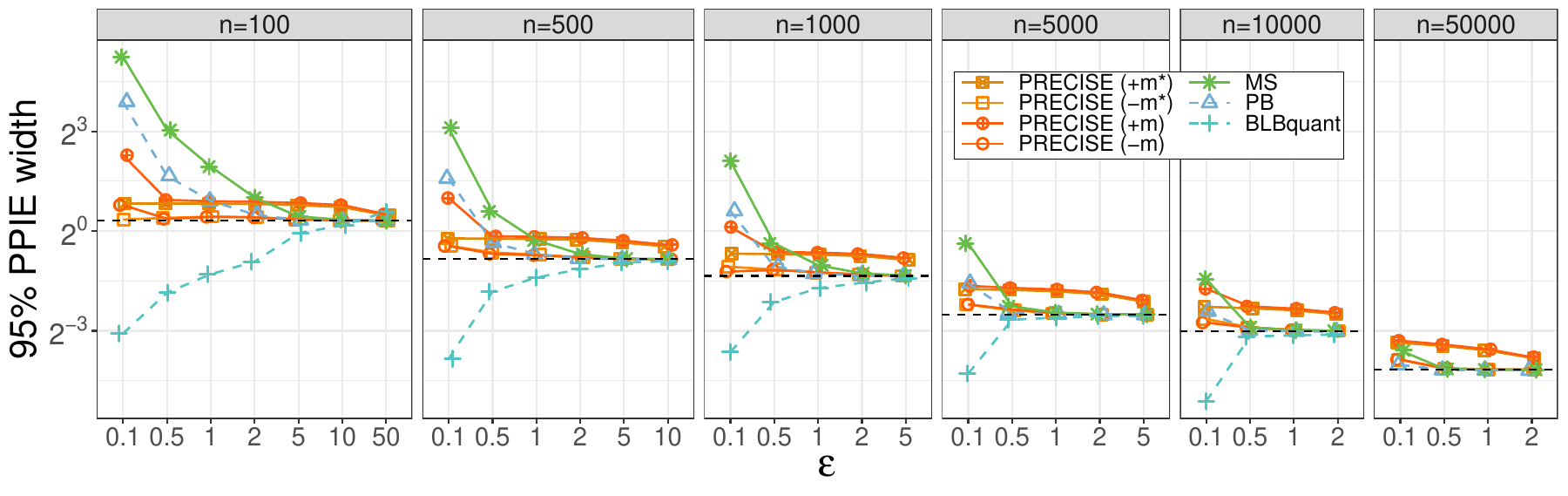}
  \includegraphics[width = 0.9\textwidth, trim={0.1in 0.1in 0.05in, 0.05in},clip]{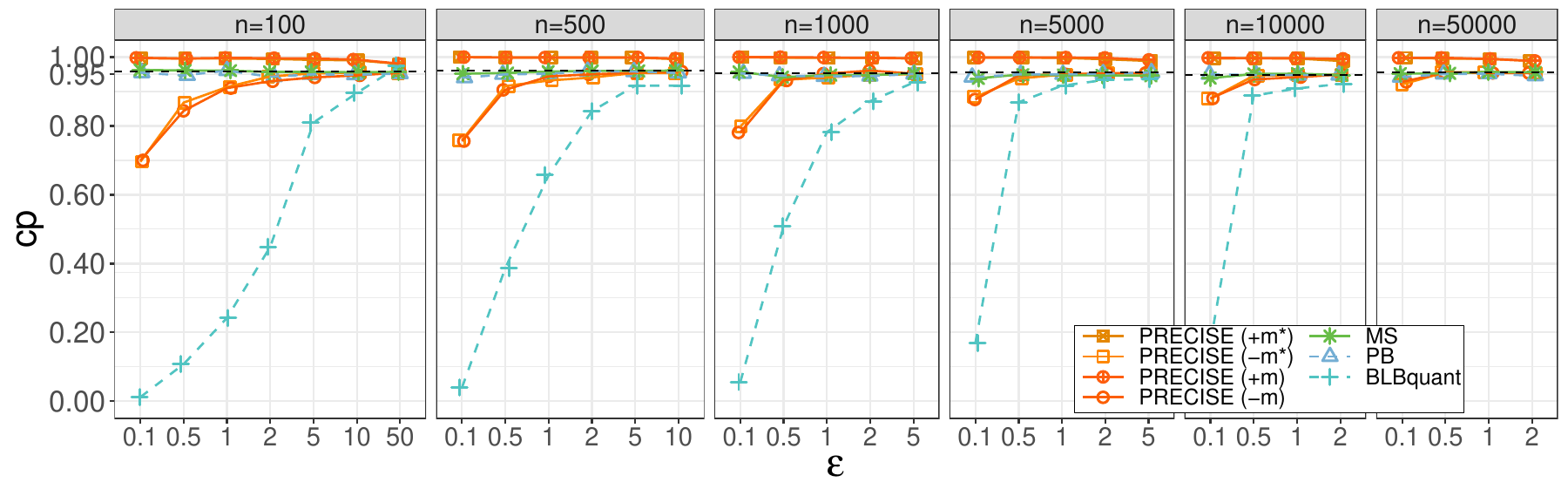}
  \footnotesize{\textbf{Bernoulli proportion}}\\[4pt]
  \includegraphics[width = 0.9\textwidth, trim={0.05in 0.1in 0.05in, 0in},clip]{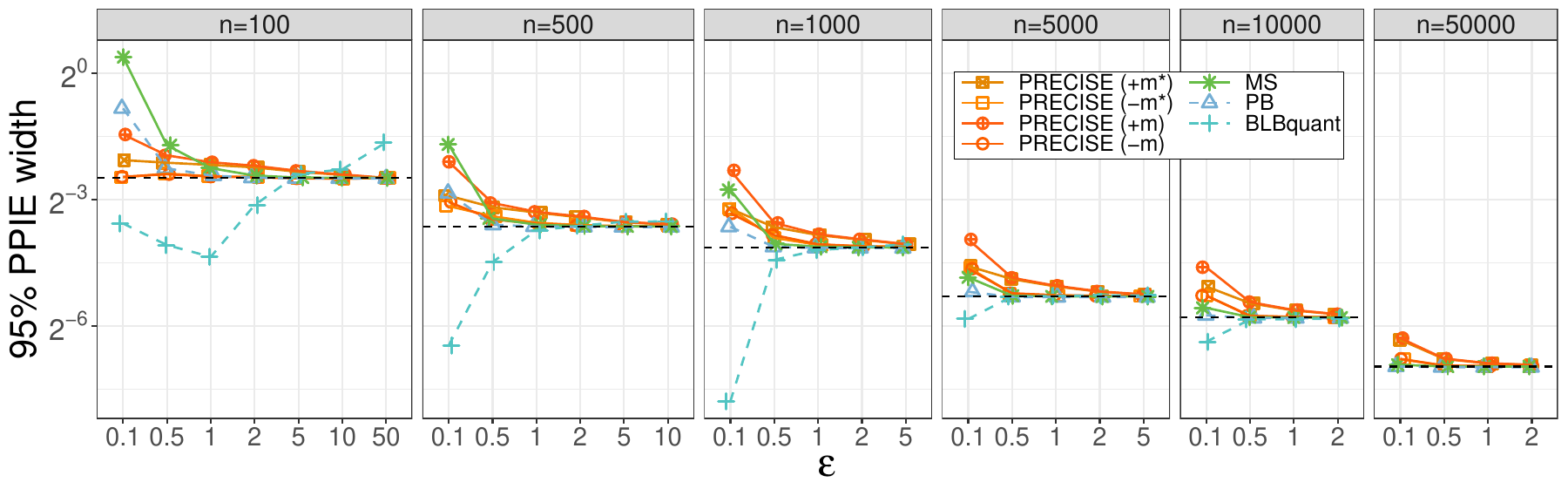}
  \includegraphics[width = 0.9\textwidth, trim={0.1in 0.1in 0.05in, 0.05in},clip]{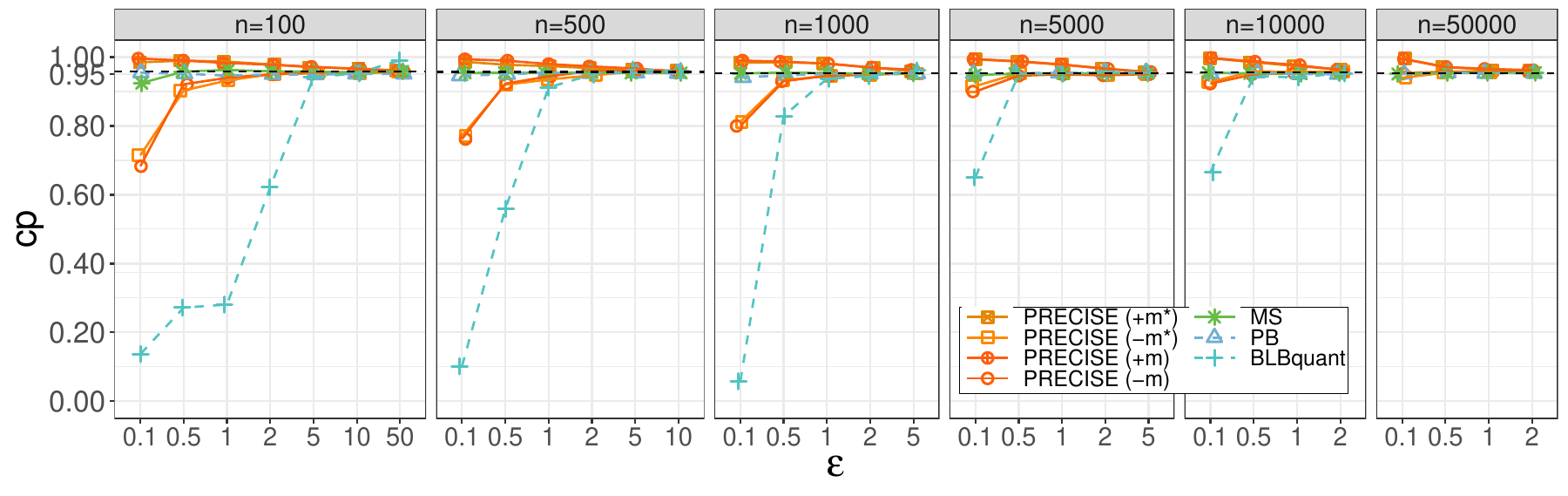}
  \vspace{-4pt}
\captionsetup{justification=raggedright, singlelinecheck=false}
\caption{Comparisons of PPIE width and CP for Poisson mean and Bernoulli proportion PPIE with $\varepsilon$-DP. Black dashed lines represent the original non-private results. The results on $\mu$-GDP are in the appendix.}\label{fig:Bern+Pois} \vspace{-12pt}
\end{figure}

\begin{figure}[!htb]
\centering
  \includegraphics[width = 0.85\textwidth, trim={0.05in 0.1in 0.05in, 0in},clip]{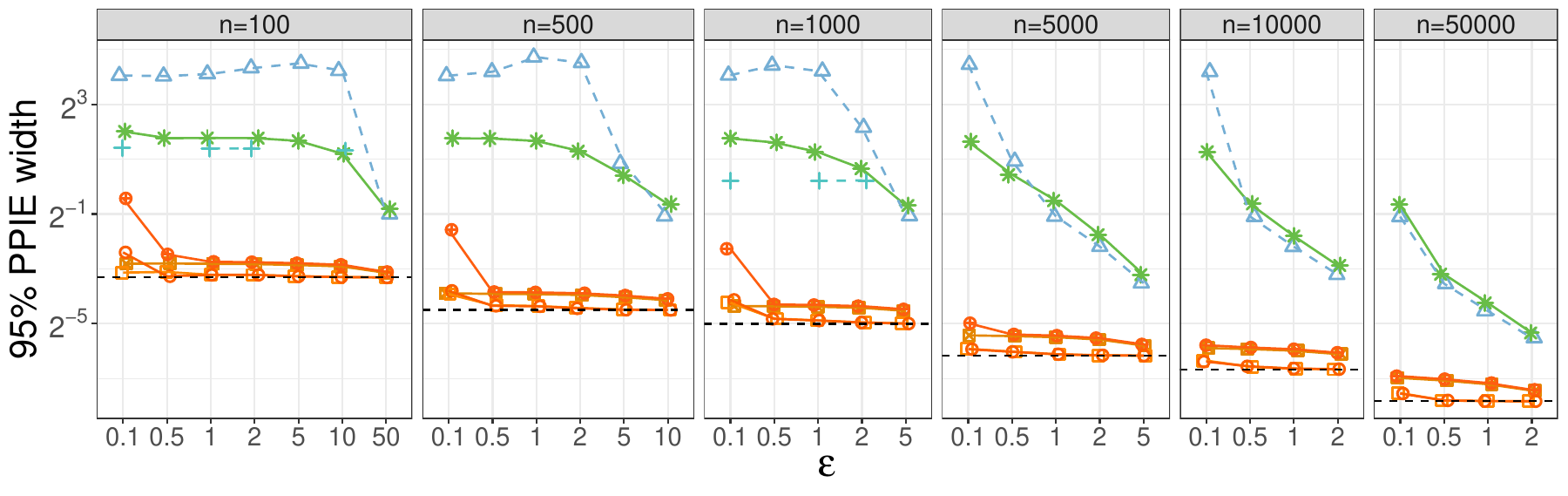}
  \includegraphics[width = 0.85\textwidth, trim={0.05in 0.1in 0.05in, 0in},clip]{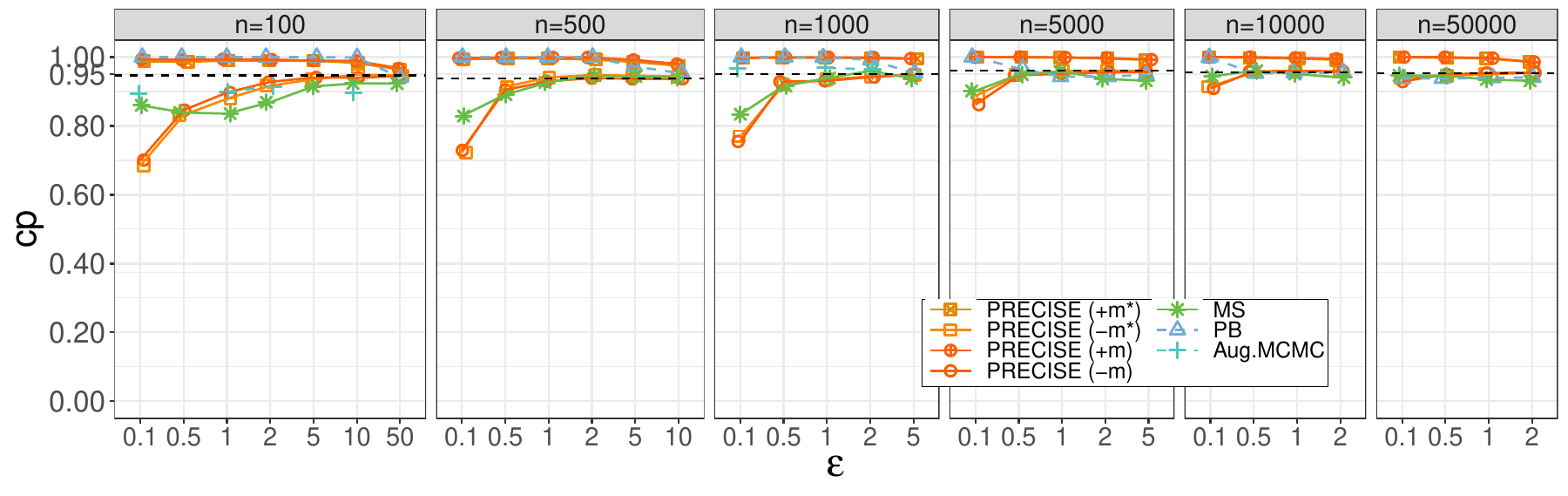}
  
  \captionsetup{justification=raggedright, singlelinecheck=false}
  \caption{Comparisons of PPIE width and CP for the linear regression slope with $\varepsilon$-DP. Black dashed lines represent the original non-private results. The results on $\mu$-DP are in the appendix.}   \vspace{-12pt}
\end{figure}

In summary, \emph{PRECISE ($+m^*$) offers nominal coverage and notably narrower interval estimates} in all $n$ scenarios, data types, inferential tasks, and for both DP types, and outperforms all competing methods examined in the experiments.

Specifically, among the four versions of PRECISE, the two with non-negativity correction  ($+m^*$ and $+m$) achieve the nominal coverage for all $\varepsilon$ and $n$. While $+m$ generates the widest PPIE intervals among the four variants of PRECISE for low privacy loss, they are still the \textit{narrowest} compared to the existing methods.  PRECISE without non-negativity correction  ($-m^*$ and $-m$) exhibits similar performances. Though both have notable under-coverage when $n\varepsilon\leq 500$, they converge quickly to the original as $\varepsilon$ increases. The differences in the results among the four PRECISE versions imply that non-negativity correction have a stronger and more lasting impact on the performance than the intra-consistency correction, but both are important for robust PP inference.

The CP and interval width for all the examined PPIE methods in Figures \ref{fig:Normal} and \ref{fig:Bern+Pois} converge toward the original metrics as $\varepsilon$ or $n$ increases. MS and PB capture the extra variability from the DP sanitization  -- as reflected by their nominal coverage, but they also output wide intervals, especially for small $\varepsilon$. BLBquant suffers from under-coverage when $n$ or $\varepsilon$ is small for every inference task due to not accounting for the sanitization variability in addition to the sampling variability, leading to invalid PPIE.   For the Gaussian mean and variance estimation, the interval widths follow the order of PRECISE($-$) $<$  PRECISE($+m^*$) $<$ PRECISE($+m$) $<$ repro $<$ PB $<$ MS $<$ deconv. For Bernoulli proportion, PB and PRECISE $(-)$ are the best performers when $\varepsilon\geq0.05$ and $n\geq 1000$  and converge to the original faster than PRECISE ($+$) as $\varepsilon$ or $n$ increases. 

For the linear regression in Figure \ref{fig:SLR}, the hybrid PB method \citep{ferrando2022parametric} designed for OLS estimation achieves the nominal coverage at the cost of wide intervals. The  Aug.MCMC method is sensitive to hyperparameter specification and also computationally costly (one MCMC chain of 10,000 iterations took about 1.5 mins for $n\!=\!100$ and 16 mins for $n\!=\!1000$). Aug.MCMC yields under-coverage with reasonably wide intervals.  Even with the help of carefully tuned hyperparameters, it still converges to the original much slower than other methods.

\begin{figure}[!htb]
    \centering
    \hspace{0.2in} \footnotesize{\textbf{Bernoulli}} \hspace{2.3in} \footnotesize{\textbf{Poisson}} \hspace{1.5in}\\
    \vspace{3pt}
    \includegraphics[width=0.42\linewidth]{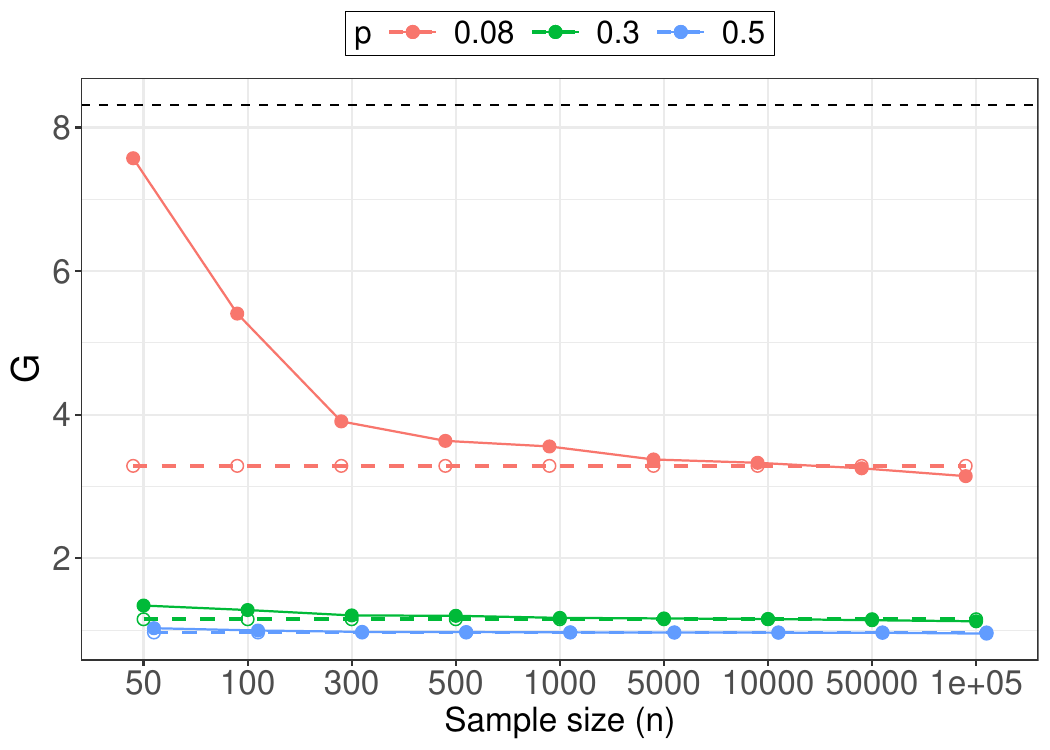}
    \includegraphics[width=0.42\linewidth]{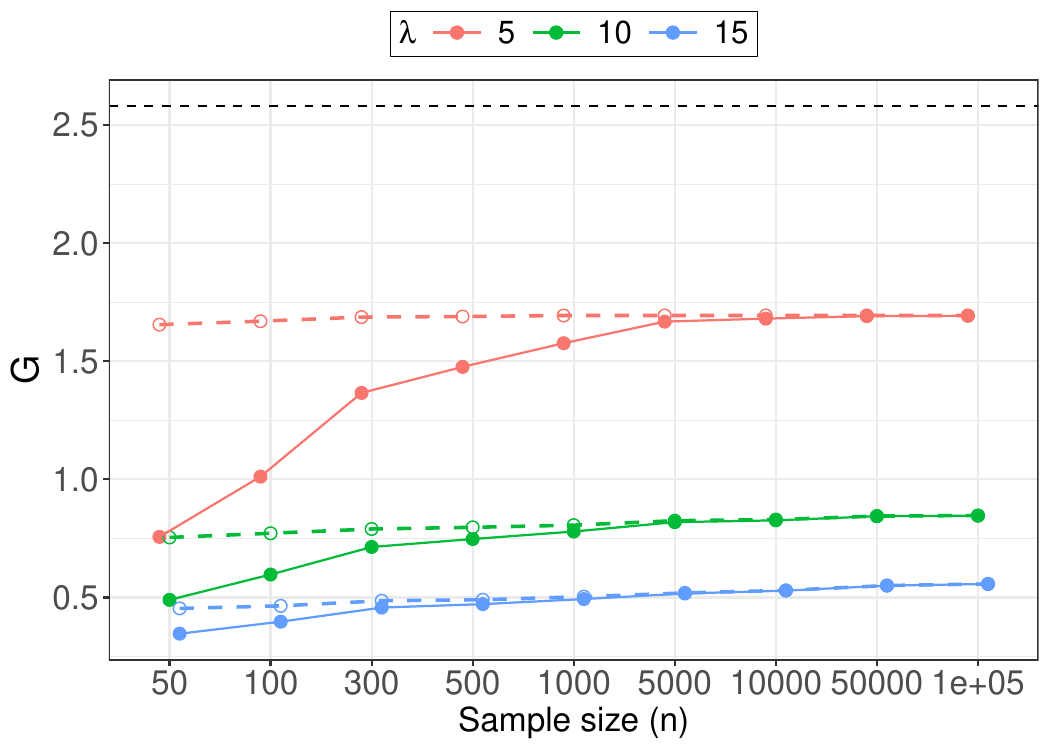}\\
    \hspace{0.4in} \footnotesize{\textbf{Gaussian mean}} \hspace{1.6in} \footnotesize{\textbf{Gaussian variance}} \hspace{1in}\\
    \vspace{3pt}
    \includegraphics[width=0.42\linewidth]{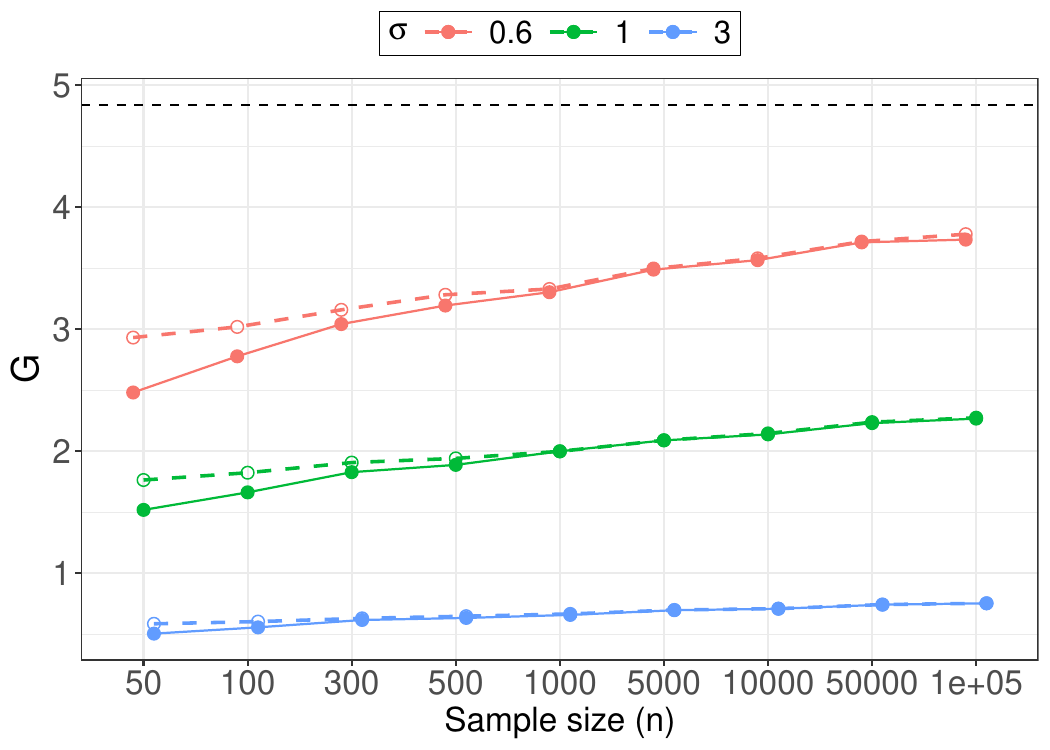}
    \includegraphics[width=0.42\linewidth]{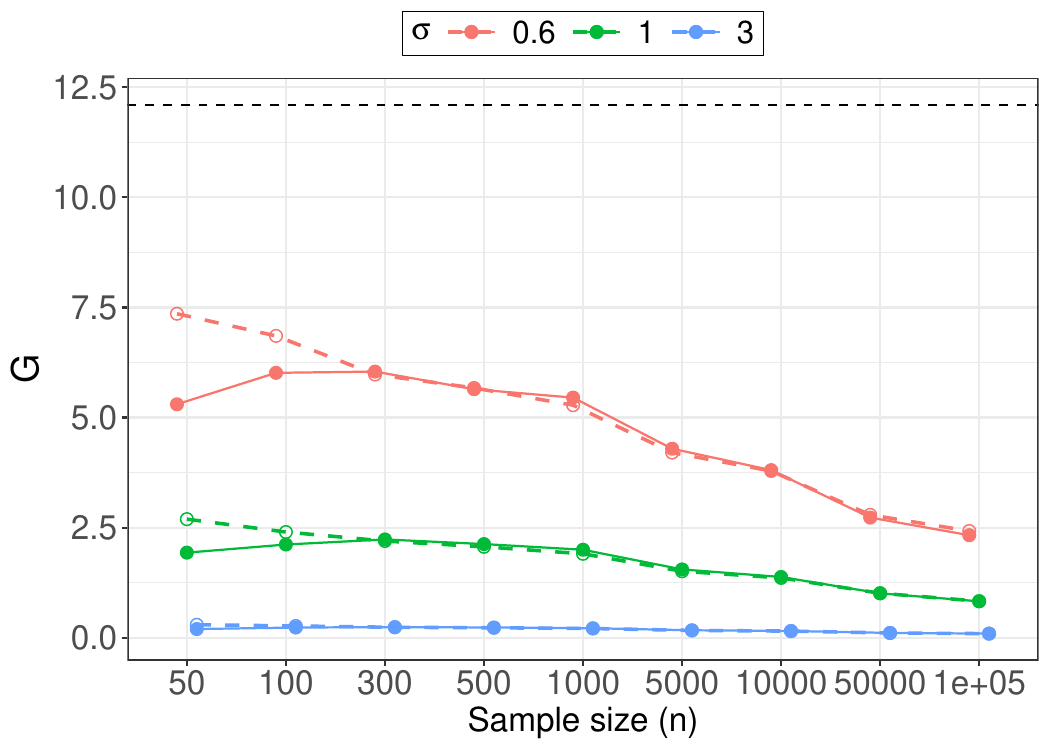}
    \captionsetup{justification=raggedright, singlelinecheck=false}
    \caption{Comparison of numerical approximation to $G(n)$ (solid lines; averaged over 100 runs), analytical approximation to $G_0$ (colored dashed lines) in Theorem \ref{thm:G}, and upper bound $\overline{G_0}$ (black dashed lines) for different true parameter values and sample sizes $n$. }
    \label{fig:G}
\end{figure}

\subsubsection{$G_0$ as a proxy for $G(n)$}\label{sec:G0Gn}
The value of $G(n)$ is used to determine the posterior sample size $m = [2G(n)h]^{-1}$ (Eq.~\ref{eqn:DeltaH1}) in the PRECISE procedure. In the simulation studies, since the true distributions $f(x|\theta)$ are known, we could simulate many pairs of neighboring datasets and perform a grid search over $\theta\in \Theta$ to numerically approximate  $G(n)=\sup_{\theta\in \Theta, d(\x,\x')=1}|f_{\theta|\x}(\theta) \!-\! f_{\theta|\x'}(\theta)|$ for a given $n$. However, since it is infeasible or impossible to exhaust the search space of $d(\x,\x')=1$ especially when $\x$ is high-dimensional or continuous, the numerical $G(n)$ is only approximate at best. Furthermore,  the numerical approach no longer applies in the real world because it may be impossible to define the search space without strong assumptions. 

In practice, $G_0$ in Eq.~\eqref{eqn:G_uni} (Theorem~\ref{thm:G}) can be used as an approximation for a sufficiently large $n$. Since $G_0$ depends on Fisher information that involves unknown parameters, we provide two ways to calculate $G_0$: 1) replace the unknown parameters with their PP estimate; 2) use an upper bound $\overline{G_0}$ for $G_0$.
Though the latter approach is more conservative from a privacy perspective, it not only conserves privacy budget by eliminating the need to sanitize additional parameter estimates using the data but also avoids the extra effort otherwise required to obtain those estimates and to develop and apply a randomized mechanism in the first place.
$\overline{G_0}$ is often informed by prior knowledge of the parameter range $(L,U)$, combined with the global bounds $(L_{\x},U_{\x})$ for data $\x$. In the simulation studies, we adopted the second approach. The hyperparameters used in determining $\overline{G_0}$ are provided in Table \ref{tab:G}; the proofs are provided in Appendix \ref{ape:eg}.

We compare the three methods to obtain $G(n)$  --  numerical approximation,  analytical approximation $G_0$, and upper bound $\overline{G_0}$ -- at different $n$ for various true parameter values in  Figure \ref{fig:G}. The results show that the numerical $G(n)$ converges rapidly to the asymptotic $G_0$ as $n$ increases; the two values are similar even for small $n$ in most cases. These findings provide reassuring evidence that $G_0$ can be reliably used in place of $G(n)$ (with $\overline{G_0}$ serving as a conservative alternative) in practical application of PRECISE.

\begin{table}[!htb]
    \centering
    \caption{Hyperparameters in the $\overline{G_0}$ calculation in the simulation studies}
    \label{tab:G}\vspace{-3pt}
    \setlength{\tabcolsep}{8pt}
    \resizebox{1\textwidth}{!}{
    \begin{tabular}{@{}c@{\hspace{3pt}}|@{\hspace{3pt}}c@{\hspace{3pt}}c@{\hspace{3pt}}c@{\hspace{3pt}}c@{}}
    \toprule
          & $G_0$ from Eq.~\eqref{eqn:G_uni} & $(L,U)$ & $(L_{\x},U_{\x})^\dagger$ & $\overline{G_0}$\\
    \midrule
         Bern($p$) & $\frac{|x'_n-x_n|}{\sqrt{2e\pi} p(1-p)}$ & $(0.03, 0.97)$ & $(0,1)$ &  $\frac{1}{\sqrt{2e\pi}\min\{L(1-L), U(1-U)\}}$\\
         Poisson($\lambda$) & $\frac{|x'_n-x_n|e^{-\frac{1}{2}}}{\sqrt{2e\pi} \lambda}$ & $(3, 35)$ & $(0,35)$ & $\frac{(U_{\x}-L_{\x})}{\sqrt{2e\pi} L}$\\
         $\mu$ in $\mathcal{N}(\mu, \sigma^2)$ &  $\frac{|x'_n-x_n|}{\sqrt{2e\pi} \sigma^2}$ & $(\mu\!-\!k\sigma, \mu\!+\!k\sigma)$ & $(\mu\!-\!k\sigma, \mu\!+\!k\sigma)$ & $\frac{\sqrt{2}k}{\sqrt{e\pi} \sigma}\leq \frac{\sqrt{2}k}{\sqrt{e\pi L_{\sigma}} }^\ddagger$\\
         $\sigma^2$ in $\mathcal{N}(\mu, \sigma^2)$ &  $\frac{|(x_n - \bar{x}_{n-1})^2 - (x_n' - \bar{x}_{n-1})^2|}{\sqrt{2e\pi}\cdot 2\sigma^4}$ & $(0.25, 25)$ & $(\mu\!-\!k\sigma, \mu\!+\!k\sigma)$ & $\frac{k^2}{2\sqrt{2e\pi} \sigma^2}\leq \frac{k^2}{2\sqrt{2e\pi} L}$\\
    \hline
    \multicolumn{5}{p{1.0\linewidth}}{\small $^\dagger$
    conservatively set to satisfy $\Pr(\x\!\notin\!(L_{\x},U_{\x}))< 10^{-5}$; $^\ddagger$: $k=5, L_\sigma = 0.25$.}\\
    \bottomrule
    \end{tabular}}
\end{table}

\subsubsection{Computational cost} 
We summarize the computational time for one run of each method to generate PPIE for the Gaussian mean in Table \ref{tab:time}. PRECISE and MS are very fast.  The computation time for repro and BLBquant increases substantially as $n$ grows, whereas the time for PRECISE, MS,  and PB remains roughly stable with $n$. Additionally, deconv shows a notable increase in time with $n$, though this increase is less pronounced compared to repro and BLBquant.
\begin{table}[!htb]
  \caption{Computational time$^\ddagger$ in one repeat for Gaussian mean PPIE ($\varepsilon=8$).}\label{tab:time}\vspace{-3pt}
  \centering
  \setlength{\tabcolsep}{8pt}
  \resizebox{0.9\textwidth}{!}{
    \begin{tabular}{ccccccc}
    \toprule
    Sample size $n$  & PRECISE & MS & PB  & deconv$^\dagger$ & repro$^\dagger$ &  BLBquant \\
    \midrule
    100 & 0.03 sec& 0.01 sec& 0.10 sec& 4.75 sec & 8.50 sec& 0.04 sec\\
    5000 & 0.05 sec &0.01 sec& 0.31 sec& 10.22 sec&1.44 min& 14.19 sec\\
    50000 & 0.05 sec & 0.01 sec& 2.02 sec& 1.26 min& 17.41 min& 8.11 min\\
    \hline
    \multicolumn{7}{p{1.0\linewidth}}{\footnotesize $^\dagger$ converted to $\mu$-GDP from $(\varepsilon, \delta\!=\!1/n)$-DP; $\mu \!=\! 2.45, 1.91, 1.71$ for $n\!=\!100, 5k, 50k$ respectively.}\\
    \multicolumn{7}{p{1.0\linewidth}}{\footnotesize  $^\ddagger$ On a MacBook Pro with an Apple M3 Mac chip (16‑core CPU, 40‑core GPU) and 64GB of unified memory. All computations were performed on a single CPU thread without GPU acceleration.}\\
    \bottomrule
  \end{tabular}}
\end{table}

\vspace{-3pt}\subsection{\textcolor{black}{Real-world case studies}}\label{sec:case}

We apply PRECISE to two real-world datasets. The first is the UCI adult dataset \citep{becker1996adult} and the goal is to obtain PPIE for the proportion $p$ of an individual's annual income $>50K$ in a randomly sampled subset of size $n=500$.  
The second is the UCI Cardiotocography dataset \citep{cardiotocography_193}, which consists of three fetal state classes (Normal, Suspect, Pathologic) with a sample size of $n=2126$; and the goal is to obtain the PPIE for the proportions of the three classes ($p_1, p_2,  p_3$).  

For the adult data, we use $f(p)=\text{Beta}(1,1)$, $\overline{G_0}= e^{-\frac{1}{2}}/(\sqrt{2\pi}L(1-L))$ with $L=0.03$, $h=2.2\times 10^{-3}$, resulting in $m =269$.
For the Cardiotocography data, we use $\mathbf{p}=(p_1,p_2,p_3)\sim$ Dirichlet$(1,1,1)$ and leverage domain knowledge to choose $L = (0.5, 0.05, 0.02)$ for $p_1,p_2,p_3$ respectively. $\overline{G_0}=e^{-\frac{1}{2}}/(\sqrt{2\pi}L(1-L))$ as we examine each proportion marginally; setting $h=(5, 0.95, 0.39)\times 10^{-4}$ for $p_1,p_2,p_3$, respectively, leads to $m= 1033$ sets of posterior samples drawn from posterior distribution  $\mathbf{p}\sim$Dirichlet$(1+n_1,1+n_2,1+n_3)$, where $n_1, n_2, n_3 = (1655, 295, 176)$ are the observed class counts. 

We run both case studies with $\varepsilon$-DP guarantees at $\varepsilon = 0.1$ and  $0.5$ 100 times to measure the stability of the methods. For the Cardiotocography data, the marginal posterior histogram of each element in $\mathbf{p}$ is sanitized with a privacy budget of $\varepsilon/3$.
The results are presented in Table \ref{tab:case}. For the adult data,
PRECISE $+m^*$ yields tighter and more stable PPIE at $\varepsilon=0.1$, while $-m^*$ performs the best at $\varepsilon=0.5$. For the Cardiotocography data, the relative performance among the  PRECISE versions is $+m^* > -m^* \approx -m > +m$ and $+m^*$ produces the PPIE closest to the original with the lowest SD, highlighting its stability\footnote{We also run MS, PB, BLBQuant in the adult data.  Consistent with the observations in the simulation studies, PRECISE  outperforms MS and PB, with narrower PPIE widths, more stable intervals,  and more accurate point estimates. BLBQuant produces invalid narrow intervals leading to under-coverage.}.
\begin{table}[!htb]
  \caption{Average PPIE width (SD) over 100 runs in two real datasets and an example on private point estimation (95\% PPIE)  from a randomly selected single repeat}\label{tab:case}\vspace{-3pt}
  \centering
  \setlength{\tabcolsep}{4pt}
  \footnotesize{\textbf{(a) Adult dataset}}\\
  \vspace{2pt}
  \resizebox{1\textwidth}{!}{
    \begin{tabular}{c|ccccc}
    \toprule
   $\varepsilon$ & original &\multicolumn{4}{c}{PRECISE} \\
    \cline{3-6}
     && $+m^*$ & $-m^*$ & $+m$ & $-m$\\
    \midrule
   & \multicolumn{5}{c}{average 95\% PPIE widths (SD) over 100 runs ($\times 10^{-2}$)}\\
    \midrule
    0.1&7.21& 9.87 (1.02)& 10.41 (12.50) & 25.04 (18.95) & 9.58 (10.58)\\ 
     0.5&7.21& 9.82 (1.06) & 7.54 (2.02)& 10.85 (3.73) & 7.76 (2.03) \\ 
     \midrule
       & \multicolumn{5}{c}{an example on posterior median (95\% PPIE) from one repeat ($\times 10^{-2}$)}\\
      \midrule
    0.1 & 21.60  (17.99, 25.21)& 22.32 (17.87, 27.37) & 21.34 (19.98, 25.37)  & 18.28 (13.72, 27.66) & 22.23 (19.33, 25.74)\\
 0.5 & 21.60  (17.99, 25.21) & 21.60 (17.93, 26.33) & 21.10 (18.04, 23.31)  & 20.41 (18.02, 26.77) & 21.46 (17.99, 26.82)\\
    \bottomrule
  \end{tabular}}\\
  \vspace{6pt}
  \footnotesize{\textbf{(b) Cardiotocography dataset}}\\
  \vspace{2pt}
  \resizebox{1\textwidth}{!}{
    \begin{tabular}{c|cccccc}
    \toprule
     \multirow{2}{*}{$\varepsilon$} & & \multirow{2}{*}{original} &\multicolumn{4}{c}{PRECISE} \\
    \cline{4-7}
     &&& $+m^*$ & $-m^*$ & $+m$ & $-m$\\
       \midrule
    \multicolumn{7}{c}{average 95\% PPIE widths (SD) over 100 runs ($\times 10^{-2}$)}\\
    \midrule
    \multirow{3}{*}{0.1} & Normal & 3.5 &11.1 (9.7)&6.1 (6.8)&20.3 (11.5)&7.5 (8.4)\\
     & Suspect & 2.9 &4.6 (0.4)&5.5 (12.0)&18.0 (23.0)&7.3 (16.0)\\
     & Pathologic & 2.3 &3.7 (0.3)&5.9 (14.4)&12.9 (19.8)&5.0 (9.6)\\
    \midrule
    \multirow{3}{*}{0.5} & Normal & 3.5 &5.2 (0.5)&4.8 (4.3)&5.6 (2.1)&4.1 (1.2)\\
     & Suspect & 2.9 &4.6 (0.3)&3.3 (0.8)&5.2 (3.2)&3.5 (0.9)\\
     & Pathologic & 2.3 &3.6 (0.3)&2.7 (0.7)&3.9 (0.7)&2.8 (0.7)\\
    \midrule
    \multicolumn{7}{c}{example posterior median (95\% PPIE) from one repeat ($\times 10^{-2}$)}\\
    \midrule
    \multirow{3}{*}{0.1} & Normal & 77.8 (76.1,79.6) &77.4 (74.9, 80.8)&77.2 (75.8, 80.9)&80.4 (50.0, 80.6)&78.7 (75.1, 79.1)\\
     & Suspect & 13.9 (12.4, 15.3) &14.0 (11.6, 16.2)&14.5 (11.6, 16.2)&12.4(5.3, 16.2)&14.0 (11.6, 16.2)\\
     & Pathologic & 8.3 (7.1, 9.4) & 8.6 (6.8, 10.5)&8.3 (7.2, 10.5)&7.2 (6.5, 10.6)&7.3 (6.8, 10.0)\\
    \midrule
    \multirow{3}{*}{0.5} & Normal & 77.8 (76.1, 79.6) &77.8 (75.1, 80.5)&77.4 (75.7, 79.7)&79.6 (75.4, 80.9)&77.5 (75.6, 80.1)\\
     & Suspect & 13.9 (12.4, 15.3) &13.9 (11.5, 16.2)&14.2 (12.9, 16.3)&13.2 (11.3, 16.4)&14.2 (11.9, 14.7)\\
     & Pathologic & 8.3 (7.1, 9.4) &8.3 (6.6, 10.0)&8.4 (7.8, 9.7)&7.3 (6.5, 10.1)&8.2 (7.2, 9.9)\\
    \bottomrule
  \end{tabular}}
\end{table}


\section{Discussion}\label{sec:discuss}
Uncertainty quantification and interval estimation are critical to scientific data interpretation and making robust decisions in the real world. 
This work provides a promising procedure -- PRECISE -- to practitioners who seek to release interval estimates from sensitive datasets without compromising the privacy of individuals who contribute personal information to the datasets. PRECISE is general-purpose and model-agnostic with theoretically proven MSE consistency. Our extensive simulation studies suggested that  PRECISE outperformed all the other examined PPIE methods, offering nominal coverage, notably narrower interval estimates, and fast computation. 

While PRECISE is theoretically applicable to any inferential task that fits in a Bayesian framework and allows for posterior sampling, its practical performance can be influenced by factors such as the dimensionality of the estimation problem, the number of posterior samples, the sample size, and the PRECISE hyperparameters. Addressing these considerations will be the focus of our upcoming work, especially the scalability of PRECISE to high-dimensional settings and its adaptability to more complex inferential tasks, such as regularized regressions and prediction intervals.

In conclusion,  with the theoretically proven statistical validity, guaranteed privacy, along with demonstrated superiority in utility and computation over other PPIE methods, we believe that PRECISE provides a practically promising and effective procedure for releasing interval estimates with DP guarantees in real-world applications, fostering trust in data collection and information sharing across data contributors, curators, and users.

\subsection*{Data and Code}
The data and code in the simulation and case studies are available at [url] (will be open after the paper has been finalized).

\vspace{12pt}

\bibliographystyle{apalike}
\setstretch{1}
\setlength\parskip{0pt}
\bibliography{ref.bib}

\begin{thebibliography}{}

\bibitem[Abowd, 2018]{abowd2018us}
Abowd, J.~M. (2018).
\newblock The us census bureau adopts differential privacy.
\newblock In {\em Proceedings of the 24th ACM SIGKDD international conference on knowledge discovery \& data mining}, pages 2867--2867.

\bibitem[Alabi et~al., 2020]{alabi2020differentially}
Alabi, D., McMillan, A., Sarathy, J., Smith, A., and Vadhan, S. (2020).
\newblock Differentially private simple linear regression.
\newblock {\em Proceedings on 23rd Privacy Enhancing Technologies Symposium}, 2022 (2):184{\textendash}204.

\bibitem[Alabi and Vadhan, 2022]{alabi2022hypothesis}
Alabi, D. and Vadhan, S. (2022).
\newblock Hypothesis testing for differentially private linear regression.
\newblock {\em Advances in Neural Information Processing Systems}, 35:14196--14209.

\bibitem[Amin et~al., 2019]{amin2019differentially}
Amin, K., Dick, T., Kulesza, A., Munoz, A., and Vassilvitskii, S. (2019).
\newblock Differentially private covariance estimation.
\newblock {\em Advances in Neural Information Processing Systems}, 32.

\bibitem[Apple, 2020]{apple}
Apple (2020).
\newblock Apple differential privacy technical overview.
\newblock \url{https://www.apple.com/privacy/docs/Differential_Privacy_Overview.pdf}.

\bibitem[Asi and Duchi, 2020]{asi2020near}
Asi, H. and Duchi, J.~C. (2020).
\newblock Near instance-optimality in differential privacy.
\newblock {\em arXiv preprint arXiv:2005.10630}.

\bibitem[Avella-Medina et~al., 2023]{avella2023differentially}
Avella-Medina, M., Bradshaw, C., and Loh, P.-L. (2023).
\newblock Differentially private inference via noisy optimization.
\newblock {\em The Annals of Statistics}, 51(5):2067--2092.

\bibitem[Awan and Slavkovi{\'c}, 2018]{awan2018differentially}
Awan, J. and Slavkovi{\'c}, A. (2018).
\newblock Differentially private uniformly most powerful tests for binomial data.
\newblock {\em Advances in Neural Information Processing Systems}, 31.

\bibitem[Awan and Wang, 2024]{awan2023simulation}
Awan, J. and Wang, Z. (2024).
\newblock Simulation-based, finite-sample inference for privatized data.
\newblock {\em Journal of the American Statistical Association}, pages 1--14.

\bibitem[Becker and Kohavi, 1996]{becker1996adult}
Becker, B. and Kohavi, R. (1996).
\newblock Adult.
\newblock {\em UCI Machine Learning Repository}, 10:C5XW20.

\bibitem[Bernstein and Sheldon, 2019]{bernstein2019differentially}
Bernstein, G. and Sheldon, D.~R. (2019).
\newblock Differentially private bayesian linear regression.
\newblock {\em Advances in Neural Information Processing Systems}, 32.

\bibitem[Biswas et~al., 2020]{biswas2020coinpress}
Biswas, S., Dong, Y., Kamath, G., and Ullman, J. (2020).
\newblock Coinpress: Practical private mean and covariance estimation.
\newblock {\em Advances in Neural Information Processing Systems}, 33:14475--14485.

\bibitem[Bojkovic and Loh, 2024]{bojkovic2024differentially}
Bojkovic, N. and Loh, P.-L. (2024).
\newblock Differentially private synthetic data with private density estimation.
\newblock {\em arXiv preprint arXiv:2405.04554}.

\bibitem[Bun and Steinke, 2016]{bun2016concentrated}
Bun, M. and Steinke, T. (2016).
\newblock Concentrated differential privacy: Simplifications, extensions, and lower bounds.
\newblock In {\em Theory of Cryptography Conference}, pages 635--658. Springer.

\bibitem[Campos and Bernardes, 2000]{cardiotocography_193}
Campos, D. and Bernardes, J. (2000).
\newblock {Cardiotocography}.
\newblock UCI Machine Learning Repository.
\newblock {DOI}: https://doi.org/10.24432/C51S4N.

\bibitem[Cesar and Rogers, 2021]{pmlr-v132-cesar21a}
Cesar, M. and Rogers, R. (2021).
\newblock Bounding, concentrating, and truncating: Unifying privacy loss composition for data analytics.
\newblock In Feldman, V., Ligett, K., and Sabato, S., editors, {\em Proceedings of the 32nd International Conference on Algorithmic Learning Theory}, volume 132 of {\em Proceedings of Machine Learning Research}, pages 421--457. PMLR.

\bibitem[Chadha et~al., 2024]{chadha2024resampling}
Chadha, K., Duchi, J., and Kuditipudi, R. (2024).
\newblock Resampling methods for private statistical inference.
\newblock {\em arXiv preprint arXiv:2402.07131}.

\bibitem[Chaudhuri et~al., 2011]{chaudhuri2011differentially}
Chaudhuri, K., Monteleoni, C., and Sarwate, A.~D. (2011).
\newblock Differentially private empirical risk minimization.
\newblock {\em Journal of Machine Learning Research}, 12(3).

\bibitem[Covington et~al., 2025]{covington2021unbiased}
Covington, C., He, X., Honaker, J., and Kamath, G. (2025).
\newblock Unbiased statistical estimation and valid confidence intervals under differential privacy.
\newblock {\em Statistica Sinica}, pages 651--670.

\bibitem[Dong et~al., 2020]{dong2020optimal}
Dong, J., Durfee, D., and Rogers, R. (2020).
\newblock Optimal differential privacy composition for exponential mechanisms.
\newblock In {\em International Conference on Machine Learning}, pages 2597--2606. PMLR.

\bibitem[Dong et~al., 2022]{dong2022gaussian}
Dong, J., Roth, A., and Su, W.~J. (2022).
\newblock Gaussian differential privacy.
\newblock {\em Journal of the Royal Statistical Society Series B: Statistical Methodology}, 84(1):3--37.

\bibitem[D'Orazio et~al., 2015]{d2015differential}
D'Orazio, V., Honaker, J., and King, G. (2015).
\newblock Differential privacy for social science inference.
\newblock {\em Sloan Foundation Economics Research Paper}, (2676160).

\bibitem[Drechsler et~al., 2022]{drechsler2022nonparametric}
Drechsler, J., Globus-Harris, I., Mcmillan, A., Sarathy, J., and Smith, A. (2022).
\newblock Nonparametric differentially private confidence intervals for the median.
\newblock {\em Journal of Survey Statistics and Methodology}, 10(3):804--829.

\bibitem[Du et~al., 2020]{du2020differentially}
Du, W., Foot, C., Moniot, M., Bray, A., and Groce, A. (2020).
\newblock Differentially private confidence intervals.
\newblock {\em arXiv preprint arXiv:2001.02285}.

\bibitem[Durfee and Rogers, 2019]{durfee2019practical}
Durfee, D. and Rogers, R.~M. (2019).
\newblock Practical differentially private top-k selection with pay-what-you-get composition.
\newblock {\em Advances in Neural Information Processing Systems}, 32.

\bibitem[Dwork et~al., 2006a]{dwork2006our}
Dwork, C., Kenthapadi, K., McSherry, F., Mironov, I., and Naor, M. (2006a).
\newblock Our data, ourselves: Privacy via distributed noise generation.
\newblock In {\em Advances in Cryptology-EUROCRYPT 2006: 24th Annual International Conference on the Theory and Applications of Cryptographic Techniques, St. Petersburg, Russia, 2006. Proceedings 25}, pages 486--503. Springer.

\bibitem[Dwork and Lei, 2009]{dwork2009differential}
Dwork, C. and Lei, J. (2009).
\newblock Differential privacy and robust statistics.
\newblock In {\em Proceedings of the 41st annual ACM symposium on Theory of computing}, pages 371--380.

\bibitem[Dwork et~al., 2006b]{dwork2006calibrating}
Dwork, C., McSherry, F., Nissim, K., and Smith, A. (2006b).
\newblock Calibrating noise to sensitivity in private data analysis.
\newblock In {\em Theory of Cryptography: Third Theory of Cryptography Conference, TCC 2006, New York, NY, USA, 2006. Proceedings 3}, pages 265--284. Springer.

\bibitem[Erlingsson et~al., 2014]{erlingsson2014rappor}
Erlingsson, {\'U}., Pihur, V., and Korolova, A. (2014).
\newblock Rappor: Randomized aggregatable privacy-preserving ordinal response.
\newblock In {\em Proceedings of the 2014 ACM SIGSAC conference on computer and communications security}, pages 1054--1067.

\bibitem[Evans et~al., 2023]{evans2023statistically}
Evans, G., King, G., Schwenzfeier, M., and Thakurta, A. (2023).
\newblock Statistically valid inferences from privacy-protected data.
\newblock {\em American Political Science Review}, 117(4):1275--1290.

\bibitem[Ferrando et~al., 2022]{ferrando2022parametric}
Ferrando, C., Wang, S., and Sheldon, D. (2022).
\newblock Parametric bootstrap for differentially private confidence intervals.
\newblock In {\em International Conference on Artificial Intelligence and Statistics}, pages 1598--1618. PMLR.

\bibitem[Gaboardi et~al., 2016]{gaboardi2016differentially}
Gaboardi, M., Lim, H., Rogers, R., and Vadhan, S. (2016).
\newblock Differentially private chi-squared hypothesis testing: Goodness of fit and independence testing.
\newblock In {\em International conference on machine learning}, pages 2111--2120. PMLR.

\bibitem[Gillenwater et~al., 2021]{gillenwater2021differentially}
Gillenwater, J., Joseph, M., and Kulesza, A. (2021).
\newblock Differentially private quantiles.
\newblock In {\em International Conference on Machine Learning}, pages 3713--3722. PMLR.

\bibitem[Gopi et~al., 2022]{gopi2022private}
Gopi, S., Lee, Y.~T., and Liu, D. (2022).
\newblock Private convex optimization via exponential mechanism.
\newblock In {\em Conference on Learning Theory}, pages 1948--1989. PMLR.

\bibitem[Ju et~al., 2022]{ju2022data}
Ju, N., Awan, J., Gong, R., and Rao, V. (2022).
\newblock Data augmentation mcmc for bayesian inference from privatized data.
\newblock {\em Advances in neural information processing systems}, 35:12732--12743.

\bibitem[Karwa and Vadhan, 2017]{karwa2017finite}
Karwa, V. and Vadhan, S. (2017).
\newblock Finite sample differentially private confidence intervals.
\newblock {\em arXiv preprint arXiv:1711.03908}.

\bibitem[Kleiner et~al., 2014]{kleiner2014scalable}
Kleiner, A., Talwalkar, A., Sarkar, P., and Jordan, M.~I. (2014).
\newblock A scalable bootstrap for massive data.
\newblock {\em Journal of the Royal Statistical Society Series B: Statistical Methodology}, 76(4):795--816.

\bibitem[Kulkarni et~al., 2021]{kulkarni2021differentially}
Kulkarni, T., J{\"a}lk{\"o}, J., Koskela, A., Kaski, S., and Honkela, A. (2021).
\newblock Differentially private bayesian inference for generalized linear models.
\newblock In {\em International Conference on Machine Learning}, pages 5838--5849. PMLR.

\bibitem[Lin et~al., 2024]{lin2024differentially}
Lin, S., Bun, M., Gaboardi, M., Kolaczyk, E.~D., and Smith, A. (2024).
\newblock Differentially private confidence intervals for proportions under stratified random sampling.
\newblock {\em Electronic Journal of Statistics}, 18(1):1455--1494.

\bibitem[Liu, 2022]{liu2016model}
Liu, F. (2022).
\newblock Model-based differentially private data synthesis and statistical inference in multiply synthetic differentially private data.
\newblock {\em Transactions on Data Privacy}, 15(3):141--175.

\bibitem[McSherry and Talwar, 2007]{mcsherry2007mechanism}
McSherry, F. and Talwar, K. (2007).
\newblock Mechanism design via differential privacy.
\newblock In {\em 48th Annual IEEE Symposium on Foundations of Computer Science}, pages 94--103. IEEE.

\bibitem[Mitzenmacher and Upfal, 2017]{mitzenmacher2017probability}
Mitzenmacher, M. and Upfal, E. (2017).
\newblock {\em Probability and computing: Randomization and probabilistic techniques in algorithms and data analysis}.
\newblock Cambridge university press.

\bibitem[Nagaraja et~al., 2015]{nagaraja2015spacings}
Nagaraja, H.~N., Bharath, K., and Zhang, F. (2015).
\newblock Spacings around an order statistic.
\newblock {\em Annals of the Institute of Statistical Mathematics}, 67(3):515--540.

\bibitem[R{\"a}is{\"a} et~al., 2023]{raisa2023noise}
R{\"a}is{\"a}, O., J{\"a}lk{\"o}, J., Kaski, S., and Honkela, A. (2023).
\newblock Noise-aware statistical inference with differentially private synthetic data.
\newblock In {\em International Conference on Artificial Intelligence and Statistics}, pages 3620--3643. PMLR.

\bibitem[Sheffet, 2017]{sheffet2017differentially}
Sheffet, O. (2017).
\newblock Differentially private ordinary least squares.
\newblock In {\em International Conference on Machine Learning}, pages 3105--3114. PMLR.

\bibitem[Smirnov, 1949]{smirnov1949limit}
Smirnov, N.~V. (1949).
\newblock Limit distributions for the terms of a variational series.
\newblock {\em Trudy Matematicheskogo Instituta imeni VA Steklova}, 25:3--60.

\bibitem[Smith, 2011]{smith2011privacy}
Smith, A. (2011).
\newblock Privacy-preserving statistical estimation with optimal convergence rates.
\newblock In {\em Proceedings of the 43rd ACM symposium on Theory of Computing}, pages 813--822.

\bibitem[Van~der Vaart, 2000]{van2000asymptotic}
Van~der Vaart, A.~W. (2000).
\newblock {\em Asymptotic statistics}, volume~3.
\newblock Cambridge university press.

\bibitem[Walker, 1968]{walker1968note}
Walker, A. (1968).
\newblock A note on the asymptotic distribution of sample quantiles.
\newblock {\em Journal of the Royal Statistical Society Series B: Statistical Methodology}, 30(3):570--575.

\bibitem[Wang et~al., 2019]{Wang_Kifer_Lee_2019}
Wang, Y., Kifer, D., and Lee, J. (2019).
\newblock Differentially private confidence intervals for empirical risk minimization.
\newblock {\em Journal of Privacy and Confidentiality}, 9(1).

\bibitem[Wang, 2018]{wang2018revisiting}
Wang, Y.-X. (2018).
\newblock Revisiting differentially private linear regression: optimal and adaptive prediction \& estimation in unbounded domain.
\newblock {\em arXiv preprint arXiv:1803.02596}.

\bibitem[Wang et~al., 2022]{wang2022differentially}
Wang, Z., Cheng, G., and Awan, J. (2022).
\newblock Differentially private bootstrap: New privacy analysis and inference strategies.
\newblock {\em arXiv preprint arXiv:2210.06140}.

\end{thebibliography}

\newpage
\appendix
\renewcommand{\thefigure}{S.\arabic{figure}}
\renewcommand{\thetable}{S.\arabic{table}}
\renewcommand{\thealgocf}{S.\arabic{algocf}}
\renewcommand{\thelemma}{S.\arabic{lemma}}
\renewcommand{\thetheorem}{S.\arabic{theorem}}
\renewcommand{\thecon}{S.\arabic{condition}}
\setcounter{page}{1}
\setcounter{figure}{0}
\setcounter{table}{0}
\setcounter{algocf}{0}
\setcounter{theorem}{0}
\setcounter{equation}{0}

\setcounter{tocdepth}{0} 
\addtocontents{toc}{\protect\setcounter{tocdepth}{3}} 
\section*{Appendix}
\renewcommand{\contentsname}{}
\vspace{-20pt}
\newcommand{\appendixtoc}{
  \begingroup
  \setlength{\cftbeforesubsecskip}{10pt}
  \tableofcontents
  \endgroup
}
\appendixtoc 
\setlength{\parskip}{6pt}
\setlength{\parindent}{0pt}
\setstretch{1}

\clearpage
\section{Proofs}

\subsection{Proof of Theorem \ref{thm:G}}\label{ape:G_proof}

\subsubsection{A single parameter $\theta$}\label{ape:G_proof_uni}
Let $G \overset{\scalebox{0.6}{$\triangle$}}{=}\sup_{\theta\in \Theta, d(\x,\x')=1}|f_{\theta|\x}(\theta) - f_{\theta|\x'}(\theta)|$, where the parameter $\theta \in \Theta$ is a scalar. By the Bernstein-von Mises theorem \citep{van2000asymptotic}, as $n \rightarrow \infty$,
\begin{align}
||f(\theta |\x)-{\mathcal {N}}({\hat {\theta }}_{n},n^{-1}I^{-1}_{\theta_0})||_{\mathrm {TVD} }= \mathcal{O}(n^{-1/2})
\end{align}

where $\x = \{x_1, x_2, \ldots, x_n\}, \hat{\theta}_n$ is the MAP based on $\x$, and $I_{\theta_0}$ is the Fisher information matrix evaluated at the true population parameter $\theta_0$. Since we assume a non-informative prior $f(\theta)$ relative to the amount of data, we will use MLE and MAP exchangeable in this proof as they converge to the same value asymptotically; the same applies to other proofs if applicable.

Assume the neighboring datasets $\x$ and $\x'$ differ by the last element, and $\hat{\theta}'_n$ denotes the MLE based on $\x'$. Assume $\hat{\theta}'_n - \hat{\theta}_n \approx \frac{C}{n} + o(n^{-1})$ as $n\rightarrow \infty$.
Per the triangle inequality,
\begin{align}
\left| f_{\theta|\x}(\theta) - f_{\theta|\x’}(\theta) \right|
&\leq \left| f_{\theta|\x}(\theta) - \phi(\theta; \hat{\theta}_n, n^{-1} I_{\theta_0}^{-1}) \right|+\left| \phi(\theta; \hat{\theta}_n, n^{-1} I_{\theta_0}^{-1}) - \phi(\theta; \hat{\theta}'_n, n^{-1} I_{\theta_0}^{-1}) \right| \nonumber \\
&\quad + \left| \phi(\theta; \hat{\theta}'_n, n^{-1} I_{\theta_0}^{-1}) - f_{\theta|\x’}(\theta) \right|. \label{eqn:tvd-triangle}
\end{align}
As $n \to \infty$, the first and third terms in Eq.~\eqref{eqn:tvd-triangle} on the right-hand side vanish uniformly in $\theta$ at the rate of $\mathcal{O}(n^{-1/2})$, due to total variation convergence of the posterior to its asymptotic Gaussian approximation. Thus, 

\paragraph{Substitution neighboring relation}

\begin{align}
    |f_{\theta|\x}(\theta) - f_{\theta|\x'}(\theta)| 
    &\rightarrow  \frac{\sqrt{n}}{\sqrt{2\pi I^{-1}_{\theta_0}}}\Bigg|\exp\left(-\frac{n(\theta-\hat{\theta}_n)^2}{2I^{-1}_{\theta_0}}\right)-\exp\left(-\frac{n(\theta-\hat{\theta}'_n)^2}{2I^{-1}_{\theta_0}}\right)\Bigg|\\
    &\overset{\triangle}{=}\frac{\sqrt{n}}{\sqrt{2\pi I^{-1}_{\theta_0}}}|g(\hat{\theta}_n)-g(\hat{\theta}'_n)|.\label{eqn: uni_objective}
\end{align}
Per Taylor expansion of $g(x)$ around $x_0$: $g(x)\approx g(x_0) + g'(x_0)(x-x_0) + \frac{g''(x_0)}{2!}(x-x_0)^2+\cdots$. 
\begin{align}
    &g(\hat{\theta}'_n)- g(\hat{\theta}_n)\notag\\
    \approx\;& \exp\left(-\frac{n(\theta-\hat{\theta}_n)^2}{2I^{-1}_{\theta_0}}\right) \Bigg[\frac{n(\theta-\hat{\theta}_n)}{I^{-1}_{\theta_0}}(\hat{\theta}'_n-\hat{\theta}_n)+\frac{(\hat{\theta}'_n-\hat{\theta}_n)^2}{2}\left(\frac{n^2(\theta-\hat{\theta}_n)^2}{I^{-2}_{\theta_0}}-\frac{n}{I^{-1}_{\theta_0}}\right)\notag\\
    &+\frac{(\hat{\theta}'_n-\hat{\theta}_n)^3}{3!}\left(\frac{n^3(\theta-\hat{\theta}_n)^3}{I^{-3}_{\theta_0}}-\frac{3n^2(\theta-\hat{\theta}_n)}{I^{-2}_{\theta_0}}\right)\Bigg]\\
    \rightarrow\;&g(\hat{\theta}_n)\Bigg[\frac{C(\theta\!-\!\hat{\theta}_n)}{I^{-1}_{\theta_0}} \!+\!\frac{C^2}{2}\left(\frac{(\theta\!-\!\hat{\theta}_n)^2}{I^{-2}_{\theta_0}}\!-\!\frac{1}{nI^{-1}_{\theta_0}}\right) \!+\!\frac{C^3}{3!\!}\left(\frac{(\theta\!-\!\hat{\theta}_n)^3}{I^{-3}_{\theta_0}}\!-\!\frac{3(\theta\!-\!\hat{\theta}_n)}{nI^{-2}_{\theta_0}}\right)\Bigg],\label{eqn:uni_taylor}\\
    &\mbox{ where }C=\hat{\theta}'_n-\hat{\theta}_n.
\end{align}
To obtain $G(n)$, we aim to solve for $\theta$ value that maximizes $|g(\hat{\theta}'_n)- g(\hat{\theta}_n)|$. Toward that end, we
take the first derivative of Eq.~\eqref{eqn:uni_taylor} with respect to $\theta$. 
\begin{align}
    &\frac{\partial(g(\hat{\theta}'_n)- g(\hat{\theta}_n))}{\partial\theta}\notag\\
    =\;&g(\hat{\theta}_n)\left(-\frac{n(\theta-\hat{\theta}_n)}{I^{-1}_{\theta_0}}\right)\Bigg[\frac{C(\theta\!-\!\hat{\theta}_n)}{I^{-1}_{\theta_0}} +\frac{C^2}{2}\left(\frac{(\theta\!-\!\hat{\theta}_n)^2}{I^{-2}_{\theta_0}}\!-\!\frac{1}{nI^{-1}_{\theta_0}}\right)\!+\!\frac{C^3}{3!}\left(\frac{(\theta\!-\!\hat{\theta}_n)^3}{I^{-3}_{\theta_0}}\!-\!\frac{3(\theta\!-\!\hat{\theta}_n)}{nI^{-2}_{\theta_0}}\right)\Bigg]\notag\\
    &+g(\hat{\theta}_n)\Bigg[\frac{C}{I^{-1}_{\theta_0}}+ \frac{C^2(\theta-\hat{\theta}_n)}{I^{-2}_{\theta_0}} +\frac{C^3}{2}\left(\frac{(\theta-\hat{\theta}_n)^2}{I^{-3}_{\theta_0}}-\frac{1}{nI^{-2}_{\theta_0}}\right)\Bigg].\label{eqn:uni_deriv}
\end{align}
WLOG, assume $C \!\geq\! 0$ so $\mathcal{N}(\hat{\theta}'_n, I_{\theta_0}^{-1})$ is shifted to the right of $\mathcal{N}(\hat{\theta}_n, I_{\theta_0}^{-1})$, with a single intersection point $\Tilde{\theta} = \frac{\hat{\theta}_n+\hat{\theta}'_n}{2} \in (\hat{\theta}_n, \hat{\theta}'_n)$; and $g(\hat{\theta}'_n)- g(\hat{\theta}_n)\le 0$  for $\theta \leq \Tilde{\theta}$, and $g(\hat{\theta}'_n)- g(\hat{\theta}_n)\ge 0$ for $\theta \geq \Tilde{\theta}$. Due to symmetry, $|g(\hat{\theta}_n)-g(\hat{\theta}'_n)|$ achieve its maximum at two $\theta$ values; that is, there exists a constant $d\ge 0$ such that $\frac{\partial(g(\hat{\theta}'_n)- g(\hat{\theta}_n))}{\partial\theta}|_{\theta = \hat{\theta}_n-d}=0$ and $\frac{\partial(g(\hat{\theta}'_n)- g(\hat{\theta}_n))}{\partial\theta}|_{\theta = \hat{\theta}'_n+d}=0$. Thus, it suffices to show that $g(\hat{\theta}'_n)-g(\hat{\theta}')$ is unimodal has a unique maximizer when $\theta \leq \Tilde{\theta}$. 
\begin{align*}
    g(\hat{\theta}'_n)- g(\hat{\theta}_n)= g(\hat{\theta}_n)\left(\frac{g(\hat{\theta}'_n)}{g(\hat{\theta}_n)}-1\right) \Rightarrow 
    \log( g(\hat{\theta}'_n)- g(\hat{\theta}_n)) = \log(g(\hat{\theta}_n))+ \log\left(\frac{g(\hat{\theta}'_n)}{g(\hat{\theta}_n)}-1\right).
\end{align*}
If we show $ \log( g(\hat{\theta}'_n)- g(\hat{\theta}_n))$ is concave and has a unique maximizer, the same maximizer applies to $ g(\hat{\theta}'_n)- g(\hat{\theta}_n)$ due the monotonicity of the log transformation. First, 
\begin{align*}
\frac{\partial^2 \log(g(\hat{\theta}_n))}{\partial \theta^2} & = \frac{\partial (-\frac{n(\theta-\hat{\theta}_n)}{I^{-1}_{\theta_0}})}{\partial \theta} = -\frac{n}{I^{-1}_{\theta_0}} < 0; \mbox{then }\\
 z(\theta) &= \log \left(\frac{g(\hat{\theta}'_n)}{g(\hat{\theta}_n)}-1\right) =\log\left( \exp\left(-\frac{n(\theta-\hat{\theta}'_n)^2}{2I^{-1}_{\theta_0}} +\frac{n(\theta-\hat{\theta}_n)^2}{2I^{-1}_{\theta_0}}\right)-1\right),\\
    \frac{\partial z(\theta)}{\partial \theta} &= \frac{-\frac{n(\theta-\hat{\theta}'_n)}{I^{-1}_{\theta_0}}g(\hat{\theta}'_n)g(\hat{\theta}_n) +\frac{n(\theta-\hat{\theta}_n)}{I^{-1}_{\theta_0}}g(\hat{\theta}_n)g(\hat{\theta}'_n)}{g^2(\hat{\theta}_n)\left(\frac{g(\hat{\theta}'_n)}{g(\hat{\theta}_n)}-1\right)} = \frac{g(\hat{\theta}'_n)}{g(\hat{\theta}'_n)-g(\hat{\theta}_n)}\cdot\frac{n(\hat{\theta}'_n-\hat{\theta}_n)}{I^{-1}_{\theta_0}}\\
    \frac{\partial^2 z(\theta)}{\partial \theta^2} &= \frac{n(\hat{\theta}'_n-\hat{\theta}_n)}{I^{-1}_{\theta_0}}\cdot \frac{g(\hat{\theta}'_n)(g(\hat{\theta}'_n)-g(\hat{\theta}_n))\frac{-n(\theta-\hat{\theta}'_n)}{I^{-1}_{\theta_0}}-g(\hat{\theta}'_n)(-\frac{n(\theta-\hat{\theta}'_n)}{I^{-1}_{\theta_0}}g(\hat{\theta}'_n)+\frac{n(\theta-\hat{\theta}_n)}{I^{-1}_{\theta_0}}g(\hat{\theta}_n))}{(g(\hat{\theta}'_n)-g(\hat{\theta}_n))^2} \\
    &=\frac{n(\hat{\theta}'_n-\hat{\theta}_n)}{I^{-1}_{\theta_0}}\cdot \frac{g'(\hat{\theta}_n)\left(\frac{-n(\theta-\hat{\theta}'_n)}{I^{-1}_{\theta_0}}g(\hat{\theta}'_n)+\frac{n(\theta-\hat{\theta}'_n)}{I^{-1}_{\theta_0}}g(\hat{\theta}_n))+\frac{n(\theta-\hat{\theta}'_n)}{I^{-1}_{\theta_0}}g(\hat{\theta}'_n)-\frac{n(\theta-\hat{\theta}_n)}{I^{-1}_{\theta_0}}g(\hat{\theta}_n)\right)}{(g(\hat{\theta}'_n)-g(\hat{\theta}_n))^2} \\
    & = -\left(\frac{n(\hat{\theta}'_n-\hat{\theta}_n)}{I^{-1}_{\theta_0}}\right)^2\cdot \frac{g(\hat{\theta}'_n)g(\hat{\theta}_n)}{(g(\hat{\theta}'_n)-g(\hat{\theta}_n))^2}< 0,
\end{align*}
Therefore, both  $g(\hat{\theta}_n)>0$ and $\frac{g(\hat{\theta}'_n)}{g(\hat{\theta}_n)}-1>0$ are log-concave. Given both $g(\hat{\theta}_n)>0$ and $\frac{g(\hat{\theta}'_n)}{g(\hat{\theta}_n)}-1>0$, per the product of log-concave functions  is also log-concave, $g(\hat{\theta}'_n)- g(\hat{\theta}_n) = g(\hat{\theta}_n)(\frac{g(\hat{\theta}'_n)}{g(\hat{\theta}_n)}-1)$ is also log-concave, thus unimodal for $\theta \geq \Tilde{\theta}$.

Now that we have shown  $g(\hat{\theta}'_n)-g(\hat{\theta}_n)$ has a unique maximum, the next step is to derive the maximizer. Let Eq.~\eqref{eqn:uni_deriv} equal to 0, we have 
\footnotesize
\begin{equation}
    \frac{n(\theta\!-\!\hat{\theta}_n)}{I^{-1}_{\theta_0}}\Bigg[(\theta-\hat{\theta}_n)+\frac{C}{2}\!\!\left(\!\frac{(\theta\!-\!\hat{\theta}_n)^2}{I^{-1}_{\theta_0}}\!-\!\frac{1}{n}\right) \!+\! \frac{C^2}{3!}\!\!\left(\!\frac{(\theta\!-\!\hat{\theta}_n)^3}{I^{-2}_{\theta_0}}\!-\!\frac{3(\theta\!-\!\hat{\theta}_n)}{nI^{-1}_{\theta_0}}\!\right)\!\Bigg]\! \!=\! 1 + \frac{C(\theta\!-\!\hat{\theta}_n)}{I^{-1}_{\theta_0}} +\frac{C^2}{2}\!\left(\!\frac{(\theta\!-\!\hat{\theta}_n)^2}{I^{-2}_{\theta_0}}\!-\!\frac{1}{nI^{-1}_{\theta_0}}\!\right)\notag
\end{equation}
\normalsize
Rearranging the terms, we have
\begin{align}
1-\frac{C^2}{2nI^{-1}_{\theta_0}} = & \frac{n(\theta-\hat{\theta}_n)^2}{I^{-1}_{\theta_0}}-\frac{3(\theta-\hat{\theta}_n)C}{2I^{-1}_{\theta_0}}-\frac{C^2(\theta-\hat{\theta}_n)^2}{I^{-2}_{\theta_0}}+ \frac{nC(\theta-\hat{\theta}_n)^3}{2I^{-2}_{\theta_0}}+ \frac{nC^2(\theta-\hat{\theta}_n)^4}{6I^{-3}_{\theta_0}}.\label{eqn:uni_deriv=0}
\end{align}
Substituting $\hat{\theta}_n-d$ and $\hat{\theta}'_n+d \approx \hat{\theta}_n+\frac{C}{n}+d+o(n^{-1})$ for $\theta$ in Eq.~\eqref{eqn:uni_deriv=0}, its RHS is
\begin{align}
    \text{RHS}|_{\theta = \hat{\theta}_n-d}& = \frac{nd^2}{I^{-1}_{\theta_0}} + \frac{3dC}{2I^{-1}_{\theta_0}}-\frac{C^2d^2}{I^{-2}_{\theta_0}}-\frac{nCd^3}{2I^{-2}_{\theta_0}}+\frac{nd^4C^2}{6I^{-3}_{\theta_0}}\label{eqn:uni_RHS-d}\\
     \text{RHS}|_{\theta = \hat{\theta}'_n+d}& \approx \frac{n(\frac{C}{n}+d+o(n^{-1}))^2}{I^{-1}_{\theta_0}}-\frac{3(\frac{C}{n}+d+o(n^{-1}))C}{2I^{-1}_{\theta_0}}-\frac{C^2(\frac{C}{n}+d+o(n^{-1}))^2}{I^{-2}_{\theta_0}}\notag\\
    &\;\quad +\frac{nC(\frac{C}{n}+d+o(n^{-1}))^3}{2I^{-2}_{\theta_0}}+\frac{n(\frac{C}{n}+d+o(n^{-1}))^4C^2}{6I^{-3}_{\theta_0}},\label{eqn:uni_RHS+d}
\end{align}
respectively. Taking the difference between Eq.~\eqref{eqn:uni_RHS+d} and Eq.~\eqref{eqn:uni_RHS-d} leads to
\begin{align}
    0 & = \text{RHS}|_{\theta = \hat{\theta}'_n+d}-\text{RHS}|_{\theta = \hat{\theta}_n-d}=\;-\frac{(3\frac{C}{n}+6d)C}{2I^{-1}_{\theta_0}} + \frac{C(2\frac{C}{n}+4d)}{2I^{-1}_{\theta_0}}-\frac{C^2((\frac{C}{n}+d)^2-d^2)}{I^{-2}_{\theta_0}}\notag\\
    &\qquad\qquad+\frac{nC((\frac{C}{n}+d)^3+d^3)}{2I^{-2}_{\theta_0}}+\frac{n((\frac{C}{n}+d)^4-d^4)C^2}{6I^{-3}_{\theta_0}}\\
    \Rightarrow&\frac{\cancel{(\frac{C}{n}+2d)}}{2}+\frac{\cancel{(\frac{C}{n}+2d)}\frac{C^2}{n}}{I^{-1}_{\theta_0}}\\
    &=\frac{n\cancel{(\frac{C}{n}+2d)}((\frac{C}{n}+d)^2-(\frac{C}{n}+d)d + d^2)}{2I^{-1}_{\theta_0}}+\frac{((\frac{C}{n}+d)^2+d^2)C^2\cancel{(\frac{C}{n}+2d)}}{6I^{-2}_{\theta_0}}\notag\\
   \Rightarrow & \frac{1}{2} + \frac{C^2}{nI^{-1}_{\theta_0}} = \frac{n((\frac{C}{n}+d)^2-(\frac{C}{n}+d)d + d^2)}{2I^{-1}_{\theta_0}}+\frac{((\frac{C}{n}+d)^2+d^2)C^2}{6I^{-2}_{\theta_0}}\notag\\
    &\qquad \qquad\; =\frac{C^2}{2nI^{-1}_{\theta_0}} +\frac{Cd}{2I^{-1}_{\theta_0}} +\frac{nd^2}{2I^{-1}_{\theta_0}}+\frac{((\frac{C}{n})^2+2\frac{C}{n}d+2d^2)C^2}{6I^{-2}_{\theta_0}}\\
    \Rightarrow& \frac{1}{2} + \frac{C^2}{2nI^{-1}_{\theta_0}} -\frac{C^4}{6n^2I^{-2}_{\theta_0}}= \frac{Cd}{2I^{-1}_{\theta_0}} +\frac{nd^2}{2I^{-1}_{\theta_0}}+\frac{(\frac{C}{n}d+d^2)C^2}{3I^{-2}_{\theta_0}}\\
    \Rightarrow& \frac{1}{2} + \frac{C^2}{2nI^{-1}_{\theta_0}} -\frac{C^4}{6n^2I^{-2}_{\theta_0}}= \left(\frac{C}{2I^{-1}_{\theta_0}}+\frac{C^3}{3nI^{-2}_{\theta_0}}\right)d +\left(\frac{n}{2I^{-1}_{\theta_0}}+\frac{C^2}{3I^{-2}_{\theta_0}}\right)d^2\notag\\
    \Rightarrow& \underbrace{I^{-1}_{\theta_0} + \frac{C^2}{n} -\frac{C^4}{3n^2I^{-1}_{\theta_0}}}_{=-c}= \underbrace{\left(C+\frac{2C^3}{3nI^{-1}_{\theta_0}}\right)}_{=b}d +\underbrace{\left(n+\frac{2C^2}{3I^{-1}_{\theta_0}}\right)}_{=a}d^2\notag\\
    \Rightarrow & ad^2+bd+c=0.
\end{align}
Given the quadratic equation with respect to $d$, its roots  can be obtained analytically
\begin{align}
    &\Delta = b^2-4ac=\left(\!C+\frac{2C^3}{3nI^{-1}_{\theta_0}}\!\right)^2\!\!+4\left(\!n+\frac{2C^2}{3I^{-1}_{\theta_0}}\!\right)\left(I^{-1}_{\theta_0} + \frac{C^2}{n} -\frac{C^4}{3n^2I^{-1}_{\theta_0}}\right)\notag\\
    =\;& 4I^{-1}_{\theta_0}n + \left(1+4+\frac{8}{3}\right)C^2 + \frac{C^4}{nI_{\theta_0}^{-1}}\left(\frac{4}{3}-\frac{4}{3}+\frac{8}{3}\right)+ \frac{C^6}{n^2I_{\theta_0}^{-2}}\left(\frac{4}{9}-\frac{8}{9}\right) \notag\\
    =\; & 4I^{-1}_{\theta_0}n + \frac{23}{3}C^2 + \frac{8}{3}\cdot \frac{C^4}{nI_{\theta_0}^{-1}}-\frac{4}{9}\cdot \frac{C^6}{n^2I_{\theta_0}^{-2}}\\
    d =\; & \frac{-b\pm\sqrt{\Delta}}{2a}=\frac{-C -\frac{2C^3}{3nI^{-1}_{\theta_0}} \pm \sqrt{4I^{-1}_{\theta_0}n + \frac{23}{3}C^2 + \frac{8}{3}\cdot \frac{C^4}{nI_{\theta_0}^{-1}}-\frac{4}{9}\cdot \frac{C^6}{n^2I_{\theta_0}^{-2}}}}{2(n+\frac{2C^2}{3I^{-1}_{\theta_0}})}\notag\\
    \approx\; & \frac{-C}{2n}\pm\sqrt{\frac{I^{-1}_{\theta_0}}{n}}\asymp n^{-1/2}.\label{eqn:uni_d}
\end{align}
Plugging $d$ from Eq.~\eqref{eqn:uni_d} into Eq.~\eqref{eqn:uni_taylor}, we have
\begin{align}
    &g(\hat{\theta}'_n)- g(\hat{\theta}_n)|_{\theta = \hat{\theta}_n-d}\notag\\
    \approx& \exp\left(-\frac{nd^2}{2I^{-1}_{\theta_0}}\right)\!\Bigg[\frac{-dC}{I^{-1}_{\theta_0}} \!+\!\frac{C^2}{2}\left(\frac{d^2}{I^{-2}_{\theta_0}}\!-\!\frac{1}{nI^{-1}_{\theta_0}}\right) \!+\!\frac{C^3}{3!\!}\left(\frac{-d^3}{I^{-3}_{\theta_0}}\!+\!\frac{3d}{nI^{-2}_{\theta_0}}\right)\Bigg]\\
    = & \exp\left(-\frac{n(\frac{I^{-1}_{\theta_0}}{n}+\frac{C^2}{4n^2}-\frac{C}{n}\sqrt{\frac{I^{-1}_{\theta_0}}{n}})}{2I^{-1}_{\theta_0}}\right)\Bigg[\frac{(\frac{C}{2n}-\sqrt{\frac{I^{-1}_{\theta_0}}{n}})C}{I^{-1}_{\theta_0}} + \mathcal{O}(n^{-1})\Bigg]\\
    = & \exp\left(-\frac{1}{2}+\mathcal{O}(n^{-1/2})\right)\Bigg[\frac{-C}{\sqrt{nI^{-1}_{\theta_0}}}+\mathcal{O}(n^{-1})\Bigg]\label{eqn:uni_taylor_plug}
\end{align}
Finally, $G(n)$ can be derived by plugging Eq.~\eqref{eqn:uni_taylor_plug} into Eq.~\eqref{eqn: uni_objective}
\begin{align}
    G(n) &= \frac{\sqrt{n}}{\sqrt{2\pi I^{-1}_{\theta_0}}}|g(\hat{\theta}_n)-g(\hat{\theta}'_n)|\Big|_{\theta = \hat{\theta}_n-d}\\
    &\approx \frac{\sqrt{n}}{\sqrt{2\pi I^{-1}_{\theta_0}}}e^{-\frac{1}{2}+\mathcal{O}(n^{-\frac{1}{2}})}\Bigg[\frac{C}{\sqrt{nI^{-1}_{\theta_0}}}+\mathcal{O}(n^{-1})\Bigg] = \frac{Ce^{-\frac{1}{2}+\mathcal{O}(n^{-\frac{1}{2}})}}{\sqrt{2\pi }I^{-1}_{\theta_0}}+ \mathcal{O}(n^{-\frac{1}{2}})\label{eqn:uni_G_res}
\end{align}

\paragraph{Removal neighboring Relation} Our work so far suggested that a generic formulation on $G(n)$ is analytically challenging to derive in the case of removal neighboring relation. We will continue to investigate this problem in the future.

\subsubsection{Specific Cases}\label{ape:eg}
In this section, we derive $C$  for some specific cases, including the cases in the simulation and case studies. $G(n)$ can be obtained by plugging in $C$ into  Eq.\eqref{eqn:uni_G_res}. Unless mentioned otherwise, all neighboring relations are assumed to be substitution. 

1. If $\hat{\theta}_n = \bar{x}$, then $C = |x'_n-x_n|$. Note that this applies to cateogrical data; for example, for binary data, where $\hat{\theta}_n$ is the proportion of a level, $x\in\{0,1\}$, then $C=1$

2. If $\hat{\theta}_n = n^{-1}\sum_{i=1}^n(x_i-\bar{x})^2$, then $ \hat{\sigma}'^{2}-\hat{\sigma}^2$ 
\vspace{-9pt}
\begin{align*}
& =\textstyle  n^{-1} \left( \sum_{i=1}^n (x_i^2 - x_i'^2) - n \left( \bar{x}^2 - \bar{x}'^2 \right) \right) \\
    &= \textstyle n^{-1} \left( x_n^2 - x_n'^2 - n^{-1}(x_n - x_n')(2 \sum_{i=1}^{n-1} x_i + x_n + x_n') \right) \\
    &= \textstyle n^{-1} (x_n - x_n') (x_n + x_n' - n^{-1} (2 \sum_{i=1}^{n-1} x_i + x_n + x_n')) \\
    &= \textstyle n^{-1} (x_n - x_n') \left( (1 - n^{-1}) x_n + (1 - n^{-1}) x_n' - 2n^{-1} \sum_{i=1}^{n-1} x_i \right) \\
    &= \textstyle n^{-1}(1 - n^{-1}) (x_n^2 - x_n'^2) - 2n^{-2}(n-1) (x_n - x_n') \bar{x}_{n-1}, \textstyle \text{ where } \bar{x}_{n-1} = (n-1)^{-1} \sum_{i=1}^{n-1} x_i \\
    &=\frac{n-1}{n^2}\left(x_n^2 - x_n'^2-2 (x_n - x_n') \bar{x}_{n-1}\right)= \frac{n-1}{n^2}\left( (x_n - \bar{x}_{n-1})^2 - (x_n' - \bar{x}_{n-1})^2 \right),
\end{align*}
leading to  $C = (x_n - \bar{x}_{n-1})^2 - (x_n' - \bar{x}_{n-1})^2$

3. For linear regression $\mathbf{y} = \x\bs\beta + \varepsilon$
\begin{align*}
    f(\sigma^2, \bs\beta | \x, \mathbf{y}) & = f(\sigma^2 | \x, \mathbf{y}) f(\bs\beta |\sigma^2,  \x, \mathbf{y}), \text{ where}\\
   f(\sigma^2 | \x, \mathbf{y}) &= \mbox{IG}\left(\frac{n-(p+1)}{2}, \frac{(\mathbf{y}- \x\widehat{\bs\beta})^{\top}(\mathbf{y}- \x\widehat{\bs\beta})}{2}\right)\\
    f(\bs\beta |\sigma^2,  \x, \mathbf{y}) &= \mathcal{N}_{p+1}(\widehat{\bs\beta}, \bs\Sigma), \text{ where } \widehat{\bs\beta} = (\x^{\top}\x)^{-1}(\x^{\top}\mathbf{y}) \text{ and }\bs\Sigma = \sigma^2(\x^{\top}\x)^{-1}.
\end{align*}
In the case of simple linear regression, the marginal posterior distributions of $\beta_0$ and $\beta_1$ are
\begin{align*}
    \beta_1|\x, \mathbf{y} &\sim t_{n-2}\left(\hat{\beta}_1, \frac{\hat{\sigma}^2}{\sum_{i=1}^n(x_i-\bar{x})^2}\right) \text{ where } \hat{\sigma}^2 = \frac{\sum_{i=1}^n(y_i-\hat{y}_i)^2}{n-2},\\
    \beta_0|\x, \mathbf{y} &\sim t_{n-2}\left(\hat{\beta}_0, \hat{\sigma}^2\left(\frac{1}{n}+\frac{\bar{x}^2}{\sum_{i=1}^n(x_i-\bar{x})^2}\right)\right) \text{ where } \hat{\sigma}^2 = \frac{\sum_{i=1}^n(y_i-\hat{y}_i)^2}{n-2}.
\end{align*}
We first derive the constant $C$ for $\beta_1$. Note that $\hat{\beta}_1 = \frac{S_{xy}}{S_{xx}}$, and     $$n(\hat{\beta}'_1-\hat{\beta}_1) = n\frac{S_{x'y'}S_{xx} - S_{xy}S_{x'x'}}{S_{x'x'}S_{xx}} = n\frac{S_{xx}(S_{x'y'}-S_{xy}) - S_{xy}(S_{x'x'}-S_{xx})}{S_{x'x'}S_{xx}},\mbox{ where}$$ 
\begin{align*}
    S_{x'x'}-S_{xx} &= \left(\sum_{i=1}^{n-1} x_i^2+x'^{2}_n - n\bar{x}'^{2}\right) -\left(\sum_{i=1}^{n-1} x_i^2+x^2_n - n\bar{x}^2\right)= x'^{2}_n - x^2_n - n(\bar{x}'^{2} - \bar{x}^2)\\
    &= x'^{2}_n - x^2_n - n\left(\left(\frac{n-1}{n}\bar{x}_{n-1}+\frac
    {x'_n}{n}\right)^2- \left(\frac{n-1}{n}\bar{x}_{n-1}+\frac
    {x_n}{n}\right)^2\right)\\
    &=x'^{2}_n - x^2_n - (x'_n-x_n)\left(\frac{2(n-1)}{n}\bar{x}_{n-1}+\frac
    {x'_n+x_n}{n}\right) = \frac{n-1}{n}\left(x'^{2}_n - x^2_n-2\bar{x}_{n-1}(x'_n-x_n)\right)\\
    S_{x'y'} - S_{xy} 
    &=x'_ny'_n-n\bar{x}'\bar{y}'-x_ny_n+n\bar{x}\bar{y}\\
    &=x'_ny'_n-x_ny_n-n\Bigg[\left(\frac{n-1}{n}\bar{x}_{n-1}+\frac
    {x'_n}{n}\right)\left(\frac{n-1}{n}\bar{y}_{n-1}+\frac
    {y'_n}{n}\right)\\
    &\qquad - \left(\frac{n-1}{n}\bar{x}_{n-1}+\frac
    {x_n}{n}\right)\left(\frac{n-1}{n}\bar{y}_{n-1}+\frac
    {y_n}{n}\right)\Bigg]\\
    &=x'_ny'_n-x_ny_n-n\left(\frac{n-1}{n^2}\bar{x}_{n-1}(y'_n-y_n)+ \frac{n-1}{n^2}\bar{y}_{n-1}(x'_n-x_n)+\frac{x'_ny'_n-x_ny_n}{n^2}\right)\\
    &=\frac{n-1}{n}(x'_ny'_n-x_ny_n-\bar{x}_{n-1}(y'_n-y_n)-\bar{y}_{n-1}(x'_n-x_n)).
\end{align*}
\begin{align*}
& \text{Thus } n(\hat{\beta}'_1-\hat{\beta}_1) = n\frac{S_{x'y'}S_{xx} - S_{xy}S_{x'x'}}{S_{x'x'}S_{xx}} = n\frac{S_{xx}(S_{x'y'}-S_{xy}) - S_{xy}(S_{x'x'}-S_{xx})}{S_{x'x'}S_{xx}}\\
    =&\frac{n-1}{S_{x'x'}S_{xx}}\Big[S_{xx}(x'_ny'_n-x_ny_n-\bar{x}_{n-1}(y'_n-y_n)-\bar{y}_{n-1}(x'_n-x_n))-S_{xy}\left(x'^{2}_n - x^2_n-2\bar{x}_{n-1}(x'_n-x_n)\right)\Big]\\
    =&\underbrace{\frac{n-1}{S_{x'x'}}}_{\rightarrow (\sigma^2)^{-1}}\Big[\underbrace{x'_ny'_n-x_ny_n-\bar{x}_{n-1}(y'_n-y_n)-(x'_n-x_n)\bar{y}_{n-1}}_{=A}-\hat{\beta}_1(x'_n-x_n)(x'_n + x_n-2\bar{x}_{n-1})\Big]
    = C_1.
    \end{align*}
Together with $I_{\bs\theta} = \begin{pmatrix}
\frac{1}{\sigma^2} & \frac{x_i}{\sigma^2} & 0 \\
\frac{x_i}{\sigma^2} & \frac{x_i^2}{\sigma^2} & 0 \\
0 & 0 & \frac{1}{2\sigma^4}
\end{pmatrix},$
    $$G(n) \approx\frac{|C_1|e^{-\frac{1}{2}+\mathcal{O}(n^{-\frac{1}{2}})}}{\sqrt{2\pi }I^{-1}_{\theta_0}}+ \mathcal{O}(n^{-\frac{1}{2}}) \rightarrow\frac{|C_1|e^{-\frac{1}{2}}}{\sqrt{2\pi }\frac{\sum_{i=1}^n x_i^2}{n\sigma^2}}
    =\frac{|A-\hat{\beta}_1(x'_n-x_n)(x'_n + x_n-2\bar{x}_{n-1})|\sigma^2e^{-\frac{1}{2}}}{\sqrt{2\pi }(\bar{x}^2 +\frac{S_{xx}}{n})(\frac{S_{xx}}{n-1}+\frac{A}{n})}.$$
For example, in the simulation study, $\y=\! \beta_0 + \beta_1\x + \mathcal{N}(0,\sigma^2\!=\!0.25^2)$ with $\beta_0 \!=\! 1, \beta_1 \!=\! 0.5$ and $\x \!\sim \!\mathcal{N}(0,1)$. Then  $ \frac{\sum_{i=1}^n x_i^2}{n} = \frac{\sum_{i=1}^n x_i^2 - n\bar{x}^2+n\bar{x}^2}{n} = \bar{x}^2 +\frac{S_{xx}}{n}\rightarrow 1$,  
$A \rightarrow x'_ny'_n-x_ny_n-(x'_n-x_n)\bar{y}_{n-1} \mbox{ as } n\to\infty$, and 
    \begin{align*}
    G(n) &\rightarrow\frac{|x'_ny'_n-x_ny_n-(x'_n-x_n)(\bar{y}_{n-1} + \hat{\beta}_1x'_n + \hat{\beta}_1x_n)|\sigma^2e^{-\frac{1}{2}}}{\sqrt{2\pi }\sigma_x^2(\sigma_x^2+\frac{A}{n})}\\
    &\leq \frac{\sigma^2e^{-\frac{1}{2}}}{\sqrt{2\pi}\sigma_x^4} \cdot (|x'_ny'_n|+|x_ny_n|+|(x'_n-x_n)(\bar{y}_{n-1} + \hat{\beta}_1x'_n + \hat{\beta}_1x_n)|)
\end{align*}

In the case of $\beta_0= \bar{y} - \bar{x}\hat{\beta}_1$,
\begin{align*}
    \hat{\beta}'_0-\hat{\beta}_0& = \frac{y'_n-y_n}{n} - \bar{x}'\hat{\beta}'_1 + \bar{x}'\hat{\beta}_1-\bar{x}'\hat{\beta}_1 + \bar{x}\hat{\beta}_1= \frac{(y'_n-y_n)-\hat{\beta}_1(x'_n-x_n)}{n} - \bar{x}'(\hat{\beta}'_1 -\hat{\beta}_1)\\
    n(\hat{\beta}'_0-\hat{\beta}_0)&= (y'_n-y_n)-\hat{\beta}_1(x'_n-x_n)-n\bar{x}'(\hat{\beta}'_1 -\hat{\beta}_1)\rightarrow (y'_n-y_n)-\hat{\beta}_1(x'_n-x_n) - \bar{x}'C_1 = C_0
\end{align*}


\subsubsection{Multi-dimensional $\bs\theta$}\label{ape:G_proof_multi}
\begin{figure}[!htb]
    \centering
    \includegraphics[width=0.48\linewidth]{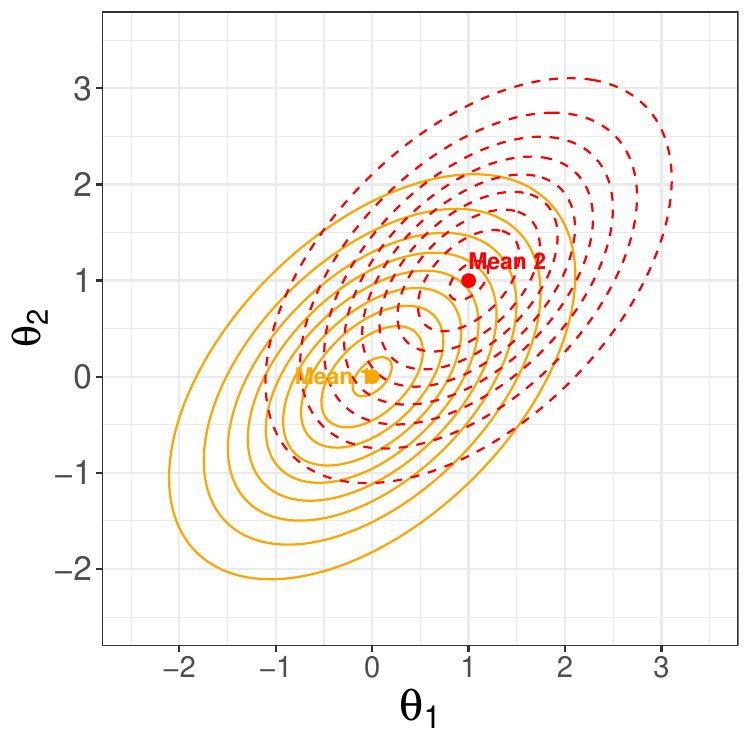}
    \includegraphics[width=0.48\linewidth]{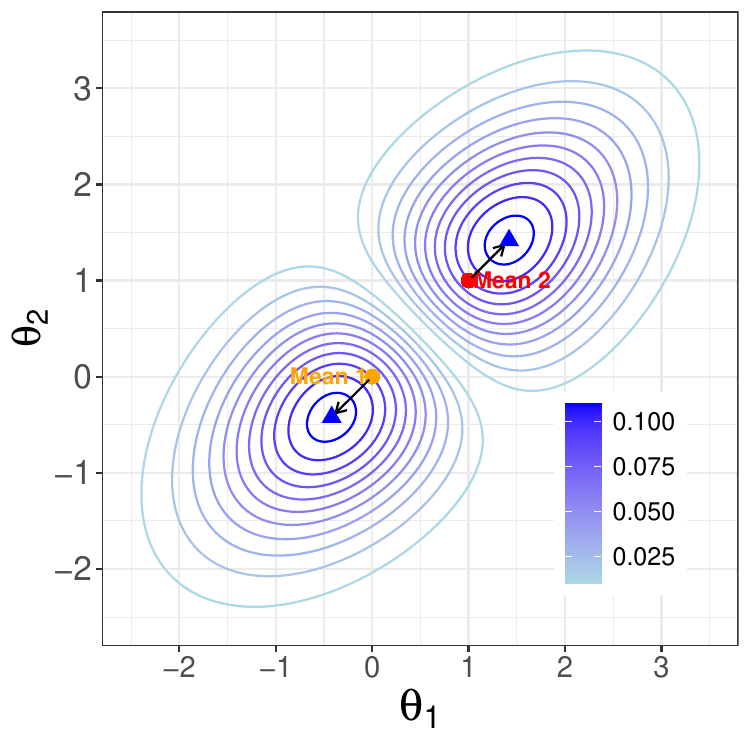}
    \captionsetup{justification=raggedright, singlelinecheck=false}
    \caption{Contour plots for the densities of two bivariate Gaussian distributions with $\bs\mu_1 = (0,0)^{\top}$ and $\bs\mu_2 = (1,1)^{\top}$ and the same covariance matrix (left) and the absolute difference between these two densities (right), where the two black vectors are identical in magnitude but in opposite directions. }
    \label{fig:bivariate_G}
\end{figure}

Let $\bs\theta = (\theta_1, \theta_2, \ldots, \theta_p)^{\top} \in \bs\Theta$ be a $p$-dimensional parameter vector. Denote the dataset by $\X = \{\x_i\}_{i=1}^n$ that contain data points on $n$ individuals, where $\x_i \in \mathbb{R}^q$. 
Per the Bernstein-von Mises theorem, as $n \rightarrow \infty$,
\begin{align}
||f(\bs\theta |\X)-{\mathcal {N}_p}({\widehat{\bs\theta}}_{n},n^{-1}I^{-1}_{\bs\theta_0})||_{\mathrm {TVD} }= \mathcal{O}(n^{-1/2}),
\end{align}
where $\widehat{\bs\theta}_n$ is the MLE based on $\X$, and $I_{\bs\theta_0}$ is the Fisher information matrix evaluated at the true population parameter $\bs\theta_0$. 

For the substitution neighboring relation, WLOG, assume datasets $\X$ and $\X'$ differ in the last observation $\x_n$ vs $\x'_n$. Let $\widehat{\bs\theta}_n'$ denote the MLE based on $\X'$. Assume $\widehat{\bs\theta}'_n - \widehat{\bs\theta}_n \approx \frac{\mathbf{C}}{n} + o(n^{-1})$ as $n\rightarrow \infty$, where $\mathbf{C}\in \mathbb{R}^p$. Then
\begin{align}
    &|f_{\bs\theta|\X}(\bs\theta) \!-\! f_{\bs\theta|\X'}(\bs\theta)| \notag\\
    =\;& (2\pi)^{-\frac{p}{2}} \sqrt{\frac{1}{\det(I^{-1}_{\bs\theta_0}/n)}}\cdot\Bigg|\exp\left(-\frac{n}{2}(\bs\theta - \widehat{\bs\theta}_n)^{\top}I_{\bs\theta_0}(\bs\theta - \widehat{\bs\theta}_n)\right)-\exp\left(-\frac{n}{2}(\bs\theta - \widehat{\bs\theta}_n')^{\top}I_{\bs\theta_0}(\bs\theta \!-\! \widehat{\bs\theta}_n')\right)\Bigg|\notag\\
    \overset{\triangle}{=}\;& (2\pi)^{-\frac{p}{2}} \sqrt{\frac{n^p}{\det(I^{-1}_{\bs\theta_0})}}\cdot \Big|g(\widehat{\bs\theta}_n)-g(\widehat{\bs\theta}_n')\Big|.\label{eqn: multi_objective}
\end{align}

Apply the Taylor expansion to $g(\x)$ around $\x_0$,
\[g(\x) \approx g(\x_0) + \nabla g(\x_0)^{\top} (\x-\x_0) + \frac{1}{2}(\x-\x_0)^{\top}\nabla^2 g(\x_0) (\x-\x_0),\]
where the gradient $\nabla g(\x)$ and Hessian matrix $\nabla^2 g(\x)$ are 
\begin{align}
    \nabla g(\x) &= \frac{\partial}{\partial \x}\exp\left(-\frac{n}{2}(\bs\theta - \x)^{\top}I_{\bs\theta_0}(\bs\theta - \x)\right)= g(\x) \cdot \left(nI_{\bs\theta_0}(\bs\theta - \x)\right)\\
    \nabla^2 g(\x) &=\frac{\partial}{\partial \x}g(\x) \cdot (nI_{\bs\theta_0}\left(\bs\theta - \x)\right)= n^2 g(\x) \cdot \left(I_{\bs\theta_0}(\bs\theta - \x)(\bs\theta - \x)^{\top}I_{\bs\theta_0}\right) - g(\x)nI_{\bs\theta_0}\notag\\
    &= g(\x)\left(n^2\left(I_{\bs\theta_0}(\bs\theta - \x)(\bs\theta - \x)^{\top}I_{\bs\theta_0}\right) -nI_{\bs\theta_0}\right).
\end{align}
Substituting $\widehat{\bs\theta}_n$ and $\widehat{\bs\theta}_n'$ for $\x$ and $\x_0$, respectively, we have
\begin{align}
    g(\widehat{\bs\theta}_n) &\approx g(\widehat{\bs\theta}_n') + g(\widehat{\bs\theta}_n')\Big[n(\bs\theta - \widehat{\bs\theta}_n')^{\top}I_{\bs\theta_0}(\widehat{\bs\theta}_n - \widehat{\bs\theta}_n')\notag\\
    &\qquad\qquad\qquad\quad + \frac{1}{2}(\widehat{\bs\theta}_n-\widehat{\bs\theta}_n')^{\top}\left(n^2\left(I_{\bs\theta_0}(\bs\theta - \widehat{\bs\theta}_n')(\bs\theta - \widehat{\bs\theta}_n')^{\top}I_{\bs\theta_0}\right) -nI_{\bs\theta_0}\right)(\widehat{\bs\theta}_n - \widehat{\bs\theta}_n')\Big]\notag\\
    &\approx g(\widehat{\bs\theta}_n')\Big[1-(\bs\theta - \widehat{\bs\theta}_n')^{\top}I_{\bs\theta_0}\mathbf{C}+ \frac{1}{2n}\mathbf{C}^{\top}\left(n\left(I_{\bs\theta_0}(\bs\theta - \widehat{\bs\theta}_n')(\bs\theta - \widehat{\bs\theta}_n')^{\top}I_{\bs\theta_0}\right) -I_{\bs\theta_0}\right)\mathbf{C}\Big].\label{eqn:multi_taylor}
\end{align}

WLOG, assume $C_j \geq 0$ for $\forall 1\leq j \leq p$ so $\mathcal{N}(\widehat{\bs\theta}'_n, I_{\bs\theta_0}^{-1})$ is shifted to the right of $\mathcal{N}(\widehat{\bs\theta}_n, I_{\bs\theta_0}^{-1})$ elementwise, with a single intersection point $\Tilde{\bs\theta}$. For $\bs\theta \leq \Tilde{\bs\theta}$, we have $g(\widehat{\bs\theta}'_n)- g(\widehat{\bs\theta}_n)\le 0$ while for $\bs\theta \geq \Tilde{\bs\theta}$, $g(\widehat{\bs\theta}'_n)- g(\widehat{\bs\theta}_n)\ge 0$. Due to symmetry (see Figure \ref{fig:bivariate_G} for an illustration), there are two maximizers in $\bs\theta$, where $|g(\widehat{\bs\theta}_n)-g(\widehat{\bs\theta}_n')|$  achieves the maximum; that is, there exists a constant vector $\mathbf{d} = (d_1, \ldots, d_p)^{\top}$ with $d_j \geq 0$ such that $\frac{\partial(g(\widehat{\bs\theta}_n)- g(\widehat{\bs\theta}_n'))}{\partial\bs\theta}|_{\bs\theta = \widehat{\bs\theta}_n'+\mathbf{d}}=0$ and $\frac{\partial(g(\widehat{\bs\theta}_n)- g(\widehat{\bs\theta}_n'))}{\partial\bs\theta}|_{\bs\theta = \widehat{\bs\theta}_n-\mathbf{d}}=0$.

Similar to Section \ref{ape:G_proof_uni}, we first prove the uniqueness of maximum for $g(\widehat{\bs\theta}'_n)- g(\widehat{\bs\theta}_n)$ when $\bs\theta \ge \Tilde{\bs\theta}$. Applying the same log-transformation to  $g(\widehat{\bs\theta}_n)( \frac{g(\widehat{\bs\theta}'_n)}{g(\widehat{\bs\theta}_n)}-1)$. First,
\begin{align*}
&\nabla^2_{\bs{\theta}} \log(g(\widehat{\bs\theta}_n)) = \frac{-nI_{\bs{\theta}_0}(\widehat{\bs\theta}'_n \!-\! \widehat{\bs\theta}_n)}{\partial \theta} = -\frac{n}{I^{-1}_{\bs\theta_0}} < 0; \\
\mbox{Let } z(\bs\theta) &= \log \left(\frac{g(\widehat{\bs\theta}'_n)}{g(\widehat{\bs\theta}_n)}-1\right) =\log\left( \exp\left(-\frac{n}{2}(\bs\theta - \widehat{\bs\theta}_n')^{\top}I_{\bs\theta_0}(\bs\theta \!-\! \widehat{\bs\theta}_n')+\frac{n}{2}(\bs\theta - \widehat{\bs\theta}_n)^{\top}I_{\bs\theta_0}(\bs\theta - \widehat{\bs\theta}_n)\right)-1\right),\\
\mbox{then }  \nabla_{\bs{\theta}} z(\bs\theta)&= \frac{\nabla \frac{g(\widehat{\bs\theta}'_n)}{g(\widehat{\bs\theta}_n)}}{\frac{g(\widehat{\bs\theta}'_n)}{g(\widehat{\bs\theta}_n)}-1} = \frac{g(\widehat{\bs\theta}_n)\nabla g(\widehat{\bs\theta}'_n) - g(\widehat{\bs\theta}'_n)\nabla g(\widehat{\bs\theta}_n)}{g^2(\widehat{\bs\theta}_n)\left(\frac{g(\widehat{\bs\theta}'_n)}{g(\widehat{\bs\theta}_n)}-1\right)}\\
    &=\frac{-n I_{\bs{\theta}_0} (\bs{\theta} \!-\! \widehat{\bs{\theta}}'_n)g(\widehat{\bs\theta}'_n) + g(\widehat{\bs\theta}'_n)n I_{\bs{\theta}_0} (\bs{\theta} - \widehat{\bs{\theta}}_n)}{g(\widehat{\bs\theta}'_n) - g(\widehat{\bs\theta}_n)} = \frac{g(\widehat{\bs\theta}'_n)}{g(\widehat{\bs\theta}'_n) - g(\widehat{\bs\theta}_n)}\cdot nI_{\bs{\theta}_0}\left(\widehat{\bs\theta}'_n \!-\!\widehat{\bs\theta}_n\right)\\
    \nabla^2_{\bs{\theta}} z(\bs\theta)&= nI_{\bs{\theta}_0}\left(\widehat{\bs\theta}'_n - \widehat{\bs\theta}_n\right) \cdot \nabla_{\bs{\theta}} \left(\frac{g(\widehat{\bs\theta}'_n)}{g(\widehat{\bs\theta}'_n) - g(\widehat{\bs\theta}_n)}\right)\\
    &= nI_{\bs{\theta}_0}(\widehat{\bs\theta}'_n \!-\! \widehat{\bs\theta}_n)\frac{(g(\widehat{\bs\theta}'_n) - g(\widehat{\bs\theta}_n))\nabla_{\bs\theta}g(\widehat{\bs\theta}'_n) - g(\widehat{\bs\theta}'_n)\nabla_{\bs\theta}(g(\widehat{\bs\theta}'_n) - g(\widehat{\bs\theta}_n))}{(g(\widehat{\bs\theta}'_n) - g(\widehat{\bs\theta}_n))^2}\\
    &= \frac{-n I_{\bs{\theta}_0} (\bs{\theta} - \widehat{\bs{\theta}}'_n)(g(\widehat{\bs\theta}'_n) - g(\widehat{\bs\theta}_n))g(\widehat{\bs\theta}'_n) - g(\widehat{\bs\theta}'_n)(-n I_{\bs{\theta}_0} (\bs{\theta} - \widehat{\bs{\theta}}'_n)g(\widehat{\bs\theta}'_n) +n I_{\bs{\theta}_0} (\bs{\theta} - \widehat{\bs{\theta}}_n) g(\widehat{\bs\theta}_n))}{(g(\widehat{\bs\theta}'_n) - g(\widehat{\bs\theta}_n))^2}\\
    &\quad \cdot nI_{\bs{\theta}_0}(\widehat{\bs\theta}'_n \!-\! \widehat{\bs\theta}_n)\\
    &= -nI_{\bs{\theta}_0}(\widehat{\bs\theta}'_n \!-\! \widehat{\bs\theta}_n)^2\frac{g(\widehat{\bs\theta}_n)g(\widehat{\bs\theta}'_n)}{(g(\widehat{\bs\theta}'_n) - g(\widehat{\bs\theta}_n))^2}< 0 
\end{align*}
Similarly to the argument in the single-parameter case in Section \ref{ape:G_proof_uni}, $g(\widehat{\bs\theta}'_n)- g(\widehat{\bs\theta}_n) = g(\widehat{\bs\theta}_n)(\frac{g(\widehat{\bs\theta}'_n)}{g(\widehat{\bs\theta}_n)}-1)$ is log-concave and has a unique maximum. 

To solve for $\bs\theta$ that leads to the maximum difference, we take the 1st derivative of Eq.~\eqref{eqn:multi_taylor} with respect to $\bs\theta$. 
\begin{align}
    &\frac{\partial (g(\widehat{\bs\theta}_n) - g(\widehat{\bs\theta}_n'))}{\partial \bs\theta}
    \approx\frac{\partial g(\widehat{\bs\theta}_n')}{\partial \bs\theta}\Big[-(\bs\theta -\widehat{\bs\theta}_n')^{\top}I_{\bs\theta_0}\mathbf{C}+ \frac{1}{2n}\mathbf{C}^{\top}\left(n\left(I_{\bs\theta_0}(\bs\theta -\widehat{\bs\theta}_n')(\bs\theta -\widehat{\bs\theta}_n')^{\top}I_{\bs\theta_0}\right) -I_{\bs\theta_0}\right)\mathbf{C}\Big]\notag\\
    & \qquad + g(\widehat{\bs\theta}_n')\Big[-I_{\bs\theta_0}\mathbf{C}+ \frac{1}{2}\underbrace{\frac{\partial \mathbf{C}^{\top}I_{\bs\theta_0}(\bs\theta -\widehat{\bs\theta}_n')(\bs\theta -\widehat{\bs\theta}_n')^{\top}I_{\bs\theta_0}^{\top}\mathbf{C}}{\partial (\bs\theta -\widehat{\bs\theta}_n')(\bs\theta -\widehat{\bs\theta}_n')^{\top} }}_{=I_{\bs\theta_0}^{\top}\mathbf{C}\mathbf{C}^{\top}I_{\bs\theta_0}}\cdot \underbrace{\frac{\partial (\bs\theta -\widehat{\bs\theta}_n')(\bs\theta -\widehat{\bs\theta}_n')^{\top}}{\partial \bs\theta}}_{=2(\bs\theta-\widehat{\bs\theta}_n')}\Big]\\
    =\;& g(\widehat{\bs\theta}_n')(-nI_{\bs\theta_0}(\bs\theta-\widehat{\bs\theta}_n'))\Big[-(\bs\theta -\widehat{\bs\theta}_n')^{\top}I_{\bs\theta_0}\mathbf{C}+ \frac{1}{2n}\mathbf{C}^{\top}\left(n\left(I_{\bs\theta_0}(\bs\theta -\widehat{\bs\theta}_n')(\bs\theta -\widehat{\bs\theta}_n')^{\top}I_{\bs\theta_0}\right) -I_{\bs\theta_0}\right)\mathbf{C}\Big]\notag\\
    & + g(\widehat{\bs\theta}_n')\Big[-I_{\bs\theta_0}\mathbf{C}+ I_{\bs\theta_0}\mathbf{C}\mathbf{C}^{\top}I_{\bs\theta_0}(\bs\theta -\widehat{\bs\theta}_n')\Big].\label{eqn:multi_deriv}
\end{align}
Set Eq.~\eqref{eqn:multi_deriv} equal to 0, then 
\begin{align}
    0=\;&nI_{\bs\theta_0}(\bs\theta-\widehat{\bs\theta}_n')(\bs\theta - \widehat{\bs\theta}_n')^{\top}I_{\bs\theta_0}\mathbf{C}+\frac{1}{2}I_{\bs\theta_0}(\bs\theta-\widehat{\bs\theta}_n')\mathbf{C}^{\top}I_{\bs\theta_0}\mathbf{C}\notag\\
    &-\frac{n}{2}I_{\bs\theta_0}(\bs\theta-\widehat{\bs\theta}_n')\mathbf{C}^{\top}I_{\bs\theta_0}(\bs\theta - \widehat{\bs\theta}_n')(\bs\theta - \widehat{\bs\theta}_n')^{\top}I_{\bs\theta_0}\mathbf{C}-I_{\bs\theta_0}\mathbf{C}+ I_{\bs\theta_0}\mathbf{C}\underbrace{\mathbf{C}^{\top}I_{\bs\theta_0}(\bs\theta - \widehat{\bs\theta}_n')}_{=c}\\
    =\;& ncI_{\bs\theta_0}(\bs\theta\!-\!\widehat{\bs\theta}_n') \!+\! \frac{1}{2}I_{\bs\theta_0}(\bs\theta\!-\!\widehat{\bs\theta}_n')\mathbf{C}^{\top}I_{\bs\theta_0}\mathbf{C}-\frac{nc^2}{2}I_{\bs\theta_0}(\bs\theta\!-\!\widehat{\bs\theta}_n')\!+\!(c-1)I_{\bs\theta_0}\mathbf{C}\notag\\
    =& \left(nc-\frac{nc^2}{2}+\frac{\mathbf{C}^{\top}I_{\bs\theta_0}\mathbf{C}}{2}\right)I_{\bs\theta_0}(\bs\theta-\widehat{\bs\theta}_n')+(c-1)I_{\bs\theta_0}\mathbf{C}.\label{eqn:multi_deriv=0}
\end{align}
Plug in $\bs\theta - \widehat{\bs\theta}_n'=\bs d$ and $\bs\theta - \widehat{\bs\theta}_n'\approx \bs\theta - (\widehat{\bs\theta}_n + \frac{\mathbf{C}}{n}+o(n^{-1})) = -\frac{\mathbf{C}}{n}-\bs d +o(n^{-1}) $ and define two constants,
\begin{align}
    c_1 &= \mathbf{C}^{\top}I_{\bs\theta_0}\bs d\label{eqn:c1}\\
    c_2 &\approx -\mathbf{C}^{\top}I_{\bs\theta_0}\left(\frac{\mathbf{C}}{n} + \bs d + o(n^{-1})\right) = -\frac{1}{n}\underbrace{\mathbf{C}^{\top}I_{\bs\theta_0}\mathbf{C}}_{=a} - c_1.\label{eqn:c2}
\end{align}
and plug Eqs~\eqref{eqn:c1} and \eqref{eqn:c2} into Eq.~\eqref{eqn:multi_deriv=0}, we have
\begin{align}
    &\left(nc_1-\frac{nc_1^2}{2}+\frac{\mathbf{C}^{\top}I_{\bs\theta_0}\mathbf{C}}{2}\right)I_{\bs\theta_0}\bs d=(1-c_1)I_{\bs\theta_0}\mathbf{C}\label{eqn:multi+d}\\
    &\left(nc_2-\frac{nc_2^2}{2}+\frac{\mathbf{C}^{\top}I_{\bs\theta_0}\mathbf{C}}{2}\right)I_{\bs\theta_0}\left(\frac{\mathbf{C}}{n}+\bs d\right)=(1-c_2)I_{\bs\theta_0}\mathbf{C}.\label{eqn:multi-d}
\end{align}
Taking the difference between Eq.~\eqref{eqn:multi+d} and Eq.~\eqref{eqn:multi-d}, we have
\begin{align}
    &(c_2-c_1)I_{\bs\theta_0}\mathbf{C} = \left(n(c_1-c_2)+\frac{n(c_2-c_1)(c_2+c_1)}{2}\right)I_{\bs\theta_0}\bs d - \left(c_2-\frac{c_2^2}{2}+\frac{\mathbf{C}^{\top}I_{\bs\theta_0}\mathbf{C}}{2n}\right)I_{\bs\theta_0}\mathbf{C}\notag\\
    &\triangleq n(c_2-c_1)\left(-1+\frac{(c_2+c_1)}{2}\right)I_{\bs\theta_0}\bs d - \left(c_2-\frac{c_2^2}{2}+\frac{a}{2n}\right)I_{\bs\theta_0}\mathbf{C}, \mbox{ where } a\triangleq \mathbf{C}^{\top}I_{\bs\theta_0}\mathbf{C}\\
    \Rightarrow & -(\frac{a}{n}+2c_1)I_{\bs\theta_0}\mathbf{C} = -n(\frac{a}{n}+2c_1)\left(-1-\frac{a}{2n}\right)I_{\bs\theta_0}\bs d - \left(c_2-\frac{c_2^2}{2}+\frac{a}{2n}\right)I_{\bs\theta_0}\mathbf{C}\\
    \Rightarrow & (2c_1-\frac{a}{2n})a = n(\frac{a}{n}+2c_1)\left(-1-\frac{a}{2n}\right)c_1 + \left(-\frac{a}{n}-c_1-\frac{(-\frac{a}{n}-c_1)^2}{2}\right)a\\
    \Rightarrow & \left(\frac{3a}{2}+2n\right)c_1^2-\left(2a+\frac{3a^2}{2n}\right)c_1+\frac{a^3}{2n^2}-2a-\frac{a^2}{2n} =  0.\label{eqn:c1_equation}
\end{align}
$c_1$ can be solved analytically from Eq.~\eqref{eqn:c1_equation}
\begin{align}
    & \Delta = \left(2a+\frac{3a^2}{2n}\right)^2 - 4\left(\frac{3a}{2}+2n\right)\left(\frac{a^3}{2n^2}-2a-\frac{a^2}{2n}\right)\notag\\
    &\; \;\;=4a^2 + \frac{9a^4}{4n^2} + \frac{6a^3}{n} - 4\left(\frac{3a^4}{4n^2}-3a^2-\frac{3a^3}{4n}+\frac{a^3}{n}-4an-a^2\right)\notag\\
    &\; \;\;= 4a^2 + \frac{9a^4}{4n^2} + \frac{6a^3}{n} + 16a^2+16an-\frac{a^3}{n}+\frac{3a^4}{n^2}\notag\\
    &\; \;\;= 16an + 20a^2 + \frac{21a^4}{4n^2}-\frac{5a^3}{n}\\
    &c_1  =-(\x_n'-\x_n)^{\top}I_{\bs\theta_0}\bs d \notag\\
    &\; \;\;= \frac{\left(2a+\frac{3a^2}{2n}\right) \pm \sqrt{\Delta}}{4n+3a} = \frac{\left(2a+\frac{3a^2}{2n}\right) \pm \sqrt{16an + 20a^2 + \frac{21a^4}{4n^2}-\frac{5a^3}{n}}}{4n+3a}\approx \frac{a}{2n} \pm \sqrt{\frac{a}{n}}.\label{eqn:c1_result}
\end{align}
Plug Eq.~\eqref{eqn:c1_result} into Eq.~\eqref{eqn: multi_objective}, we can have
\begin{align}
    G(n) &= (2\pi)^{-\frac{p}{2}} \sqrt{\frac{n^p}{\det(I^{-1}_{\bs\theta_0})}}\Big|g(\widehat{\bs\theta}_n)-g(\widehat{\bs\theta}_n')\Big|_{\bs\theta-\widehat{\bs\theta}_n'=\bs d}\notag\\
    &\approx (2\pi)^{-\frac{p}{2}} \sqrt{\frac{n^p}{\det(I^{-1}_{\bs\theta_0})}}\exp\left(-\frac{n}{2}\bs d^{\top} I_{\bs\theta_0} \bs d\right)\cdot \Big|\!-\bs d^{\top}I_{\bs\theta_0}\mathbf{C}+ \frac{1}{2n}\mathbf{C}^{\top}\left(nI_{\bs\theta_0}\bs d \bs d^{\top}I_{\bs\theta_0} -I_{\bs\theta_0}\right)\mathbf{C}\Big|\notag\\
    &= (2\pi)^{-\frac{p}{2}} \sqrt{\frac{n^p}{\det(I^{-1}_{\bs\theta_0})}}\exp\left(-\frac{n}{2}\bs d^{\top} I_{\bs\theta_0} \bs d\right)\Big|-c_1+ \frac{c_1^2}{2}-\frac{a}{2n}\Big|\notag\\
    &= (2\pi)^{-\frac{p}{2}} \sqrt{\frac{n^p}{\det(I^{-1}_{\bs\theta_0})}}\exp\left(\!-\frac{nc_1^2}{2}\left([\mathbf{C}^{\top}\!I_{\bs\theta_0}]^{-1} \right)^{\top}I_{\bs\theta_0} [\mathbf{C}^{\top}I_{\bs\theta_0}]^{-1}\!\right)\Big|\!-\!c_1\!+\! \frac{c_1^2}{2}\!-\!\frac{a}{2n}\Big|\notag\\
    &= (2\pi)^{-\frac{p}{2}} \sqrt{\frac{n^p}{\det(I^{-1}_{\bs\theta_0})}}\exp\left(-\frac{nc_1^2}{2}\underbrace{[I_{\bs\theta_0}\mathbf{C}]^{-1} [\mathbf{C}^{\top}]^{-1}}_{=a^{-1}}\right)\Big|\!-\!c_1\!+\! \frac{c_1^2}{2}\!-\!\frac{a}{2n}\Big|\notag\\
    &= (2\pi)^{-\frac{p}{2}} \sqrt{\frac{n^p}{\det(I^{-1}_{\bs\theta_0})}}\exp\left(-\frac{nc_1^2}{2a}\right)\Big|-c_1+ \frac{c_1^2}{2}-\frac{a}{2n}\Big|\\
    &\approx (2\pi)^{-\frac{p}{2}} \sqrt{\frac{n^p}{\det(I^{-1}_{\bs\theta_0})}}\exp\left(-\frac{1}{2}+\mathcal{O}(n^{-1/2})\right)\Big|\sqrt{\frac{a}{n}}+\mathcal{O}(n^{-1})\Big|\\
    &= n^{\frac{p-1}{2}} (2\pi)^{-\frac{p}{2}} \sqrt{\frac{\mathbf{C}^{\top}I_{\bs\theta_0}\mathbf{C}}{\det(I^{-1}_{\bs\theta_0})}}\exp\left(-\frac{1}{2}+\mathcal{O}(n^{-1/2})\right)+\mathcal{O}(n^{-1/2})
\end{align}

\subsection{Proof of Theorem \ref{thm:DP}}\label{ape:DPproof}
\begin{proof}

$H_p$ and $H_p'$, the histograms with bin width $h$  represented in probability  based on $m$ samples of $\theta$, are discretized probability distribution estimates for  $f_{\theta|\x}$ and $f_{\theta|\x'}$, respectively.  The $\ell_1$ distance between $H_p$ and $H_p'$ is  $\|H_p \!-\!H'_p\|_1=2\text{TVD}_{H_p, H'_p}\!=\! 2\sup_{b\in \{1, \ldots, B\}}|p_b - p'_b|$, where TVD stands for total variation distance, $p_b\!=\!\Pr(\theta\!\in\!\text{bin $b$ in $H_p$})$, and $p'_b\!=\!\Pr(\theta\!\in\!\text{bin $b$ in $H'_p$})$. The $\ell_1$ global sensitivity of the histogram with one-record change in $\x$ is given by  $\Delta_H\!=\!\max_{d(\x,\x')=1}\|H_p\!-\!H'_p\|_1\!=\!2\max_{d(\x,\x')=1}\sup_{b\in \{1, \ldots, B\}}|p_b - p'_b|\le2\sup_{d(\x,\x')=1,b\in \{1, \ldots, B\}}|p_b - p'_b|$. Since $p_b=f_{\theta|\x}(\xi_b)h$ and $p'_b=f_{\theta|\x'}(\xi_b')h$ per the mean value theorem, where $\xi'_b\approx\xi_b\in \Lambda_b$ if $h$ is small enough, $|p_b - p'_b|=|f_{\theta|\x}(\xi_b)-f_{\theta|\x'} (\xi_b)|h$, which is $\le Gh$. Thus $\Delta_{H_p}=2Gh$ and  $\Delta_{H}=2mGh$, where $H$ is the histogram represented in frequencies/counts.

\end{proof}

\subsection{Proof of Theorem \ref{thm:MSE_RAP}}\label{proof:PRECISE}

\begin{proof}
Let $\mathcal{M}$ denote the PRECISE procedure in Algorithm \ref{alg:PRECISE}; and we use $\theta^*_{(q)}$ and $ \theta^*_{([qm])}$ interchangeably  to denote the PP $q^{th}$ sample quantile in this section.

\textit{First}, we can expand the Mean Squared Error (MSE) between the sanitized $q^{th}$ posterior sample quantile $\theta^*_{([qm])}$ and the population posterior quantile $F^{-1}_{\theta|\mathbf{x}}(q)$ as
\vspace{-6pt}
\begin{align}
&\mathbb{E}_{\boldsymbol{\theta}}\mathbb{E}_{\mathcal{M}|\boldsymbol{\theta}}\left(\theta^*_{([qm])} - F^{-1}_{\theta|\mathbf{x}}(q)\right)^2 
= \mathbb{E}_{\boldsymbol{\theta}}\mathbb{E}_{\mathcal{M}|\boldsymbol{\theta}}\left(\theta^*_{([qm])} - \theta_{([qm])} + \theta_{([qm])} - F^{-1}_{\theta|\mathbf{x}}(q)\right)^2 \notag\\
=\;& \mathbb{E}_{\boldsymbol{\theta}}\mathbb{E}_{\mathcal{M}|\boldsymbol{\theta}}\left(\theta^*_{([qm])} - \theta_{([qm])}\right)^2  + \mathbb{E}_{\boldsymbol{\theta}}\mathbb{E}_{\mathcal{M}|\boldsymbol{\theta}}\left(\theta_{([qm])} - F^{-1}_{\theta|\mathbf{x}}(q)\right)^2 \notag\\
& \; +2\mathbb{E}_{\boldsymbol{\theta}}\mathbb{E}_{\mathcal{M}|\boldsymbol{\theta}}\left(\theta^*_{([qm])} - \theta_{([qm])}\right)\left(\theta_{([qm])} - F^{-1}_{\theta|\mathbf{x}}(q)\right) \notag\\
\leq\;&\mathbb{E}_{\boldsymbol{\theta}}\mathbb{E}_{\mathcal{M}|\boldsymbol{\theta}}\left(\theta^*_{([qm])} - \theta_{([qm])}\right)^2  + \mathbb{E}_{\boldsymbol{\theta}}\mathbb{E}_{\mathcal{M}|\boldsymbol{\theta}}\left(\theta_{([qm])} - F^{-1}_{\theta|\mathbf{x}}(q)\right)^2 \notag\\
& \; + 2\sqrt{\mathbb{E}_{\boldsymbol{\theta}}\mathbb{E}_{\mathcal{M}|\boldsymbol{\theta}}\left(\theta^*_{([qm])} - \theta_{([qm])}\right)^2 \mathbb{E}_{\boldsymbol{\theta}}\mathbb{E}_{\mathcal{M}|\boldsymbol{\theta}}\left(\theta_{([qm])} - F^{-1}_{\theta|\mathbf{x}}(q)\right)^2}.\label{eqn:CSinequality}
\end{align}
The last inequality in Eq.~\eqref{eqn:CSinequality} holds per the Cauchy-Schwarz inequality. 
Per Theorem 1 in \citep{walker1968note}, sample quantiles are asymptotically Gaussian, that is
\vspace{-3pt}
\begin{equation}
    \sqrt{m}\left(\theta_{([qm])} - F^{-1}_{\theta|\mathbf{x}}(q)\right)\overset{d}{\rightarrow}\mathcal{N}\left(0, q(1-q)\cdot \left(f_{\theta|\mathbf{x}}\left(F^{-1}_{\theta|\mathbf{x}}(q)\right)\right)^{-2}\right),
\end{equation}
based on which, we obtain the following result for the second square term in Eq.~\eqref{eqn:CSinequality} as $m\rightarrow\infty$,
\begin{align}\label{eqn:CS_1}
\mathbb{E}_{\boldsymbol{\theta}}\mathbb{E}_{\mathcal{M}|\boldsymbol{\theta}}\!\left(\theta_{([qm])} \!-\! F^{-1}_{\theta|\mathbf{x}}(q)\right)^2 \!\!=\mathbb{E}_{\boldsymbol{\theta}}\left(\theta_{([qm])} \!-\! F^{-1}_{\theta|\mathbf{x}}(q)\right)^2\!\!\rightarrow \frac{q(1-q)}{m}\!\cdot\! \left(\!f_{\theta|\mathbf{x}}\!\left(\!F^{-1}_{\theta|\mathbf{x}}(q)\right)\right)^{-2}.
\end{align}
As $n \to \infty$, under the regularity conditions of the Bernstein–von Mises theorem, the posterior density $f_{\theta|\mathbf{x}}(\theta)$ converges to a Gaussian density centered at the MAP with variance shrinking at the rate of $\mathcal{O}(1/n)$. i.e. $f_{\theta|\mathbf{x}}(\theta) \approx \phi\left(\theta; \hat{\theta}_n, n^{-1} I_{\theta_0}^{-1}\right)$. So, for any fixed $q\in (0,1)$, $F^{-1}_{\theta|\mathbf{x}}(q) \rightarrow \theta_0 + z_q\sqrt{I_{\theta_0}^{-1}/n}$, the density at the posterior quantile, $f_{\theta|\mathbf{x}}(F^{-1}_{\theta|\mathbf{x}}(q))$ can be approximated as
\begin{align}
    f_{\theta|\mathbf{x}}(F^{-1}_{\theta|\mathbf{x}}(q)) &\approx \frac{1}{\sqrt{2\pi \cdot n^{-1} I_{\theta_0}^{-1}}} \cdot \exp\left( -\frac{(z_q)^2}{2} \right) = \mathcal{O}(\sqrt{n})\\
    \left(f_{\theta|\mathbf{x}}\left(F^{-1}_{\theta|\mathbf{x}}(q)\right)\right)^{-2} &= \mathcal{O}(n^{-1}).
\end{align}
Next we upper bound the term $\mathbb{E}_{\boldsymbol{\theta}}\mathbb{E}_{\mathcal{M}|\boldsymbol{\theta}}\!\left(\!\theta^*_{([qm])} \!-\! \theta_{([qm])}\!\right)^2$ in Eq.~\eqref{eqn:CSinequality}. We first show the bias introduced by the truncation at $0$  (step~\ref{step:DP} Algorithm \ref{alg:PRECISE})  decays exponentially as $\varepsilon \!\rightarrow\! \infty$. For $\forall b\in \{0, \ldots, B'+1\}$,
\begin{align}
\mathbb{E}_{\mathcal{M}|\boldsymbol{\theta}}\left(c_b^*\right) &= \int_{-\infty}^{\infty}\max\{0, c_b+x\}\frac{\varepsilon}{2}e^{-\varepsilon|x|}dx = 0 + \int_{-c_b}^{\infty}(c_b+x)\frac{\varepsilon}{2}e^{-\varepsilon|x|}dx\notag\\
    &= \int_{-c_b}^{0}(c_b+x)\frac{\varepsilon}{2}e^{\varepsilon x}dx + \int_{0}^{\infty}(c_b+x)\frac{\varepsilon}{2}e^{-\varepsilon x}dx\notag\\
    &= \frac{\varepsilon c_b}{2}\int_{-c_b}^{0}e^{\varepsilon x}dx + \frac{\varepsilon}{2}\int_{-c_b}^{0}xe^{\varepsilon x}dx + \frac{\varepsilon c_b}{2}\int_{0}^{\infty}e^{-\varepsilon x}dx + \frac{\varepsilon}{2}\int_{0}^{\infty}x e^{-\varepsilon x} dx\notag\\
    &=\frac{c_b}{2}\left(1-e^{-\varepsilon c_b}\right) + \frac{1}{2}\left(c_be^{-\varepsilon c_b} - \frac{1}{\varepsilon}(1-e^{-\varepsilon c_b})\right) + \frac{c_b}{2} + \frac{1}{2\varepsilon}\notag\\
    &= c_b + \frac{e^{-\varepsilon c_b}}{2\varepsilon}. \label{eqn:bias_c*}
\end{align}

Then, we calculate the second moment for $c^*_b$ in a similar manner
\begin{align}
\mathbb{E}_{\mathcal{M}|\boldsymbol{\theta}}\left(c_b^{*2}\right) &= \int_{-\infty}^{\infty}\left(\max\{0, c_b+x\}\right)^2\frac{\varepsilon}{2}e^{-\varepsilon|x|}dx = 0 + \int_{-c_b}^{\infty}(c_b+x)^2\frac{\varepsilon}{2}e^{-\varepsilon|x|}dx\notag\\
    &= \int_{-c_b}^{0}(c^2_b+x^2 + 2c_bx)\frac{\varepsilon}{2}e^{\varepsilon x}dx + \int_{0}^{\infty}(c^2_b+x^2 + 2c_bx)\frac{\varepsilon}{2}e^{-\varepsilon x}dx\notag\\
    &= \frac{\varepsilon c_b^2}{2}\int_{-c_b}^{0}e^{\varepsilon x}dx + \frac{\varepsilon}{2}\int_{-c_b}^{0}x^2e^{\varepsilon x}dx + \varepsilon c_b\int_{-c_b}^0 x e^{\varepsilon x}dx\notag\\
    &\;\;\;\;+\frac{\varepsilon c_b^2}{2}\int_{0}^{\infty}e^{-\varepsilon x}dx + \frac{\varepsilon}{2}\int_{0}^{\infty}x^2e^{-\varepsilon x}dx + \varepsilon c_b\int_{0}^{\infty} x e^{-\varepsilon x}dx\notag\\
    &=\frac{c_b^2}{2}\left(1\!-\!e^{-\varepsilon c_b}\right)\! - \!\frac{c_b^2}{2}e^{-\varepsilon c_b} \!- \!\left(\!\frac{c_b}{\varepsilon}e^{- \varepsilon c_b} \!-\! \frac{1}{\varepsilon^2}(1\!-\!e^{-\varepsilon c_b})\!\right)\!+ \frac{c_b^2}{2} +\frac{1}{\varepsilon^2} +\frac{c_b}{\varepsilon}\notag\\
    & \;\;\;\;+ c_b\left(c_be^{-\varepsilon c_b} - \frac{1}{\varepsilon}(1-e^{-\varepsilon c_b})\right) \notag\\
    &= c^2_b + \frac{2}{\varepsilon^2} - \frac{e^{-\varepsilon c_b}}{\varepsilon^2}. \label{eqn:2ndmoment_c*}
\end{align}
Given Eqs~\eqref{eqn:bias_c*} and \eqref{eqn:2ndmoment_c*}, we are ready to show the MSE consistency of sanitized bin count $c^*_b$ for $c_b$ over sanitization, that is, $\mathbb{E}_{\mathcal{M}|\boldsymbol{\theta}}(c^*_b - c_b)^2\rightarrow 0$.
For $\forall b\in \{0, \ldots, B'+1\}$, 
\begin{align}
\mathbb{E}_{\mathcal{M}|\boldsymbol{\theta}}\left((c_b^* - c_b)^2\right) 
    &= \mathbb{E}_{\mathcal{M}|\boldsymbol{\theta}}\left(c_b^{*2}\right) -2c_b \mathbb{E}_{\mathcal{M}|\boldsymbol{\theta}}\left(c_b^* \right) + c_b^2\notag\\
    &=c^2_b + \frac{2}{\varepsilon^2} - \frac{e^{-\varepsilon c_b}}{\varepsilon^2} - 2 c_b\left(c_b + \frac{e^{-\varepsilon c_b}}{2\varepsilon}\right) + c_b^2\notag\\
    &= \frac{2}{\varepsilon^2} - \frac{e^{-\varepsilon c_b}}{\varepsilon^2} - \frac{c_be^{-\varepsilon c_b}}{\varepsilon} = \mathcal{O}\left(\varepsilon^{-2}\right).\label{eqn:MSE_c*}
\end{align}
In addition, we derive the variance for $c^*_b$  for later use,
\begin{align}
\mathbb{V}_{\mathcal{M}|\boldsymbol{\theta}}\left(c_b^*\right) &= \mathbb{E}_{\mathcal{M}|\boldsymbol{\theta}}\left(c_b^{2*}\right) -\left(\mathbb{E}_{\mathcal{M}|\boldsymbol{\theta}}\left(c_b^*\right)\right)^2= c^2_b + \frac{2}{\varepsilon^2} - \frac{e^{-\varepsilon c_b}}{\varepsilon^2} - \left(c_b + \frac{e^{-\varepsilon c_b}}{2\varepsilon}\right)^2 \notag\\
&= \frac{2-e^{-\varepsilon c_b}-e^{-2\varepsilon c_b}/4}{\varepsilon^2} -\frac{c_b e^{-\varepsilon c_b}}{\varepsilon}. \label{eqn:var_c*}
\end{align}
Let $g(b) \!=\! |\!\sum_{i\leq b}\!c_i - qm|$ and $g^*(b) \!=\! |\sum_{i\leq b} \!c^*_i - qm^*|$, where $m^* \!=\! \sum_{b=0}^{B'+1}c^*_b$ in step~\ref{step:normalization*} of Algorithm \ref{alg:PRECISE}. Let $\hat{b}$ be the true index of bin for $\theta_{([qm])}$, i.e., $\theta_{([qm])} \!\in\! I_{\hat{b}} $, where $\hat{b} = \min\{{\arg\min}_{b\in \{0, \ldots, B'+1\}}g(b)\}.$ We prove the bin index identification error $\mathbb{E}_{\mathcal{M}|\boldsymbol{\theta}}\left(b^* - b\right)^2\!\rightarrow\! 0$, where  \[b^*\!=\min\{{\arg\min}_{b\in \{0, \ldots, B'+1\}}g^*(b)\},\] by first showing the squared error between the objective functions, from which $b$ ad $b^*$ are solved, converges to $0$ as $\varepsilon \rightarrow \infty$; that is
\vspace{-3pt}
\begin{equation}\label{eqn:MSE_b*}
\mathbb{E}_{\mathcal{M}|\boldsymbol{\theta}}\left(g^*(b) - g(b)\right)^2 = \mathbb{E}_{\mathcal{M}|\boldsymbol{\theta}}\left(\Big|\sum_{i\leq b} c^*_i - qm^*\Big| - \Big|\sum_{i\leq b} c_i - qm\Big|\right)^2\!\!\rightarrow 0.
\end{equation}
Per definition of $m^*$, $m^* = \sum_{i\leq b}c_i^* + \sum_{i\geq b+1}c_i^*$ for $\forall b\in \{0, \ldots, B'+1\}$, then
\begin{align}
&\mathbb{E}_{\mathcal{M}|\boldsymbol{\theta}}\left(\sum_{i\leq b} c^*_i - qm^*\right)^2
    = \mathbb{E}_{\mathcal{M}|\boldsymbol{\theta}}\left((1-q)\sum_{i\leq b} c^*_i - q\sum_{i\geq b+1} c^*_i \right)^2\notag\\
    =\;& (1-q)^2 \mathbb{E}_{\mathcal{M}|\boldsymbol{\theta}}\!\left(\sum_{i\leq b} c^*_i\!\!\right)^2 \!\!\! + q^2 \mathbb{E}_{\mathcal{M}|\boldsymbol{\theta}}\!\left(\!\sum_{i\geq b+1} \!\!c^*_i \!\!\right)^2 \!\!\!- 2q(1\!-\!q)\mathbb{E}_{\mathcal{M}|\boldsymbol{\theta}}\!\left(\!\sum_{i\leq b} c^*_i\!\!\right)\!\mathbb{E}_{\mathcal{M}|\boldsymbol{\theta}}\!\left(\!\sum_{i\geq b+1}\!\! c^*_i\!\!\right)\label{eqn:MSEb*_1}\\
    =\;& (1-q)^2\!\!\left(\!\mathbb{V}_{\mathcal{M}|\boldsymbol{\theta}}\!\left(\!\sum_{i\leq b} c^*_i\!\right)\! + \!\left(\!\sum_{i\leq b}\mathbb{E}_{\mathcal{M}|\boldsymbol{\theta}}(c_i^*)\!\right)^2\!\right) \!\!+ q^2\!\!\left(\!\mathbb{V}_{\mathcal{M}|\boldsymbol{\theta}}\!\left(\sum_{i\geq b+1} \!c^*_i\!\right)\!\! + \!\left(\sum_{i\geq b+1}\!\!\mathbb{E}_{\mathcal{M}|\boldsymbol{\theta}}(c_i^*)\!\right)^2\!\right)\notag\\
    & -2q(1-q)\left(\sum_{i\leq b}\left(c_i + \frac{e^{-\varepsilon c_i}}{2\varepsilon}\right)\right)\left(\sum_{i\geq b+1}\left(c_i + \frac{e^{-\varepsilon c_i}}{2\varepsilon}\right)\right)\label{eqn:MSEb*_2}\\
    =\;&(1-q)^2\!\left(\!\sum_{i\leq b}\left(\frac{2-e^{-\varepsilon c_i}-e^{-2\varepsilon c_i}/4}{\varepsilon^2} -\frac{c_i e^{-\varepsilon c_i}}{\varepsilon}\right) + \!\left(\!\sum_{i\leq b}\left(c_i + \frac{e^{-\varepsilon c_i}}{2\varepsilon}\right)\!\right)^2\right)\notag\\
    & + q^2\!\left(\!\sum_{i\geq b+1}\left(\frac{2-e^{-\varepsilon c_i}-e^{-2\varepsilon c_i}/4}{\varepsilon^2} -\frac{c_i e^{-\varepsilon c_i}}{\varepsilon}\right) + \!\left(\sum_{i\geq b+1}\left(c_i + \frac{e^{-\varepsilon c_i}}{2\varepsilon}\right)\!\right)^2\right)\notag\\
    & -2q(1-q)\left(\sum_{i\leq b}\left(c_i + \frac{e^{-\varepsilon c_i}}{2\varepsilon}\right)\right)\left(\sum_{i\geq b+1}\left(c_i + \frac{e^{-\varepsilon c_i}}{2\varepsilon}\right)\right)\label{eqn:MSEb*_3}\\
    =&\left((1-q)\sum_{i\leq b} c_i - q\sum_{i\geq b+1} c_i \right)^2  + \mathcal{O}(\varepsilon^{-2}) = \left(\sum_{i\leq b} c_i - qm\right)^2  + \mathcal{O}(\varepsilon^{-2}).\label{eqn:MSEb*_res}
\end{align}
Eq.~\eqref{eqn:MSEb*_1} holds since noises are drawn independently from the DP mechanism for sanitizing each bin count in the histogram (e.g. $\mbox{Lap}(1/\varepsilon)$); and Eqs~\eqref{eqn:MSEb*_2} and \eqref{eqn:MSEb*_3} follow after plugging in Eqs~\eqref{eqn:bias_c*} and \eqref{eqn:var_c*}. 

Based on Eq.~\eqref{eqn:MSEb*_res}, expanding the LHS of Eq.~\eqref{eqn:MSE_b*} and leveraging the fact that $|X|\geq X$, and $\mathbb{E}(|X|)\geq\mathbb{E}(X)$, we have
\vspace{-2pt}
\begin{align}
&\mathbb{E}_{\mathcal{M}|\boldsymbol{\theta}}\left(\Big|\sum_{i\leq b} c^*_i - qm^*\Big| - \Big|\sum_{i\leq b} c_i - qm\Big|\right)^2 \notag\\
    = \;& \mathbb{E}_{\mathcal{M}|\boldsymbol{\theta}}\!\left(\!\sum_{i\leq b} c^*_i - qm^*\!\right)^2 \!+ \left(\sum_{i\leq b} c_i \!-\! qm\!\right)^2 \!\!-\! 2\underbrace{\Big|\sum_{i\leq b} c_i \!- qm\Big|}_{\geq \left(\sum_{i\leq b} c_i \!- qm\right)}\cdot\underbrace{\mathbb{E}_{\mathcal{M}|\boldsymbol{\theta}}\left(\Big|\sum_{i\leq b} c^*_i - qm^*\Big|\!\right)}_{\geq \mathbb{E}_{\mathcal{M}|\boldsymbol{\theta}}\left(\sum_{i\leq b} c^*_i - qm^*\!\right)}\notag\\
    \leq\; & 2\left(\sum_{i\leq b} c_i - qm\right)^2  + \mathcal{O}(\varepsilon^{-2}) \!-\! 2\left(\!\sum_{i\leq b} c_i \!-\! qm\!\right)\!\!\cdot\mathbb{E}_{\mathcal{M}|\boldsymbol{\theta}}\left((1-q)\sum_{i\leq b} c^*_i - q\sum_{i\geq b+1} c^*_i \right)\notag\\
    =\; & 2\!\left(\!\sum_{i\leq b}\!c_i \!-\! qm\!\!\right)^2 \!\!\!- 2\!\left(\!\sum_{i\leq b}\!c_i \!-\! qm\!\right)\!\!\cdot\!\!\left(\!\!(1-q)\!\!\left(\!\sum_{i\leq b}\!\left(\!c_i \!+\! \frac{e^{-\varepsilon c_i}}{2\varepsilon}\!\right)\!\right) \!- \!q \!\!\left(\!\sum_{i\geq b+1}\!\!\left(\!c_i \!+\! \frac{e^{-\varepsilon c_i}}{2\varepsilon}\!\right)\!\right)\!\right)\!+ \mathcal{O}(\varepsilon^{-2})\notag\\
    =\; & 2\left(\sum_{i\leq b} c_i \!-\! qm\!\right)^2  \!-\! 2\left(\!\sum_{i\leq b} c_i \!-\! qm\!\right)\!\!\cdot\! \!\left(\sum_{i\leq b} c_i\! -\! qm + \mathcal{O}\left(\frac{B'e^{-m\varepsilon/B'}}{\varepsilon}\!\right)\!\right)\!+\! \mathcal{O}(\varepsilon^{-2})\notag\\
    =\; & \mathcal{O}\left(\frac{B'me^{-m\varepsilon/B'}}{\varepsilon}+\varepsilon^{-2}\!\right) = \mathcal{O}\left(\frac{2G(n)m^2(u-l)e^{-\frac{\varepsilon}{2G(n)(u-l)}}}{\varepsilon}+\varepsilon^{-2}\!\right)\\
    =\; & \mathcal{O}\left(\varepsilon^{-2} + \frac{m^2e^{-\sqrt{n}\varepsilon}}{\sqrt{n}\varepsilon}\right) = \mathcal{O}\left(\varepsilon^{-2} + \frac{e^{-\sqrt{n}\varepsilon}}{h^2\sqrt{n}\varepsilon}\right).\label{eqn:MSEb*_objective}
\end{align}
The last equality in Eq.~\eqref{eqn:MSEb*_objective} holds because of the following: given bin width $h$, we require $m=(2G(n)h)^{-1}$ for DP guarantees in Eq.~\eqref{eqn:DeltaH1} by setting $\Delta_H = 1$. Denote the ``local'' bounds for the histogram $H$ after bin collapsing by $(l,u)$, we can conclude that $B' \!=\! (u-l)/h \!=\! 2G(n)m(u-l)$. Per the Bernstein-von Mises theorem, as $n \rightarrow \infty$, 
\[||f(\theta |\x)-{\mathcal {N}}({\hat {\theta }}_{n},n^{-1}I^{-1}_{\theta_0})||_{\mathrm {TVD} }= \mathcal{O}(n^{-1/2})\]
where $\hat{\theta}_n$ is the MAP and $ I_{\theta_0}$ is the Fisher information. Therefore, $u-l \asymp n^{-1/2}$, and $m/B' \asymp \sqrt{n}$.

Also, Eq.~\eqref{eqn:MSEb*_objective} captures the effects of both histogram bin granularity $h$ and privacy loss $\varepsilon$ on the accuracy of identifying the correct bin index that contains the target quantile. To ensure the second term in Eq.~\eqref{eqn:MSEb*_objective} converges to 0 as $n$ or $\varepsilon \rightarrow \infty$, a necessary upper bound on the number of posterior samples $m$ is 
\[m = o(e^{\varepsilon\sqrt{n}/2}n^{1/4}\varepsilon^{1/2}).\]
Under this condition, we can establish the MSE consistency of $|\sum_{i\leq b} c^*_i \!-\! qm^*|$ for $|\sum_{i\leq b} c_i \!-\! qm|$ for $\forall b\!\in \!\{1, \ldots, B'\}$  at the rate of Eq.~\eqref{eqn:MSEb*_objective}. 

Additionally, we show that $g(b)$ is Lipschitz continuous as follows. $\forall b, b' \!\in\!\{0, \ldots, B'+1\}$, without loss of generality, assume $b>b'$
\vspace{-5pt}
\begin{align}
   & |g(b') - g(b)| = \Bigg|\Big|\!\sum_{i\leq b} c_i \!-\! qm\Big| - \Big|\!\sum_{i\leq b'} c_i \!-\! qm\Big|\Bigg|\notag\\
    &\leq \Bigg|\!\left(\!\sum_{i\leq b} c_i \!-\! qm\!\right)\! -\! \left(\!\sum_{i\leq b'} c_i \!-\! qm\!\right)\!\Bigg|= \Big|\sum_{i=b'+1}^{b}c_i\Big| \leq (b-b')\max_{i}|c_i|\leq |b-b'|\cdot m.
\end{align}
Combined with the uniqueness of minimizers $b^*$ and $\hat{b}$, and condition on the convergence of Eq.~\eqref{eqn:MSEb*_objective}, we can conclude the DP-induced bin index mismatch error $\mathbb{E}_{\mathcal{M}|\boldsymbol{\theta}}(b^* - \hat{b})^2\rightarrow 0$  at the rate of at least Eq.~\eqref{eqn:MSEb*_objective}. 

Per step \ref{step:uniform} of Algorithm \ref{alg:PRECISE}, the privatized quantile estimate $\theta^*_{([qm])}\sim \mbox{Unif}(I_{b^*})$, where $I_{b^*} = [L+(b^*-1)h, L+b^* h]$ and $h = [2G(n)m]^{-1}$, then
\begin{align}
    &\theta^*_{([qm])}  - \left(L+(b^*-1)h\right) \sim \mbox{Unif}[0,h]; \;\; \theta_{([qm])}  - \left(L+(\hat{b}-1)h\right) \sim \mbox{Unif}[0,h].\notag\\
    &\text{Let } V_1, V_2 \sim \mbox{Unif}[0,h], \text{independent of }b^* \text{ and } \hat{b}\Rightarrow \begin{cases}
        \theta^*_{([qm])}= \left(L+(b^*-1)h\right) +V_1\\ \theta_{([qm])}  = \left(L+(\hat{b}-1)h\right)+V_2\\
    \end{cases}\!\!\!\!\!\!;\notag\\
    \;& \mathbb{E}_{\mathcal{M}|\boldsymbol{\theta}}\left(\theta^*_{([qm])}-\theta_{([qm])}\right)^2 = \mathbb{E}_{\mathcal{M}|\boldsymbol{\theta}}\left((b^* - \hat{b})h + (V_1 - V_2)\right)^2\notag\\
    =\;& h^2\mathbb{E}_{\mathcal{M}|\boldsymbol{\theta}}\left(b^* - \hat{b}\right)^2 + \mathbb{E}_{\mathcal{M}|\boldsymbol{\theta}}\left(V_1 - V_2\right)^2 + 2h\mathbb{E}_{\mathcal{M}|\boldsymbol{\theta}}\left((b^* - \hat{b})(V_1 - V_2)\right)\notag\\
    =\;& h^2\mathbb{E}_{\mathcal{M}|\boldsymbol{\theta}}\left(b^* - \hat{b}\right)^2 + \frac{2h^2}{3} - 2\cdot\frac{h}{2}\cdot\frac{h}{2}\notag \\
    =\;& h^2\mathbb{E}_{\mathcal{M}|\boldsymbol{\theta}}\left(b^* - \hat{b}\right)^2 + \underbrace{\frac{h^2}{6}}\notag\\
    & \qquad \text{discretization error and uniform sampling within the identified bin}\notag\\
    = \;&\mathcal{O}\left(m^{-2}+m^{-2}\left(\varepsilon^{-2} + \frac{m^2e^{-\sqrt{n}\varepsilon}}{\sqrt{n}\varepsilon}\right)\right).\label{eqn:CS_2}
\end{align}
Plugging Eqs~\eqref{eqn:CS_1} and \eqref{eqn:CS_2} into Eq.~\eqref{eqn:CSinequality}, we have
\begin{align}
&\mathbb{E}_{\boldsymbol{\theta}}\mathbb{E}_{\mathcal{M}|\boldsymbol{\theta}}\big(\theta^*_{(q)} - F^{-1}_{\theta|\mathbf{x}}(q)\big)^2 =\mathbb{E}_{\boldsymbol{\theta}}\mathbb{E}_{\mathcal{M}|\boldsymbol{\theta}}\left(\theta^*_{([qm])} - F^{-1}_{\theta|\mathbf{x}}(q)\right)^2 \notag\\
\leq& \underbrace{\mathcal{O}\left(m^{-2}\right)}_{T_0}+\underbrace{\mathcal{O}\left(\frac{1}{\sqrt{n}}e^{-\varepsilon\sqrt{n}/2}\right)}_{T_1} + \underbrace{\frac{q(1-q)}{m}\!\cdot\! \left(\!f_{\theta|\mathbf{x}}\!\left(\!F^{-1}_{\theta|\mathbf{x}}(q)\right)\right)^{-2}}_{T_2}.\label{eqn:thm2_res}
\end{align}
There are three types of errors in Eq.~\eqref{eqn:thm2_res}: the histogram discretization error $T_0$, the DP-induced error term \( T_1 \), and the sampling error term \( T_2 \). The dominant term among these depends on the posterior sample size $m$. The following conditions characterize the regimes where each term dominates:
\begin{align}
    &\text{If $T_2$ dominates $T_1$:} \quad \frac{1}{\sqrt{n}} e^{-\varepsilon\sqrt{n}/2}\lesssim m^{-1}n^{-1}\quad  \Rightarrow \quad m = \mathcal{O}\left(e^{\varepsilon\sqrt{n}/2} \cdot n^{-1/2}\right).\\
    &\text{If $T_2$ dominates $T_0$:} \quad m^{-2}\lesssim m^{-1}n^{-1}\quad  \Rightarrow \quad m = \Omega(n).\\
    &\text{If $T_1$ dominates $T_0$:} \quad \frac{1}{\sqrt{n}} e^{-\varepsilon\sqrt{n}/2}\lesssim m^{-2}\quad  \Rightarrow \quad m = \mathcal{O}\left(e^{\varepsilon\sqrt{n}/4} \cdot n^{1/4}\right).
\end{align}
Additionally, the validity of Eq.~\eqref{eqn:thm2_res} implicitly relies on the convergence of an intermediate DP-dependent result in Eq.~\eqref{eqn:MSEb*_objective}, which requires $m$ to satisfy $m = o\left(e^{\varepsilon\sqrt{n}/2} n^{1/4} \varepsilon^{1/2} \right).$
Together, these four constraints divide the valid range of $m$ into three asymptotic regimes, each dominated by a different error source:
\begin{align*}
&\mathbb{E}_{\boldsymbol{\theta}}\mathbb{E}_{\mathcal{M}|\boldsymbol{\theta}}\big(\theta^*_{(q)} - F^{-1}_{\theta|\mathbf{x}}(q)\big)^2\notag\\
=\: & \begin{cases}
    T_0 = \mathcal{O}(m^{-2})      &\text{if } m = o(n)\\
    T_2 = \mathcal{O}(m^{-1}n^{-1}) &\text{if } m \in \left(\Omega(n),  \mathcal{O}(e^{\varepsilon\sqrt{n}/2} \cdot n^{-1/2})\right)\\
    T_1 = \mathcal{O}\left(\frac{1}{\sqrt{n}}e^{-\varepsilon\sqrt{n}/2}\right) &\text{if } m \in \left(\Omega(e^{\varepsilon\sqrt{n}/2} \cdot n^{-1/2}), o(e^{\varepsilon\sqrt{n}/2}n^{1/4}\varepsilon^{1/2})\right)\\
\end{cases}
\end{align*}
\end{proof}

\subsection{Proof of Proposition~\ref{prop:coverage}}
\begin{proof}
Data $\mathbf{x}\sim f(X|\theta_0)$. Per the definition of $F^{-1}_{\theta|\mathbf{x}}(q)\! =\! \inf \{\theta\!:\! F(\theta|\mathbf{x})\!\geq\! q\}$, where $0\!<\!q\!<\!1$, 
\begin{align}
    &\Pr\left(F^{-1}_{\theta|\mathbf{x}}\left(\frac{\alpha}{2}\right)\leq \theta_0 \leq F^{-1}_{\theta|\mathbf{x}}\left(1-\frac{\alpha}{2}\right)\Big| \mathbf{x}\right)\notag\\
    =\; & \Pr\left(\theta_0 \leq F^{-1}_{\theta|\mathbf{x}}\left(1-\frac{\alpha}{2}\right)\Big| \mathbf{x}\right)  - \Pr\left(\theta_0 \leq F^{-1}_{\theta|\mathbf{x}}\left(\frac{\alpha}{2}\right)\Big| \mathbf{x}\right)\notag\\
    =\; & F_{\theta|\mathbf{x}}\left(F^{-1}_{\theta|\mathbf{x}}\left(1-\frac{\alpha}{2}\right)\right) - F_{\theta|\mathbf{x}}\left(F^{-1}_{\theta|\mathbf{x}}\left(\frac{\alpha}{2}\right)\right)= 1-\alpha.
\end{align}
Following  Eq.~\eqref{eqn:thm2_res}, as $n \rightarrow \infty$ or $\varepsilon \rightarrow \infty$
\begin{align}
    \Pr\left(\theta^*_{([\frac{\alpha}{2}m])}\leq \theta_0 \leq \theta^*_{([(1-\frac{\alpha}{2})m])}\Big| \mathbf{x}\right)\rightarrow \Pr\left(F^{-1}_{\theta|\mathbf{x}}\left(\frac{\alpha}{2}\right)\leq \theta_0 \leq F^{-1}_{\theta|\mathbf{x}}\left(1-\frac{\alpha}{2}\right)\Big| \mathbf{x}\right)=1-\alpha.
\end{align}
\end{proof}
\subsection{PrivateQuantile and its MSE consistency}\label{ape:PQ}

\begin{algorithm}[H]
\caption{PrivateQuantile of $\varepsilon$-DP \citep{smith2011privacy}}\label{alg:PQ}
\SetAlgoLined
\SetKwInOut{Input}{input}
\SetKwInOut{Output}{output}
\Input{data $\x \!=\! \{x_i\}_{i=1}^n$, privacy loss parameter $\varepsilon$,  quantile $q \!\in\! (0,1)$, global bounds $(L_{\x},U_{\x})$ for $\x$.}
\Output{ PP $q^{th}$ quantile estimate $x^*_{([qn])}$ of $\varepsilon$-DP.}
Sort $\x$ in ascending order $x_{(1)},\ldots,x_{(n)}$\;
Replace $x_i \!<\! L_{\x}$ with $L$ and $x_i \!>\! U_{\x}$ with $U_{\x}$\;
For $i=0,\ldots,n$, define $y_{i}\triangleq\left(x_{(i+1)}-x_{(i)}\right)\exp(-\varepsilon|i-q n|/2)$\;
Sample an integer $i^* \!\in\! \{0,\ldots,n\}$ with probability $y_i/\left(\sum_{i=0}^n y_i\right)$\;\label{step:i^*}
Draw $x^*_{([qn])}$ from Unif$\left(x_{(i^*)},x_{(i^*+1)}\right)$. 
\end{algorithm}

\begin{theorem}[MSE consistency of PrivateQuantile]  \label{thm:PQ}
Denote the sensitive dataset by $\x=\{x_i\}_{i=1}^n$. Let $x_{([qn])}^*$ be the private $q$-th sample quantile of $\x$ from $\mathcal{M}$: PrivateQuantile of $\varepsilon$-DP \citep{smith2011privacy}. Under Assumption \ref{con:PQ}, and assume $\exists \text{ constant } C \geq 0$ such that global bounds $(L_{\x}, U_{\x})$ for $\x$ satisfy $\lim_{n\rightarrow\infty}\left(U_{\x}-x_{(n)}\right) =\lim_{n\rightarrow\infty}\left(x_{(1)}-L_{\x}\right) =C$, then
\begin{align} \label{Eqn:MSE_PQ}
&\mathbb{E}_{\x}\mathbb{E}_{\mathcal{M}|\x}\left(x_{([qn])}^*-F_{\x}^{-1}(q)\right)^2= \mathcal{O}(n^{-1})+O\left(e^{-\mathcal{O}(n\varepsilon)}n^{-3/2}\right).
\end{align}
\end{theorem}
The detailed proof is provided below. Briefly, Eq.~\eqref{Eqn:MSE_PQ} suggests that the MSE between $x_{([qn])}^*$ and $F_{\x}^{-1}(q)$ can be decomposed into two components: (1) the MSE between $x_{([qn])}^*$ and $x_{([qn])}$ introduced by DP sanitization noise that converges at rate $\mathcal{O}(e^{-\mathcal{O}(n\varepsilon)n^{-3/2}})$ and (2) the MSE between $x_{([qn])}$ and $F_{\x}^{-1}(q)$ due to the sampling error that converges at rate $\mathcal{O}(n^{-1})$). The faster convergence rate of the former implies that the sampling error, rather than the sanitization error, is the limiting factor in the convergence of $x_{([qn])}^*$ to $F_{\x}^{-1}(q)$.

\begin{con}\label{con:PQ}
    Let $x_{(1)} \leq x_{(2)} \leq \ldots \leq x_{(n)}$ be the order statistics of a random sample $x_1, \ldots, x_n$ from a continuous distribution $f_{\x}$, and  $F_{\x}^{-1}(q)\!=\!\inf \{x: F_{\x}(x) \geq q\}$ be the unique quantile at $q$, where $0\!<\!q\!<\!1$ and $F_{\x}$ is the CDF. Assume $f_{\x}$ is positive, finite, and continuous at $F_{\x}^{-1}(q)$.
\end{con}

\begin{lemma}[Asymptotic distribution of the spacing between two consecutive order statistics]\label{lemma:asym_gap}
Let $\x=(x_1,\ldots, x_n)$ be a sample from a continuous distribution $f_{\x}$, and $x_{([qn])}$ be the sample quantile at $q$ and $x_{([qn]+1)}$ be the value immediately succeeding  $x_{([qn])}$. Given the regularity conditions in Assumption \ref{con:PQ}, 
\begin{equation}\label{eqn:asym_gap}
 n\cdot(x_{([qn]+1)}-x_{([qn])})\cdot f_{\x}(F_{\x}^{-1}(q))\stackrel{d}{\longrightarrow}\mbox{exp}(1) \mbox{ as } n\rightarrow \infty.
 \end{equation}
\end{lemma}
\begin{proof}
Given a sample $X=(x_1,\ldots, x_n)$, since $\lim _{n \rightarrow \infty}[qn]/n=q \!\in\! (0,1)$, per Thm. 3 in \citep{smirnov1949limit}, 
\begin{equation}\label{eqn:smirnov}
    x_{([qn])} \stackrel{a . s.}{\longrightarrow} F_{\x}^{-1}(q) \text{ as } n\rightarrow\infty
\end{equation}
at rate $n^{-1/2}$. Let $y_{([qn])}$ be the $[qn]^{th}$ order statistic in a random sample of size $n$ from Uni$(0,1)$. From Eq.~\eqref{eqn:smirnov}, it follows that $y_{([qn]+1)}-y_{([qn])}\stackrel{a . s.}{\longrightarrow}0 \text{ as } n\rightarrow\infty$. Then, per Lemma 1 in \citep{nagaraja2015spacings},
\vspace{-5pt}
\begin{align}\label{eqn:uniform}
n \cdot(y_{([q n]+1)}-y_{([qn])})\stackrel{d}{\longrightarrow}\mbox{exp}(1),
\end{align} 
where $\mbox{exp}(1)$ represents an exponential random variable with rate parameter 1.

In addition, $\forall 1\leq i \leq n$, $x_{(i)} \stackrel{d}{=}F_{\x}^{-1}\left(y_{(i)}\right)$. Then 
\begin{align}\label{eqn:gaptouniform}
(x_{([qn]+1)}-x_{([qn])}) &
\stackrel{d}{=}F_{\x}^{-1}\left(y_{([qn]+1)}\right)-F_{\x}^{-1}\left(y_{([qn])}\right),\notag\\
n\cdot(x_{([qn]+1)}-x_{([qn])})&
\stackrel{d}{=}\frac{F_{\x}^{-1}\left(y_{([qn]+1)}\right)-F_{\x}^{-1}\left(y_{([qn])}\right)}{\left(y_{([qn]+1)}-y_{([qn])}\right)} \cdot n \cdot\left(y_{([qn]+1)}-y_{([qn])}\right),
\end{align}
where $\stackrel{d}{=}$ stands for ``equal in distribution'', meaning two random variables have the same distribution. Per the definition of  pdf and the assumptions around $f_{\x}$,
\begin{equation} \label{eqn:derivative}
\frac{F_{\x}^{-1}\left(y_{([q n]+1)}\right)-F_{\x}^{-1}\left(y_{([q n])}\right)}{\left(y_{([qn]+1)}-y_{([qn])}\right)} \stackrel{a . s.}{\longrightarrow} \frac{1}{f_{\x}(F_{\x}^{-1}(q))}.
\end{equation}
Plugging Eqns \eqref{eqn:uniform} and \eqref{eqn:derivative} into Eq.  \eqref{eqn:gaptouniform}, along with Slutsky's Theorem, we have
\begin{equation}
\begin{aligned}\label{eqn:exp(1)}
n\cdot(x_{([qn]+1)}-x_{([qn])}) & \stackrel{d}{\longrightarrow}(f_{\x}(F_{\x}^{-1}(q)))^{-1}\mbox{exp}(1),\\
\mathbb{E}[n(x_{([qn]+1)}-x_{([qn])})]&\longrightarrow(f_{\x}(F_{\x}^{-1}(q)))^{-1},\\
\mathbb{E}[n(x_{([qn]+1)}-x_{([qn])})]^2\!&\longrightarrow2(f_{\x}(F_{\x}^{-1}(q)))^{-2}.  
\end{aligned}
\end{equation}
\end{proof}

\begin{theorem}[MSE consistency of PrivateQuantile in Algorithm \ref{alg:PQ}]  \label{convergence_PQ}
Denote the sample data of size $n$ by $\x$ and let $x_{([qn])}^*$ be the sanitized $q^{th}$ sample quantile of $X$  from $\mathcal{M}$: \verb|PrivateQuantile| of $\varepsilon$-DP in Algorithm \ref{alg:PQ}. Under the regularity conditions in Assumption \ref{con:PQ}, and assume that $\exists \text{ constant } C \geq 0$ such that the user-provided global bounds $(L_{\x}, U_{\x})$ for $\x$ satisfy  \\ $\lim_{n\rightarrow\infty}\left(U_{\x}\!-\!x_{(n)}\right) \! = \! \lim_{n\rightarrow\infty}\left(x_{(1)}\!-\!L_{\x}\right) \!=\!C$, then
\begin{align}\label{Eqn:MSE_PQ_DP}
&\mathbb{E}_{\x}\mathbb{E}_{\mathcal{M}|\x}\!\!\left(\!x_{([qn])}^*\!-\!F_{\x}^{-1}(q)\!\right)^2\!=\! \mathcal{O}(n^{-1})\!+O\!\left(\!\frac{e^{-\mathcal{O}(n\varepsilon)}}{n^{3/2}}\!\right)\!\!=\!\begin{cases}
\!\mathcal{O}(n^{-1}) & \mbox{for constant $\varepsilon$} \\
\!\mathcal{O}(e^{-\mathcal{O}(\varepsilon)}) & \mbox{for constant $n$}
\end{cases}.
\end{align}
If the PrivateQuantile procedure $\mathcal{M}$ of $\rho$-zCDP is used, then \begin{align}\label{Eqn:MSE_PQ_zCDP}
&\mathbb{E}_{\x}\mathbb{E}_{\mathcal{M}|\x}\!\!\left(\!x_{([qn])}^*\!-\!F_{\x}^{-1}(q)\!\right)^2\!=\! \mathcal{O}(n^{-1})\!+O\!\left(\!\frac{e^{-\mathcal{O}(n\sqrt{\rho})}}{n^{3/2}}\!\right)\!\!=\!\begin{cases}
\!\mathcal{O}(n^{-1}) & \mbox{for constant $\rho$} \\
\!\mathcal{O}(e^{-\mathcal{O}(\sqrt{\rho})}) & \mbox{for constant $n$}
\end{cases}.
\end{align}
\end{theorem}

\begin{proof} Let $\mathcal{M}$ standards for the PrivateQuanitile procedure in Algorithm \ref{alg:PRECISE} through this section and $x_{([qn])}$ be the original sample quantile at $q$. Similar to proof in Appendix \ref{proof:PRECISE}, we first expand the MSE between the sanitized $q^{th}$ sample quantile $x_{([qn])}^*$ and the population quantile $F_{\x}^{-1}(q)$ as
\vspace{-6pt}
\begin{align} 
&\mathbb{E}_{\x} \mathbb{E}_{\mathcal{M}|\x}\left(x_{([qn])}^*-F_{\x}^{-1}(q)\right)^2 \notag =\mathbb{E}_{\x} \mathbb{E}_{\mathcal{M}|\x}\left(x_{([qn])}^* - x_{([qn])} + x_{([qn])}-F_{\x}^{-1}(q)\right)^2\\
= \; & \mathbb{E}_{\x} \mathbb{E}_{\mathcal{M}|\x}\left(x_{([qn])}^*-x_{([qn])}\right)^2 + \mathbb{E}_{\x} \mathbb{E}_{\mathcal{M}|\x}\left(x_{([qn])}-F_{\x}^{-1}(q) \right)^2\notag\\
& \; +2\mathbb{E}_{\x} \mathbb{E}_{\mathcal{M}|\x}\left(x_{([qn])}^*-x_{([qn])}\right)\left(x_{([qn])}-F_{\x}^{-1}(q) \right)\notag\\
\leq \; &\mathbb{E}_{\x} \mathbb{E}_{\mathcal{M}|\x}\left(x_{([q P])}^*-x_{([qn])}\right)^2 + \mathbb{E}_{\x}\left(x_{([qn])}-F_{\x}^{-1}(q) \right)^2 \notag\\
& \; +2\sqrt{\mathbb{E}_{\x} \mathbb{E}_{\mathcal{M}|\x}\left(x_{([q P])}^*-x_{([qn])}\right)^2\mathbb{E}_{\x} \left(x_{([qn])}-F_{\x}^{-1}(q) \right)^2}.
\label{eqn:CS_PQ}
\end{align}
The last inequality in Eq.~\eqref{eqn:CS_PQ} holds per the Cauchy-Schwarz inequality. Similarly based on Thm. 1 in \citep{walker1968note}, we have asymptotic normality for sample quantile 
\begin{align}
\sqrt{n}\left(x_{([qn])}-F_{\x}^{-1}(q) \right) &\stackrel{d}{\rightarrow} \mathcal{N}\left(0, \frac{q(1-q)}{\{f_{\x}(F_{\x}^{-1}(q))\}^2}\right),\notag\\
\Rightarrow \; \mathbb{E}_{\x}\mathbb{E}_{\mathcal{M}|\x}\left(x_{([qn])}-F_{\x}^{-1}(q) \right)^2 &=\mathbb{E}_{\x}\left(x_{([qn])}-F_{\x}^{-1}(q) \right)^2 \rightarrow \frac{n^{-1} q(1-q)}{\{f_{\x}(F_{\x}^{-1}(q))\}^2}.\label{eqn:popuquantile}
\end{align}

An intermediate step of the PrivateQuantile procedure is the sampling of index $i^*$ via the exponential mechanism with privacy loss $\varepsilon$ (step~\ref{step:i^*} in Algorithm \ref{alg:PQ}),
\begin{align}
    \Pr(i^*) &= \frac{\left(x_{(i^*+1)} - x_{(i^*)}\right)\exp(-\varepsilon|i^*-[qn]|/2)}{\sum_{i=0}^n \left(x_{(i+1)} - x_{(i)}\right)\exp(-\varepsilon|i-[qn]|/2)}\label{eqn:Pri*}\\
    \text{where }& \sum_{i=0}^n \left(x_{(i+1)} - x_{(i)}\right)\exp(-\varepsilon|i-[qn]|/2)\label{eqn:Pri*_denominator}\\
    =\; & \left(x_{(1)} - L_{\x}\right)\exp\left(-\frac{\varepsilon\cdot[qn]}{2}\right) + \left(U_{\x} - x_{(n)}\right)\exp\left(-\frac{\varepsilon|n-[qn]|}{2}\right)\label{eqn:Pri*_UL}\\
    & + \sum_{i\notin \{0, n, [qn]\}}\left(x_{(i+1)} - x_{(i)}\right)\exp(-\varepsilon|i-[qn]|/2) + \left(x_{([qn]+1)} - x_{([qn])}\right)\label{eqn:Pri*_middle}.
\end{align}

For $i\!\notin\! \{0, n, [qn]\}$, per Lemma~\ref{lemma:asym_gap}, $\left(x_{(i+1)} \!-\! x_{(i)}\right)\!\rightarrow \!0$ at the rate of $n^{-1}\!$ as $n\!\rightarrow\!\infty$, so the first term in Eq.~\eqref{eqn:Pri*_middle} converges to $0$ at the rate of $\mathcal{O}(n^{-1}e^{-\mathcal{O}(n\varepsilon)})$, while the second term in Eq.~\eqref{eqn:Pri*_middle} converges to $0$ at the rate of $\mathcal{O}(n^{-1})$. 

Also, per the assumption that $\exists \text{ constant } C \geq 0$ such that the user-provided global bounds $(L_{\x}, U_{\x})$ for $\x$ satisfy $\lim_{n\rightarrow\infty}\left(U_{\x}\!-\!x_{(n)}\right) \! = \! \lim_{n\rightarrow\infty}\left(x_{(1)}\!-\!L_{\x}\right) \!=\!C$. The two terms in Eq.~\eqref{eqn:Pri*_UL} $\approx C\cdot e^{-\mathcal{O}(n\varepsilon)}$. Therefore, as $n\rightarrow \infty$ or $\varepsilon \rightarrow \infty$,
\begin{align}
    &\Pr(i^* = [qn]) = \frac{\mathcal{O}(n^{-1})}{\mathcal{O}(n^{-1}+n^{-1}e^{-\mathcal{O}(n\varepsilon)}) + C\cdot e^{-\mathcal{O}(n\varepsilon)}} \rightarrow 1.\label{eqn:Prqn}\\
    \Rightarrow \; & \Pr\left(x_{([qn])}^* \sim \mbox{Unif}\left(x_{([qn])},x_{([qn]+1)}\right)\right)\rightarrow 1\label{Eqn:limitdist}.
\end{align}
Eqns \eqref{eqn:Prqn} and \eqref{Eqn:limitdist} imply the limiting distribution of $x_{([qn])}^*$ is a uniform distribution from $x_{([qn])}$ to $x_{([qn]+1)}$, achieved at the rate of $e^{\mathcal{O}(n\varepsilon)}$. Define $h\triangleq x_{([qn])}^*-x_{([qn])}$, then
\begin{equation}
e^{\mathcal{O}(n\varepsilon)}h \overset{d}{\rightarrow}\mbox{Unif}\left(0, x_{([qn]+1)}-x_{([qn])}\right).
\end{equation}
Therefore, as $n\rightarrow \infty$ or $\varepsilon\rightarrow \infty$
\begin{align}\label{eqn:sanitizetopublic}
&\mathbb{E}_{\x}\mathbb{E}_{\mathcal{M}|\x} \left(x_{([qn])}^*-x_{([qn])}\right)^2= \mathbb{E}_{\x}\mathbb{E}_{\mathcal{M}|\x}(h^ 2)=\mathbb{E}_{\x} \{\mathbb{V}_{\mathcal{M}|X}(h)+ (\mathbb{E}_{\mathcal{M}|\x}(h))^2\}\notag \\
\rightarrow &\; e^{-\mathcal{O}(n\varepsilon)}\mathbb{E}_{\x}\left[\frac{\left(x_{([qn]+1)}-x_{([qn])}\right)^2}{12} + \frac{\left(x_{([qn]+1)}-x_{([qn])}\right)^2}{4}\right]\notag\\ 
=&\;e^{-\mathcal{O}(n\varepsilon)} \mathbb{E}_{\x}\left[\frac{(x_{([qn]+1)}-x_{([qn])})^2}{3}\right]\rightarrow\frac{2\cdot n^{-2}e^{-\mathcal{O}(n\varepsilon)}}{3(f_{\x}(F_{\x}^{-1}(q)))^2} \mbox{ per Eq.~\eqref{eqn:exp(1)} in Lemma~\eqref{lemma:asym_gap}}.
\end{align}

Plugging Eqns~\eqref{eqn:sanitizetopublic} and \eqref{eqn:popuquantile} into the right-hand side of Eq.~\eqref{eqn:CS_PQ}, we have 
\begin{equation}
\begin{aligned}
& \qquad \mathbb{E}_{\x} \mathbb{E}_{\mathcal{M}|\x}\left(x_{([qn])}^*-F_{\x}^{-1}(q)\right)^2 \\
& \leq \frac{2\cdot n^{-2}e^{-\mathcal{O}(n\varepsilon)}}{3(f_{\x}(F_{\x}^{-1}(q)))^2} + \frac{n^{-1}q(1-q)}{\{f_{\x}(F_{\x}^{-1}(q))\}^2}+ 2\sqrt{\frac{2\cdot n^{-2}e^{-\mathcal{O}(n\varepsilon)}}{3(f_{\x}(F_{\x}^{-1}(q)))^2}  \cdot \frac{n^{-1}q(1-q)}{\{f_{\x}(F_{\x}^{-1}(q))\}^2}}\\
& = \mathcal{O}(e^{-\mathcal{O}(n\varepsilon)}n^{-2}+n^{-1}+e^{-\mathcal{O}(n\varepsilon)}n^{-3/2})\\
&=\! \mathcal{O}(n^{-1})\!+O\!\left(\!\frac{e^{-\mathcal{O}(n\varepsilon)}}{n^{3/2}}\!\right)\!\!=\!\begin{cases}
\!\mathcal{O}(n^{-1}) & \mbox{for constant $\varepsilon$} \\
\!\mathcal{O}(e^{-\mathcal{O}(\varepsilon)}) & \mbox{for constant $n$}
\end{cases}.
\end{aligned}
\end{equation}
\end{proof}


\subsection{Proof of Theorem \ref{thm:utility_PP}}\label{proof:PPQ}
We first present Lemmas \ref{lemma:chernoff} and \ref{lemma:quantile} that will be used in proving Theorem \ref{thm:utility_PP}.
\begin{lemma}[Chernoff bounds \citep{mitzenmacher2017probability}]\label{lemma:chernoff}
    Let $Z_i \stackrel{iid}\sim \mbox{Bernoulli}(p)$ and $Z = \sum_{i=1}^n \!Z_i$, then for $\delta \!\in\![0,1],$
    \begin{align*}
        \mathbb{P}(Z \geq (1+\delta)np) &\leq e^{-np\delta^2/3};\\
        \mathbb{P}(Z \leq (1-\delta)np) &\leq e^{-np\delta^2/2}.
    \end{align*}
\end{lemma}

\begin{lemma}[sample quantile is concentrated around the population quantile]\label{lemma:quantile}
    Let $\boldsymbol{\theta} \!=\!(\theta_1, \ldots, \theta_m)$ be a set of samples from posterior distribution with CDF $f_{\theta|\mathbf{x}}$, and $F^{-1}_{\theta|\mathbf{x}}(q)\! =\! \inf \{\theta\!:\! f_{\theta|\mathbf{x}}\!\geq\! q\}$ where $0\!<\!q\!<\!1$. Assume the posterior density $f_{\theta|\mathbf{x}}$ is continuous at $F^{-1}_{\theta|\mathbf{x}}(q)$. Let $\eta>0$ and $0\leq u\leq \eta$ and $p_{\min} = \inf_{|\tau-F^{-1}_{\theta|\mathbf{x}}(q)|\leq 2\eta} f_{\theta|\mathbf{x}}(\tau)$, then
    \[\mathbb{P}\left(\Big|\theta_{([qm])} - F^{-1}_{\theta|\mathbf{x}}(q)\Big| > u\right)\leq
    \begin{cases}
        2\exp\left(-mu^2p^2_{\min}/2q\right) &\text{ if 
 }\frac{3}{5}<q<1; \\
        2\exp\left(-mu^2p^2_{\min}/3(1-q)\right) &\text{ if 
 } 0<q<\frac{3}{5}.
    \end{cases}\]
\end{lemma}
\begin{proof}
    Let $Z_j = 1\{\theta_{(j)} > F^{-1}_{\theta|\mathbf{x}}(q) + u\}$ and $Z = \sum_{j=1}^m Z_j$ denote the number of posterior samples larger than $F^{-1}_{\theta|\mathbf{x}}(q) + u$. Then
    \vspace{-5pt}
    \[\hat{p} = \mathbb{P}(Z_j = 1) \leq 1-q-up_{\min}.\]
    If $\theta_{([qm])} > F^{-1}_{\theta|\mathbf{x}}(q) + u$, then $Z \geq (1-q)m$, therefore per Chernoff bound in Lemma~\ref{lemma:chernoff},
    \begin{align}
        &\mathbb{P}\left(\theta_{([qm])} > F^{-1}_{\theta|\mathbf{x}}(q) + u\right)
        \leq \mathbb{P}\left(Z \geq (1-q)m\right) =\mathbb{P}\left(Z \geq \left(1+\frac{1-q}{\hat{p}}-1\right)m\hat{p}\right)\notag\\
        \leq & \exp\left(-\frac{m\hat{p}}{3}\left(\frac{1-q}{\hat{p}}-1\right)^2\right) 
        = \exp\left(-\frac{m}{3\hat{p}}\left(\underbrace{1-q-\hat{p}}_{\geq up_{min}}\right)^2\right)\notag\\
        \leq&\exp\left(-\frac{mu^2p^2_{\min}}{3\hat{p}}\right)
        \leq \exp\left(-\frac{mu^2p^2_{\min}}{3(1-q)}\right).
    \end{align}
    Similarly, let $Z'_j = 1\{\theta_j < F^{-1}_{\theta|\mathbf{x}}(q) - u\}$ and $Z' = \sum_{j=1}^m Z'_j$ denote the number of posterior samples smaller than $F^{-1}_{\theta|\mathbf{x}}(q) - u$. Then $\hat{p}' = \mathbb{P}(Z'_j = 1) \leq q-up_{\min}.$
    If $\theta_{([qm])} < F^{-1}_{\theta|\mathbf{x}}(q) - u$, then $Z' \geq qm$; therefore, per the Chernoff bound in Lemma~\ref{lemma:chernoff},
    \begin{align}
        &\mathbb{P}\left(\theta_{([qm])} \!<\! F^{-1}_{\theta|\mathbf{x}}(q)\! -u\right)
        \!\leq\! \mathbb{P}\left(Z' \!\geq\! qm\right) 
        \!=\!\mathbb{P}\!\left(\!\!Z' 
        \!\leq \!\left(1\!+\!\frac{q}{\hat{p}'}-1\!\right)\!m\hat{p}'\!\right)\!\leq \exp\!\left(\!-\frac{m\hat{p}'}{2}\!\!\left(\!\frac{q}{\hat{p}'}\!-\!1\!\right)^2\!\right)\notag\\
        = & \exp\!\left(\!-\frac{m}{2\hat{p}'}\!\!\left(\underbrace{q-\hat{p}'}_{\geq up_{min}}\right)^2\right)\!\leq  \exp\!\left(\!-\frac{mu^2p^2_{\min}}{2\hat{p}'}\!\right)\!
        \leq\! \exp\left(\!-\frac{mu^2p^2_{\min}}{2q}\!\right).
    \end{align}
If $3(1-q)<2q \Leftrightarrow \frac{3}{5}<q<1$, then
\[\mathbb{P}\left(|\theta_{([qm])} - F^{-1}_{\theta|\mathbf{x}}(q)| > u\right)\leq 2\exp\left(-\frac{mu^2p^2_{\min}}{2q}\right)\leq 2\exp\left(-\frac{mu^2p^2_{\min}}{3(1-q)}\right);\]
otherwise, if $0<q<\frac{3}{5}$,
\[\mathbb{P}\left(|\theta_{([qm])} - F^{-1}_{\theta|\mathbf{x}}(q)| > u\right)\leq 2\exp\left(-\frac{mu^2p^2_{\min}}{3(1-q)}\right)\leq 2\exp\left(-\frac{mu^2p^2_{\min}}{2q}\right).\]
\end{proof}


We can now move on to the proof of Theorem \ref{thm:utility_PP}.
\begin{proof}
Since $\theta_{(q)} = \theta_{([qm])}$ and $\theta^*_{(q)} = \theta^*_{([qm])}$, we the notations interchangeably  for the non-private and PP $q^{th}$ sample quantiles. The proof is inspired \citet{asi2020near}, with substantial extensions to address our specific problem.

First, we divide the interval $[\theta_{(k)}-\eta, \theta_{(k)} + \eta]$ to blocks of size $u$: $I_1, I_2, \ldots, I_{2\eta/u}$. Let $N_i$ denote the number of elements in $I_i$. We also define the following three events:
    \begin{align*}
    A &= \{\forall i, N_i \geq (m+1)u p_{\min}/2\};\\
    B &= \{|\theta_{(k)} - F^{-1}_{\theta|\mathbf{x}}(q)| \leq \eta/2\};\\
    D &= \{|\theta_{([qm])} - F^{-1}_{\theta|\mathbf{x}}(q)| \leq \eta/2\}.
    \end{align*}
Recall that $k \!=\! \arg\min_{j\in \{0, 1, \ldots, m+1\}}{|\theta_{(j)}\!-\!F^{-1}_{\theta|\mathbf{x}}(q)|}$ as defined in Algorithm \ref{alg:DPP}, thus $|\theta_{(k)} - F^{-1}_{\theta|\mathbf{x}}(q)| \leq |\theta_{([qm])} - F^{-1}_{\theta|\mathbf{x}}(q)| \Rightarrow D\subset B \Rightarrow \mathbb{P} (D) \leq \mathbb{P} (B) \Rightarrow \mathbb{P} (D^c) \geq \mathbb{P} (B^c)$.

Next, we derive a lower bound for $\mathbb{P}(A|B)$:
    \begin{align}\label{eqn:A|B inequality}
        \mathbb{P}(A|B)
        &\geq \mathbb{P}(A|B)\mathbb{P}(B) = \mathbb{P}(A) - \mathbb{P}(A|B^c)\mathbb{P}(B^c) \geq \mathbb{P}(A) - \mathbb{P}(B^c) \geq \mathbb{P}(A) - \mathbb{P}(D^c). 
    \end{align}
    For $\mathbb{P}(D^c)$, per Lemma~\ref{lemma:quantile},
    \begin{equation}\label{eqn:D^c}
        \mathbb{P}(D^c) \leq 2\exp\left(-\frac{m\eta^2p^2_{\min}}{12(1-q)}\right).
    \end{equation}
    For $\mathbb{P}(A)$, we first let $Z_j \!=\! 1\{\theta_{(j)} \!\in\! I_i\}$, then $N_i \!=\! \sum_{j=0}^mZ_j$. As $\hat{p}\!=\!\mathbb{P}(Z_j\!=\!1) \!\geq\! u p_{\min}$, per the Chernoff bound in Lemma~\ref{lemma:chernoff},
    \begin{align}
        \mathbb{P}\!\left(\!N_i\!<\!\frac{(m+1)up_{min}}{2}\!\right)
        \!=\;&\mathbb{P}\left(N_i<(m+1)\hat{p}\left(1-(1-\frac{u p_{\min}}{2\hat{p}})\right)\right)\notag\\
        \leq\;&\exp\!\left(\!-\frac{(m\!+\!1)\hat{p}}{2}\!\left(\underbrace{1\!-\!\frac{u p_{\min}}{2\hat{p}}}_{\geq 1/2}\right)^2\right)
        \!\leq \exp\!\left(\!-\frac{(m\!+\!1)u p_{\min}}{8}\!\right).
    \end{align}
    By taking a union bound across all blocks,
    \begin{equation}\label{eqn:A^c}
        \mathbb{P}(A^c)
        \leq \frac{2\eta}{u}\exp\left(-\frac{(m+1)u p_{\min}}{8}\right), 
    \end{equation}
    and thus
    \begin{equation}\label{eqn:A}
        \mathbb{P}(A)
        \geq 1-\frac{2\eta}{u}\exp\left(-\frac{(m+1)u p_{\min}}{8}\right).
    \end{equation}
    Plug Eqs.~\eqref{eqn:D^c} and \eqref{eqn:A} into the RHS of Eq.~\eqref{eqn:A|B inequality},
    \begin{align}
        \mathbb{P}(A|B)
        &\geq 1-\frac{2\eta}{u}\exp\left(-\frac{(m+1)u p_{\min}}{8}\right) - 2\exp\left(-\frac{m\eta^2p^2_{\min}}{12(1-q)}\right),\notag\\
        \Rightarrow\mathbb{P}(A^c|B)
        &\leq \frac{2\eta}{u}\exp\left(-\frac{(m+1)u p_{\min}}{8}\right) + 2\exp\left(-\frac{m\eta^2p^2_{\min}}{12(1-q)}\right).\label{eqn:A|B}
    \end{align}
    Next, we establish the following inequality for later use. For any event $E$,
    \begin{align}
        \mathbb{P}(E) &= \mathbb{P}(E|A\cap B)\mathbb{P}(A\cap B) + \mathbb{P}(E|(A\cap B)^c)\mathbb{P}((A\cap B)^c\notag\\
        &\leq \mathbb{P}(E|A\cap B) + \mathbb{P}((A\cap B)^c)\notag\\
        &=\mathbb{P}(E|A\cap B) + \mathbb{P}(A^c\cup B^c)\notag\\
        &=\mathbb{P}(E|A\cap B) + \mathbb{P}(B^c) + \mathbb{P}(A^c) - \mathbb{P}(A^c\cap B^c)\notag\\
        &=\mathbb{P}(E|A\cap B) + \mathbb{P}(B^c) + \mathbb{P}(A^c\cap B)\notag\\
        &\leq\mathbb{P}(E|A\cap B) + \mathbb{P}(B^c) + \mathbb{P}(A^c|B)\notag\\
        &\leq \mathbb{P}(E|A\cap B) + \mathbb{P}(D^c) + \mathbb{P}(A^c|B).\label{eqn:E inequality}
    \end{align}
   If both events $A$ and $B$ occur, then for any $\theta_{(j^*)}$ such that $|\theta_{(j^*)}\!-\!\theta_{(k)}|\!>\! 2u$, there are at least $(m+1)u p_{\min}/2$ elements between $\theta_{(k)}$ and $\theta_{(j^*)}$. This implies that $|j^* - k|\geq (m+1)up_{\min}/2$. Therefore, if $C(m,\varepsilon)\triangleq\sum_{i=0}^{m} (\theta_{(i+1)} - \theta_{(i)})\cdot \mbox{exp}(-\frac{\varepsilon}{2(m+1)}|i-k|)$,
    \begin{align}\label{eqn:DPPj^*}
        &\exp\!\left(\!-\frac{\varepsilon}{2(m+1)}|j^*\!-\!k| \right)
        \leq \exp\left(-\frac{\varepsilon (m+1)u p_{\min}}{4(m+1)}\right) = \exp\left(-\frac{\varepsilon u p_{\min}}{4}\right)\\
        \Rightarrow & \: \mathbb{P}(j^*|A, B)\leq \frac{\theta_{(j^*+1)}-\theta_{(j^*)}}{C(m,\varepsilon)}\exp\left(-\frac{\varepsilon u p_{\min}}{4}\right).
    \end{align}
    Let $j^*_{\max} \!\triangleq\! \arg\min_{j}\{\theta_{(j)}\!-\!\theta_{(k)}\!>\!2u\}$ and $j^*_{\min} \!\triangleq\! \arg\max_{j}\{\theta_{(j)} -\theta_{(k)} \!<\! -2u\}$. Then, 
    \[\sum_{|\theta_{(j^*\!)} - \theta_{(k)}| > 2u}\!\!\!\!\!\!\mathbb{P}(j^*|A, B) \leq \frac{e^{-\varepsilon u p_{\min}/4}}{C(m, \varepsilon)}\left(U - \theta_{(j^*_{\max})} \!\!+ \theta_{(j^*_{\min})}\!\!- L\right).\]
    WLOG, assume $2k \leq m+1$, let $s \!=\! \min_{i \in \{0, 1, \ldots, m\}}(\theta_{(i+1)} \!-\! \theta_{(i)})$ and $\xi = \mbox{exp}(-\frac{\varepsilon}{2(m+1)})$,
    \begin{align}
        C(m,\varepsilon)
        =& \sum_{i=0}^{m} (\theta_{(i+1)} - \theta_{(i)})\cdot \mbox{exp}\left(-\frac{\varepsilon}{2(m+1)}|i-k|\right)\notag\\
        \geq\;& s\left(1+2\xi + 2\xi^2 + 2\xi^3 + \cdots + 2\xi^{k-1} +2\xi^{k} + \xi^{k+1} + \cdots \xi^{m-k}\right)\notag\\
        \geq \;& s\left(1+2\frac{\xi(1-\xi^{k})}{1-\xi} + \frac{\xi^{k+1}(1-\xi^{m-2k})}{1-\xi}\right)\notag\\
        =\;&s\;\frac{1+\xi-\xi^{k+1}-\xi^{m-k+1}}{1-\xi}.
    \end{align}

     Since $\theta_{(j^*_{\max})} - \theta_{(j^*_{\min})} = \theta_{(j^*_{\max})} - \theta_{(k)} + \theta_{(k)} -\theta_{(j^*_{\min})} > 4u$, 
    \begin{align}
        \sum_{|\theta_{(j^*)} - \theta_{(k)}| > 2u}\!\!\!\!\!\!\mathbb{P}(j^*|A, B) &\leq \frac{e^{-\varepsilon u p_{\min}/4}}{C(m,\varepsilon)}\left(U - \theta_{(j^*_{\max})} \!\!+ \theta_{(j^*_{\min})}\!\!- L\right)\notag\\
        &\leq \frac{U-L - 4u}{s}\cdot\frac{1-\xi}{1+\xi-\xi^{k+1}-\xi^{m-k+1}}\cdot\exp\left(-\frac{\varepsilon u p_{\min}}{4}\right).\label{eqn:j*|AB}
    \end{align}
    
Using the inequality in Eq.~\eqref{eqn:E inequality}, 
 \begin{align*}
 &\Pr\left(\Big|\theta^*_{(q)}-\theta_{(k)}\Big|>2u\right)\!=\!\Pr\left(\Big|\theta^*_{([qm])}-\theta_{(k)}\Big|>2u\right)
\!\!\leq\!\!\sum_{|\theta_{(j^*)} - \theta_{(k)}| > 2u}\!\!\!\!\!\!\mathbb{P}(j^*|A, B) + \mathbb{P}(D^c) + \mathbb{P}(A^c|B),
\end{align*}
and plugging in Eqns~\eqref{eqn:D^c}, \eqref{eqn:A|B}, and \eqref{eqn:j*|AB} to the RHS of the above inequality, we have
    \begin{align}
Pr\left(\Big|\theta^*_{(q)}-\theta_{(k)}\Big|>2u\right)\leq \;&\frac{U-L - 4u}{s}\cdot\frac{1-\xi}{1+\xi-\xi^{k+1}-\xi^{m-k+1}}\cdot\exp\left(-\frac{\varepsilon u p_{\min}}{4}\right) \label{eqn: term1}\\
        &+ \frac{2\eta}{u}\exp\left(-\frac{(m+1)u p_{\min}}{8}\right) + 2\exp\left(-\frac{m\eta^2p^2_{\min}}{12(1-q)}\right).\label{eqn: term23}
    \end{align}
\end{proof}

\begin{proposition}\label{prop:Pr}
    Under the conditions of Theorem \ref{thm:utility_PP}, the probability of PPquantile in Algorithm \ref{alg:DPP} selecting the correct index $[qm]$ is
    \begin{equation}\label{eqn:Prj*_upper}
        \Pr(j^* = [qm]) \leq \left(1 + \frac{(U\!-\!\theta_{(m)} + \theta_{(1)}\!-\!L) + s\cdot(m-2)}{\theta_{([qm]+1)} - \theta_{([qm])}}\cdot e^{-\varepsilon}\right)^{-1}.
    \end{equation}
\end{proposition}

\begin{proof}
The probability of selecting the correct index $[qm]$ in Algorithm \ref{alg:DPP} is $\Pr(j^* = [qm]) = (\theta_{([qm]+1)} - \theta_{([qm])})/C(m,\varepsilon)$, where
\begin{align}
    C(m,\varepsilon)= &\left(\theta_{(1)} - L\right)\exp\left(-\frac{\varepsilon\cdot[qm]}{2(m+1)}\right) + \left(U - \theta_{(m)}\right)\exp\left(-\frac{\varepsilon|m-[qm]|}{2(m+1)}\right) \label{eqn:Prj*_UL}\\
    & + \!\!\!\sum_{i\notin \{0, m, [qm]\}}\!\!\!\!\!\!\left(\theta_{(j+1)} - \theta_{(j)}\right)\exp(-\varepsilon|j-[qm]|/2(m+1))+ \left(\theta_{([qm]+1)} - \theta_{([qm])}\right)\label{eqn:Prj*_middle}\\
    \geq \; & (U\!-\!\theta_{(m)} + \theta_{(1)}\!-\!L)\cdot e^{-c_1\cdot \varepsilon} + s\cdot(m-2)\cdot e^{-c_2\cdot \varepsilon}+ \left(\theta_{([qm]+1)} - \theta_{([qm])}\right)
\end{align}
for $c_1, c_2 \in (0,1)$. Therefore, 
\begin{align}
    \Pr(j^* = [qm]) &\leq \left(1 + \frac{(U\!-\!\theta_{(m)} + \theta_{(1)}\!-\!L)}{\theta_{([qm]+1)} - \theta_{([qm])}}\cdot e^{-c_1\cdot \varepsilon} + \frac{s\cdot(m-2)}{\theta_{([qm]+1)} - \theta_{([qm])}}\cdot e^{-c_2\cdot \varepsilon}\right)^{-1}\\
    &\leq \left(1 + \frac{(U\!-\!\theta_{(m)} + \theta_{(1)}\!-\!L) + s\cdot(m-2)}{\theta_{([qm]+1)} - \theta_{([qm])}}\cdot e^{-\varepsilon}\right)^{-1}
\end{align}

$(L, U)$ need to be chosen carefully in practice. First, $(L, U)$ should cover the spread of the posterior samples so as not to bias the posterior distribution or clip the true posterior quantiles. On the other hand, 
Loose $(L,U)$  leads to large $U-\theta_{(m)} + \theta_{(1)}-L$  and small $\Pr(j^* = [qm])$, resulting in inaccurate estimation of $j^*$.
Given the randomness of posterior sampling, especially when $m$ is not large, $\theta_{(m)}, \theta_{(1)}, s$, and $\theta_{([qm]+1)} - \theta_{([qm])}$ can vary significantly across different sets of posterior samples, making a precise calibration of $(L, U)$ even more important, without compromising privacy.


\end{proof}

\section{Experiment details}
\subsection{Hyperparameters and code}\label{ape: hyper}

All the  PPIE methods require specification of the global bounds $(L_{\x}, U_{\x})$ for data $\x$ for the population mean \& variance case. We set $(L_{\x}\!=\!-4, U_{\x}\!=\!4)$ for $\x\sim\mathcal{N}(0,1)$ and  $(L_{\x}\!=\!0, U_{\x}\!=\!25)$ for $X\sim\mbox{Pois}(10)$ so that $\Pr (L_{\x} \!\leq\! x_i \!\leq\! U_{\x})\!\geq\!99.99\%$. For the Bernoulli case, $\x$ is 0 or 1 and thus naturally bounded.  The hyperparameters for the method to be compared with PRECISE in the simulation studies are listed below.
\begin{itemize}
\item SYMQ: The code is located \href{https://github.com/wxindu/dp-conf-int}{\textbf{here}}. We set the number of parametric bootstrap sample sets at 500.

\item PB:  The code is located 
 \href{https://github.com/ceciliaferrando/PB-DP-CIs}{\textbf{here}}. We set the number of parametric bootstrap sample sets at 500. We also identified and corrected a bug in the original code of OLS, where $\varepsilon$ is supposed to be split into 3 portions -- that is, 
\verb|np.random.laplace(0, Delta_w/eps/3, 1)| in the original code   should be replaced by \verb|np.random.laplace(0, Delta_w/(eps/3), 1)|.
\item  repro: The code is located \href{https://github.com/Zhanyu-Wang/Simulation-based_Finite-sample_Inference_for_Privatized_Data}{\textbf{here}}. We set the number of repro samples $R=200$. 
\item deconv:  The code is located \href{https://github.com/Zhanyu-Wang/Differentially_Private_Bootstrap}{\textbf{here}}: we used $B=\max\{2000\mu^2, 2000\}$, where $B$ is the number of bootstrap samples and $\mu$ is the privacy loss in $\mu$-GDP.
\item Aug.MCMC:  The code is located \href{https://github.com/nianqiaoju/dataaugmentation-mcmc-differentialprivacy}{\textbf{here}}. The prior for $\beta$ is $ \mathcal{N}_{p+1}(\mu, \tau^2I_{p+1})$, where $\mu\!=\!0.5, \tau\!=\!1$ for $n\!=\!100$ and $\mu\!=\!1, \tau\!=\!0.25$ for $n\!=\!1000$. We run 10,000 iterations per MCMC chain and with a  5,000 burn-in period.
\item \textbf{MS}: We set the number of multiple syntheses at 3.
\item \textbf{BLBquant}: BLBquant involve multiple hyperparameters. Readers may refer to the original paper for what each hyperparameter is. In terms of their values in our experiments, we set the multipliers $c=3$, $K=14$ to be more risk-averse as suggested by the authors. The other hyperparameters follow the settings in the original paper, specifically $R=50$, the number of Monte Carlo iterations for each little bootstrap $m_{\text{boot}} = \min\{10000, \max\{100, n^{1.5}/(s\log(n)\}\}$, the number of partitions of the dataset $s=\lfloor K\log(n)/\varepsilon_{(q)}\rfloor$, where $\varepsilon_{(q)} = 0.5\varepsilon$, where $\varepsilon$ is the total privacy loss, and the sequence of sets $I_t = [-tc/\sqrt{n}, tc/\sqrt{n}]$ for $t=1,2,...$
\end{itemize}

\subsection{Sensitivity of LS linear regression coefficients}
The MS implementation in the linear regression simulation study is based on sanitized $\hat{\bs\beta}=(\X'\X)^{-1}(\X'\y)$. To that end, we sanitize $(\X'\X)$ and $(\X'\y)$, respectively, the sensitivities of which are provided below. Let $\|\x_i\|_2=(\sum_{j=1}^{p-1}x^2_j)^{1/2}\le 1$ (no intercept) for any $p\ge1$ and $|y_i|\le C_Y$ for every $i=1,\ldots,n$. In the simulation study,  $C_Y = 4$. For the substitution neighboring relationship between datasets $D_1$ and $D_2$, WLOG, assuming the last data points $(\x'_{1n},y_{1n})$ and $(\x'_{2n},y_{2n})$ differ between $D_1$ and $D_2$, then
\begin{align*}
\Delta(\X'\y)&=\sup\|\sum_i\x'_{1i}y_{1i}-\sum_i\x'_{2i}y_{2i}\|_2
=\sup\|\x'_{1n}y_{1n}-\x'_{2n}y_{2n}\|_2\\
&\le 2\sup_{\x',y}\|\x'y\|_2\le2\sup_{\x',y}\|\x'\|_2\cdot\|y\|_2=2C_Y,
\end{align*}
where the first inequality holds due to the triangle inequality and the second is built upon the Cauchy-Schwarz inequality; and 

\begin{align*}
\Delta(\X'\X)&=\textstyle\sup\|\sum_i\x'_{1i}\x_{1i}-\sum_i\x'_{2i}\x_{2i}\|_F
=\sup\|\x'_{1n}\x_{1n}-\x'_{2n}\x_{2n}\|_F\\
&\leq2\sup_{\x',\x}\|\x'\x\|_F=\textstyle2\sup_{\x',\x}(1+2\sum_{j=1}^{p-1}x^2_j+2\sum_{j=1}^{p-1}x^2_jx^2_{j'} +\sum_{j=1}^{p-1}x^4_j)\\
& \mbox{since $\|\x_i\|_2^4=(\sum_{j=1}^{p-1}x^2_j)^2 \le 2\sum_{j=1}^{p-1}x^2_jx^2_{j'} +\sum_{j=1}^{p-1}x^4_j\le 1$, then}\\
\Delta(\X'\X)&=\textstyle2(1+2\sup_{\x',\x}(\sum_{j=1}^{p-1}x^2_j)+\sup_{\x',\x}(2\sum_{j=1}^{p-1}x^2_jx^2_{j'} +\sum_{j=1}^{p-1}x^4_j))\le 2(1+ 2 + 1)= 8.
\end{align*}
After the sensitivities are derived,  $\hat{\bs\beta}$ can be sanitized as in $(\x'\x+\mathbf{e}_x)$ and $(\x'\y+ \mathbf{e}_y)$, where  $\mathbf{e}_x$ and $\mathbf{e}_y$ are samples drawn independently from either a Laplace distribution or a Gaussian distribution, depending on the DP mechanism.

\section{Additional experimental results}\label{ape:exp_results}
This section presents results for $\mu$-GDP as a supplement to the $\varepsilon$-DP results shown in Figures \ref{fig:Normal} to \ref{fig:SLR} in Section \ref{sec:results}, the trends and the performances of methods are similar to those observed under $\varepsilon$-DP.

\begin{figure}[!htb]
\centering
  \includegraphics[width = 0.85\textwidth, trim={0.05in 0.1in 0.05in, 0in},clip]{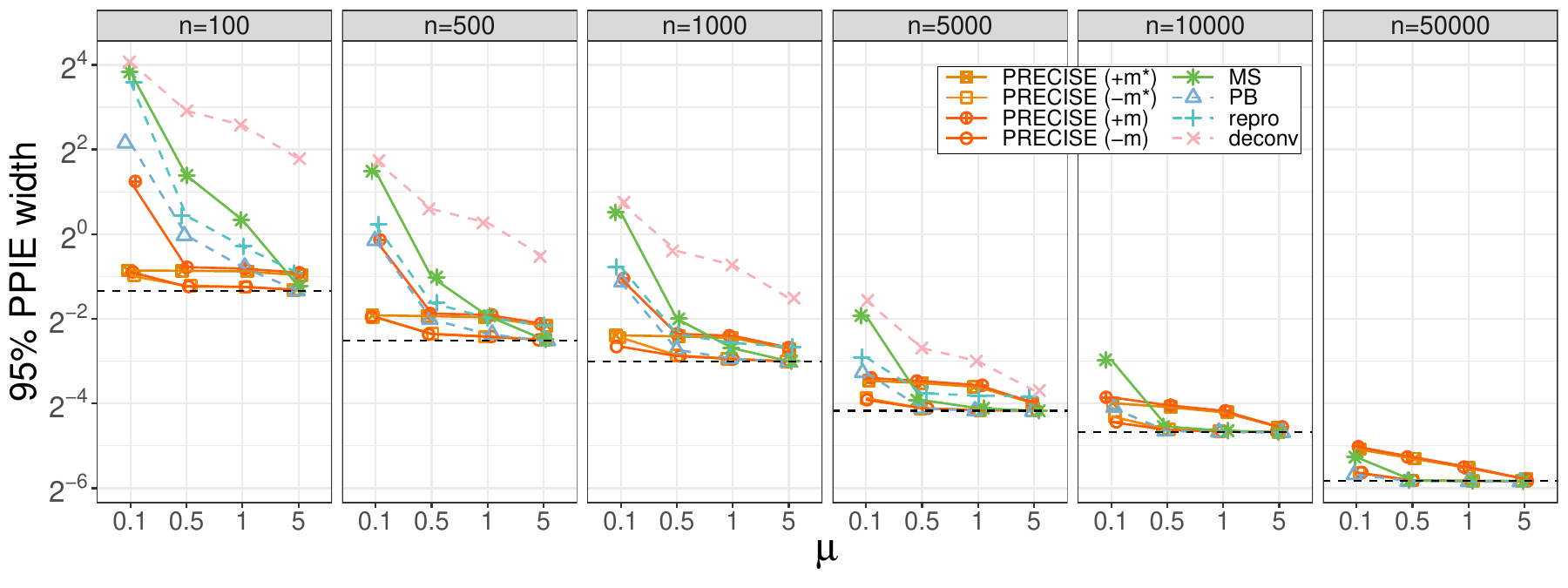}
  \includegraphics[width = 0.85\textwidth, trim={0.05in 0.1in 0.05in, 0.05in},clip]{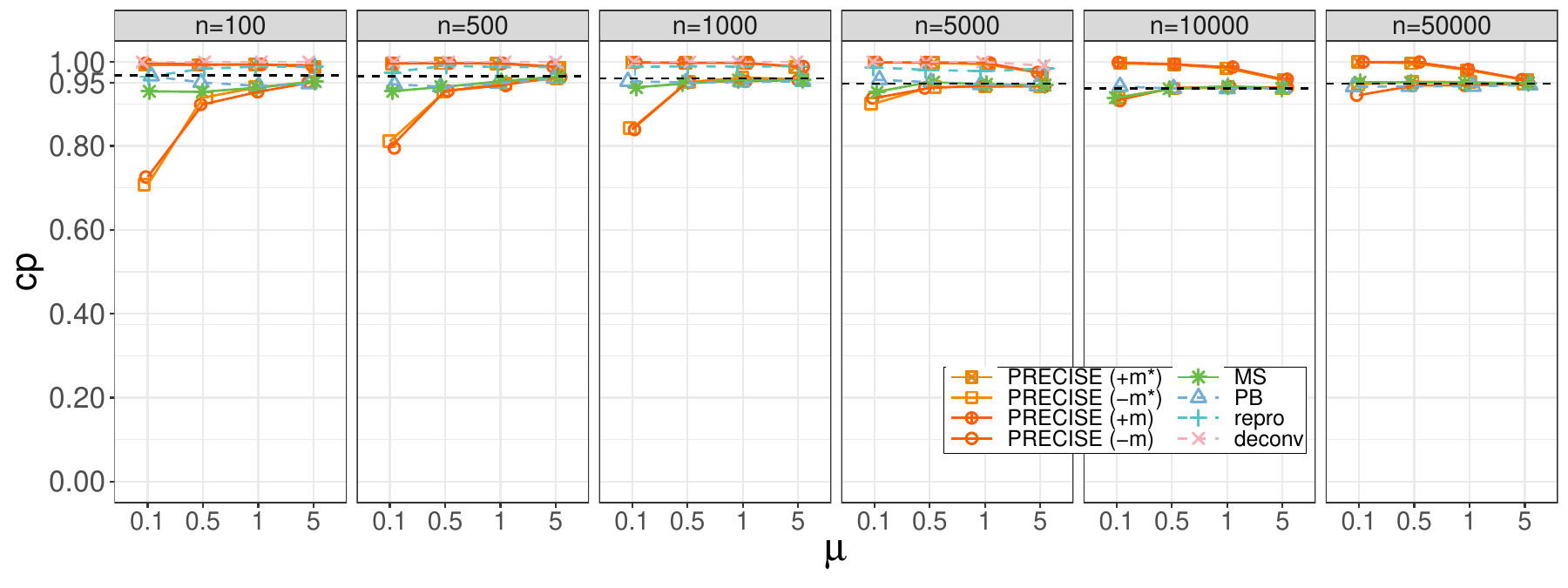}
  \vspace{-8pt}
\caption{PPIE width and CP for Gaussian mean ($\mu$-GDP).}
\end{figure}


\begin{figure}[!htb]
\centering
  \includegraphics[width = 0.85\textwidth, trim={0.05in 0.1in 0.05in, 0in},clip]{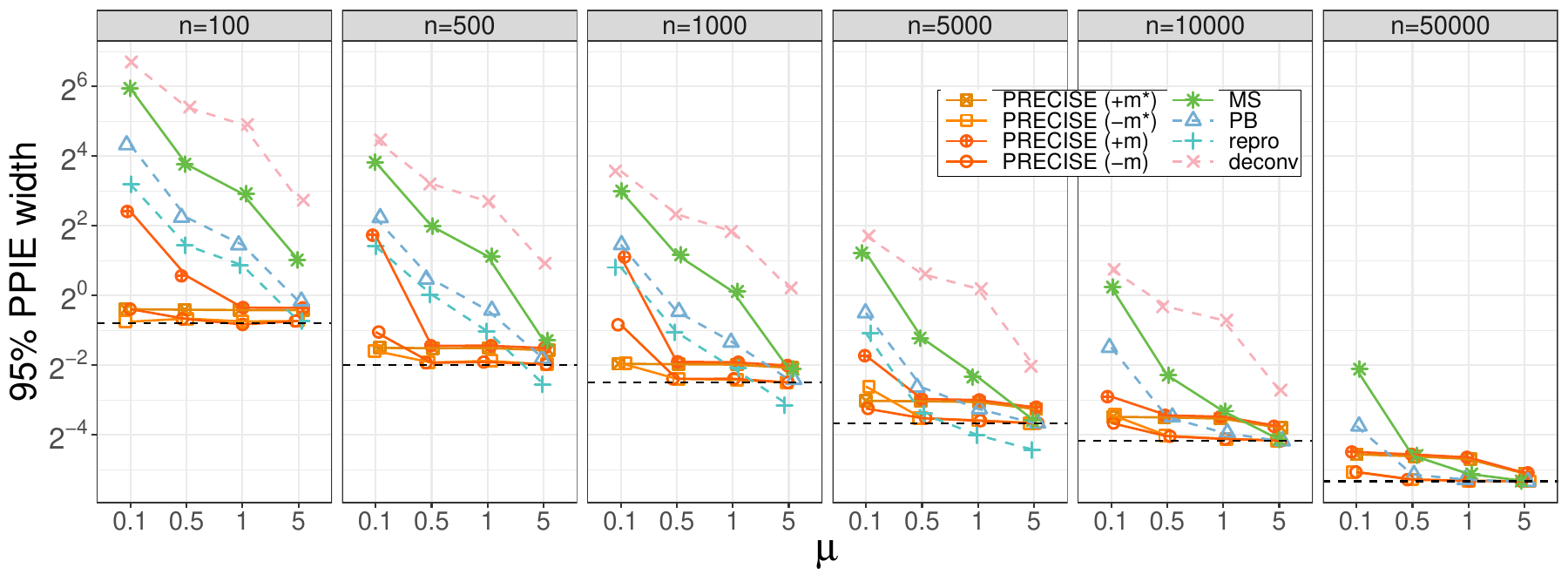}
  \includegraphics[width = 0.85\textwidth, trim={0.05in 0.1in 0.05in, 0in},clip]{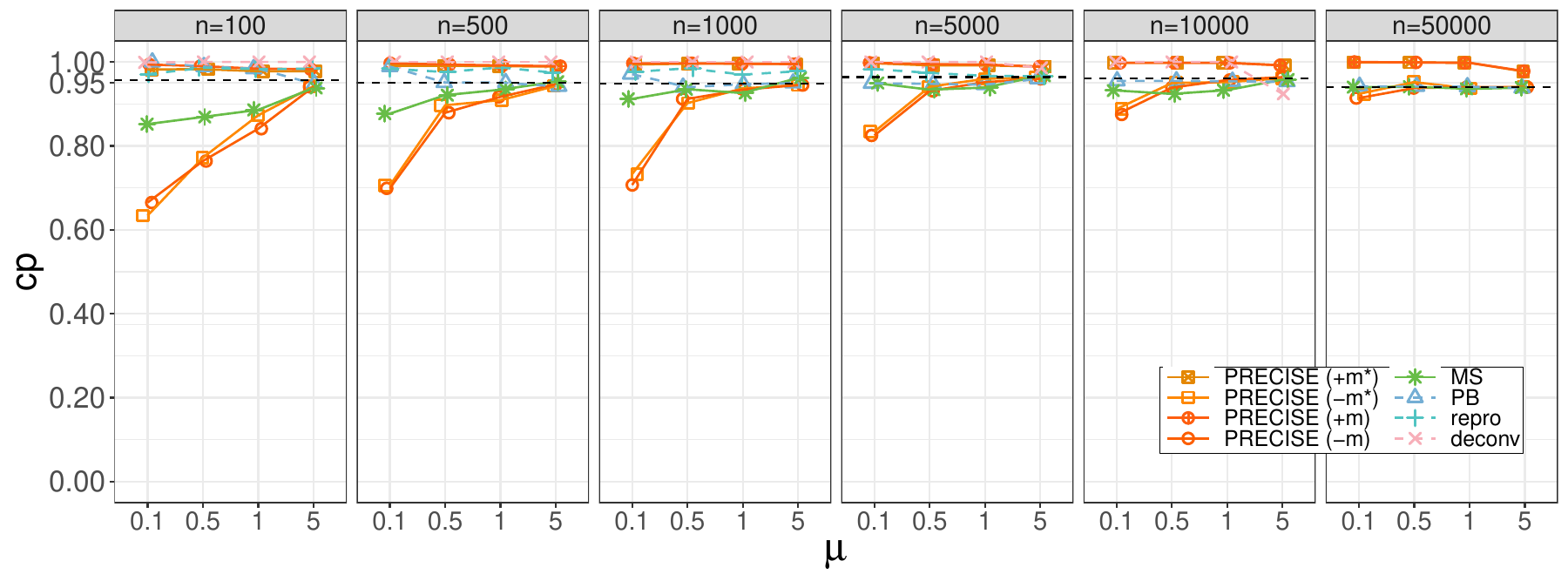}
  \vspace{-8pt}
  \caption{PPIE width and CP for Gaussian variance ($\mu$-GDP).}
\end{figure}


\begin{figure}[!htbp]
\centering
  \includegraphics[width = 0.85\textwidth, trim={0.05in 0.1in 0.05in, 0in},clip]{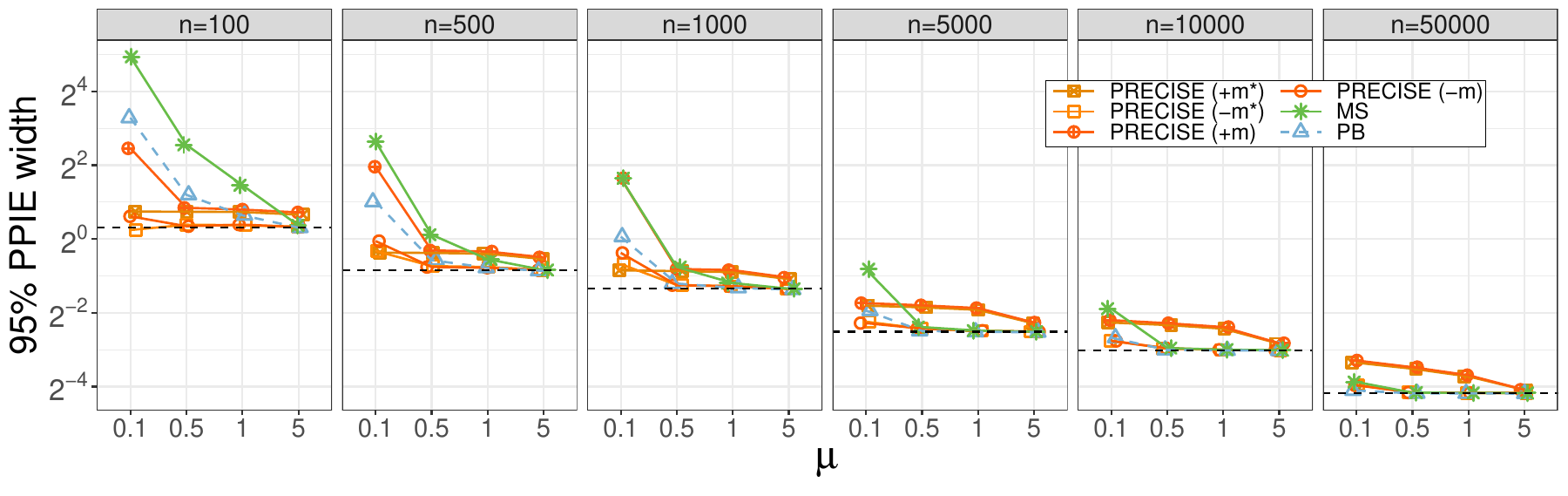}
  \includegraphics[width = 0.85\textwidth, trim={0.05in 0.1in 0.05in, 0in},clip]{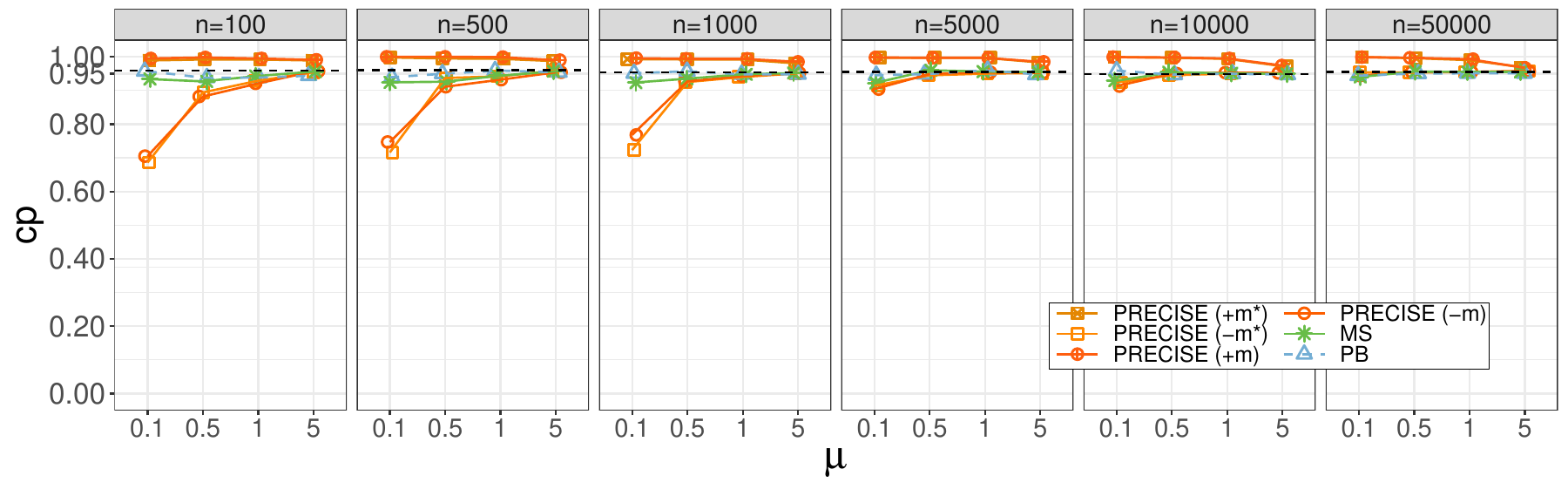}
  \vspace{-8pt}
  \caption{PPIE width and CP for Poisson mean ($\mu$-GDP).}
\end{figure}


\begin{figure}[!htb]
\centering
  \includegraphics[width = 0.85\textwidth, trim={0.05in 0.1in 0.05in, 0in},clip]{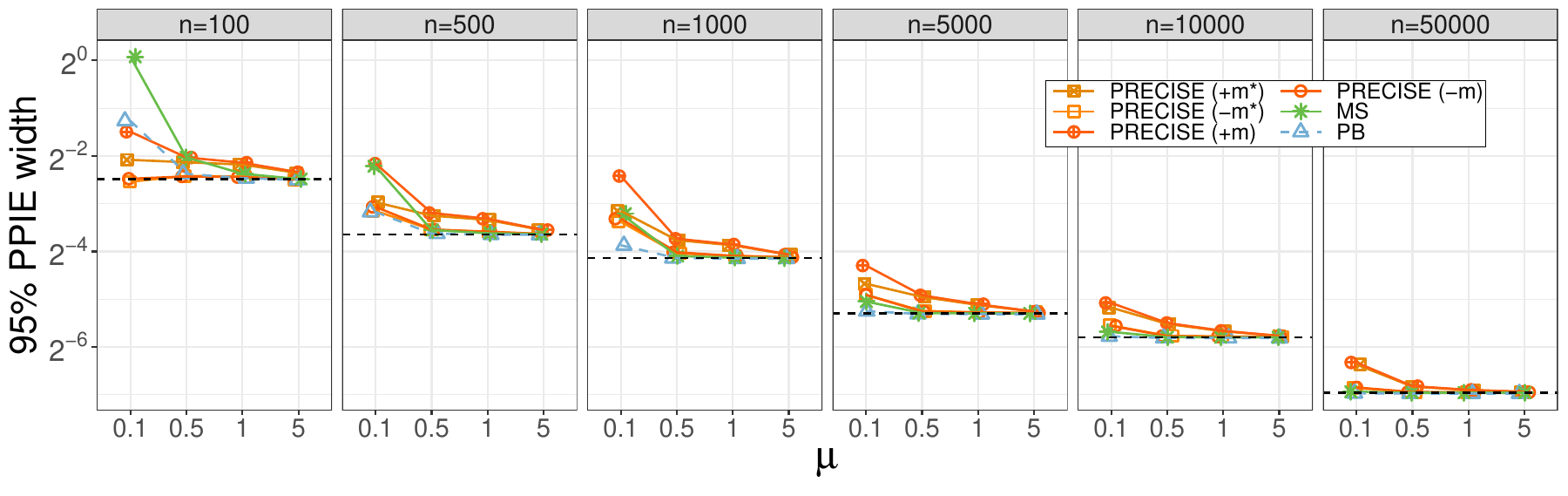}
  \includegraphics[width = 0.85\textwidth, trim={0.05in 0.1in 0.05in, 0in},clip]{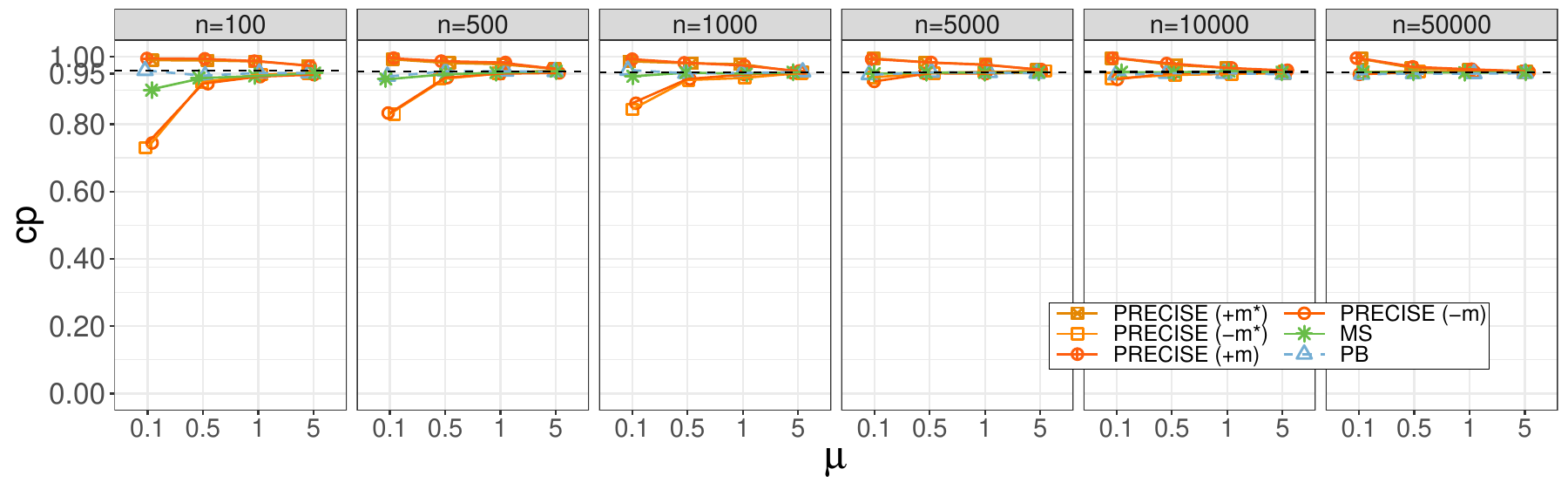}
  \vspace{-8pt}
  \caption{PPIE width and CP for Bernoulli proportion ($\mu$-GDP).}
\end{figure}

\begin{figure}[!htbp]
  \centering
  \includegraphics[width = 0.9\textwidth, trim={0.05in 0.14in 0.07in, 0.06in},clip]{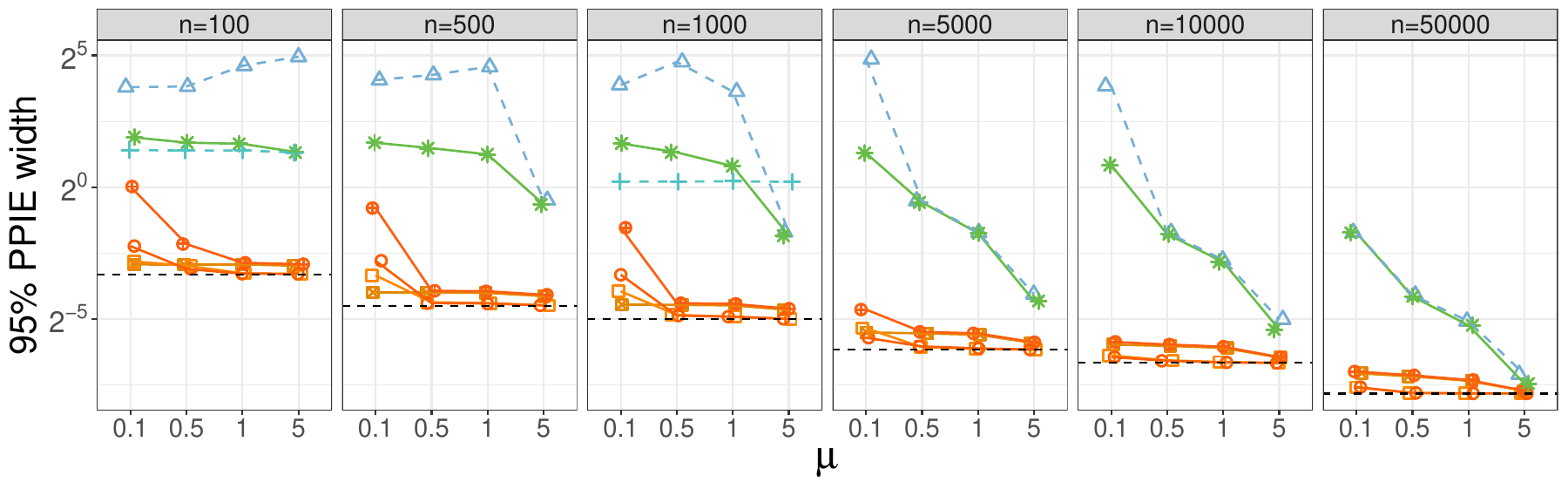}
  \includegraphics[width = 0.9\textwidth,trim={0.1in 0.14in 0.07in, 0.06in},clip]{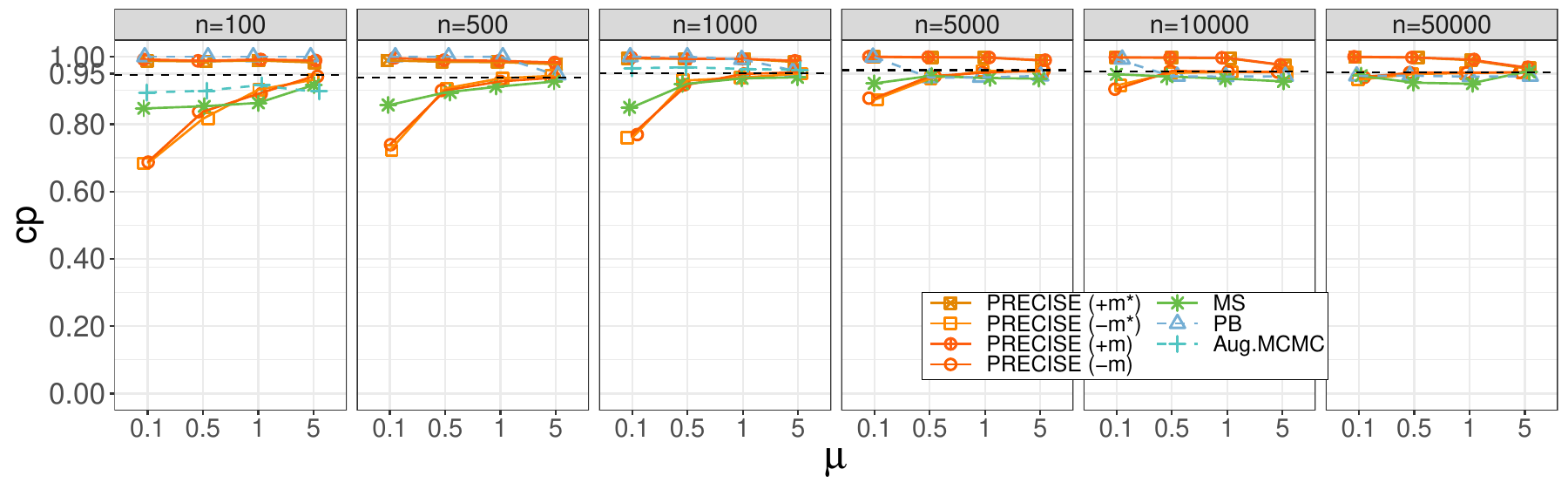}
\caption{PPIE width and CP for the slope in linear regression ($\mu$-GDP). }
\label{fig:SLR} 
\end{figure}

\end{document}